\documentclass[a4paper,11pt]{article}

\usepackage{mathrsfs} 

\usepackage{amsthm}
\usepackage{hyperref}
\usepackage{amssymb}
\usepackage{latexsym}
\usepackage{amsmath}
\usepackage{amsfonts}
\usepackage{mathabx}
\usepackage[dvips]{graphicx}
\usepackage[utf8]{inputenc}
\usepackage{cancel}
\usepackage{subcaption}

\usepackage[english]{babel}
\usepackage{url}
\usepackage{enumerate}
\usepackage{hyperref}
\usepackage{color}
\usepackage{babelbib}
\selectbiblanguage{english}

\usepackage[normalem]{ulem}

\usepackage{fancyhdr}
\newtheorem{Theorem}{Theorem}[part]

\newtheorem{Proposition}{Proposition}[part]

\newtheorem{Assumption}{Assumption}[part]
\newtheorem{Lemma}{Lemma}[part]
\newtheorem{Corollary}{Corollary}[part]
\newtheorem{Remark}{Remark}[part]
\newtheorem{Example}{Example}[part]

\makeatletter \@addtoreset{equation}{section}

\@addtoreset{Definition}{section}

\@addtoreset{Theorem}{section}

\@addtoreset{Proposition}{section}

\@addtoreset{Property}{section}

\@addtoreset{Assumption}{section}

\@addtoreset{Corollary}{section}

\@addtoreset{Lemma}{section}

\@addtoreset{Remark}{section}

\@addtoreset{Example}{section}

\DeclareMathOperator{\sgn}{sgn}

\newcommand{\cA}{\mathcal{A}}
\newcommand{\cB}{\mathcal{B}}
\newcommand{\cC}{\mathcal{C}}
\newcommand{\cD}{\mathcal{D}}

\newcommand{\cF}{\mathcal{F}}
\newcommand{\cG}{\mathcal{G}}
\newcommand{\cH}{\mathcal{H}}
\newcommand{\cI}{\mathcal{I}}

\newcommand{\cL}{\mathcal{L}}
\newcommand{\cM}{\mathcal{M}}

\newcommand{\cR}{\mathcal{R}}
\newcommand{\cS}{\mathcal{S}}
\newcommand{\cT}{\mathcal{T}}
\newcommand{\cU}{\mathcal{U}}
\newcommand{\cV}{\mathcal{V}}

\newcommand{\cY}{\mathcal{Y}}
\newcommand{\cZ}{\mathcal{Z}}

\renewcommand{\P}{\mathbb{P}}
\renewcommand{\H}{\mathbb{H}}

\newcommand{\R}{\mathbb{R}}

\newcommand{\Z}{\mathbb{Z}}

\def \eproof{\hbox{ }\hfill$\Box$}

\newcommand{\ud}{\, \mathrm{d}}

\newcommand{\HYP}[1]
    {\ensuremath{({H#1} ) }}

\newcommand{\1}{{\bf 1}}
    
\newcommand{\set}[1]
    {\ensuremath{\{ #1 \}}}
    
\newcommand{\HP}[1] 
    {\ensuremath{\mathscr{H}^{#1}}}

\newcommand{\esp}[1]{\ensuremath{\mathbb{E} \! \left[#1\right] }}

\DeclareMathOperator{\diag}{diag}

\DeclareMathOperator{\Tr}{Tr}

\newcommand{\we}{{}^{\epsilon\!}w}
\newcommand{\ce}{{}^{\epsilon\!}c}

\title{A numerical scheme for the quantile hedging problem}

\author{Cyril B\'en\'ezet\thanks{UFR de Math{\'e}matiques \& LPSM, Universit{\'e} Paris Diderot, B{\^a}timent Sophie Germain, 8 place Aur{\'e}lie Nemours, 75013 Paris, France
({\tt benezet@lpsm.paris, chassagneux@lpsm.paris})},
Jean-Fran{\c{c}}ois Chassagneux\footnotemark[1], Christoph Reisinger\thanks{Mathematical Institute, University of Oxford, Woodstock Road, Oxford, OX2 6GG, United Kingdom ({\tt christoph.reisinger@maths.ox.ac.uk})}
              }

\begin{document}
\maketitle

\begin{abstract}
We consider the numerical approximation of the quantile hedging price in a non-linear market. In a Markovian framework, we propose a numerical method based on a 
Piecewise Constant Policy Timestepping (PCPT) scheme coupled with a monotone finite difference approximation.
We prove the convergence of our algorithm combining BSDE arguments with the {Barles \& Jakobsen and Barles \& Souganidis} approaches for non-linear equations. In a numerical section, we illustrate the efficiency of our scheme by considering a financial example in a market with imperfections.
\end{abstract}

\vspace{5mm}


\vspace{5mm}



\section{Introduction}

In this work, we study the numerical approximation of the \emph{quantile hedging price} of a European contingent claim in a market with possibly some imperfections. The quantile hedging problem is a specific case of a broader class of approximate hedging problems. It consists in finding the minimal {initial endowment}
 of a portfolio that will allow the hedging a European claim with a given probability $p$ of success, the case $p=1$ corresponding to the classical problem of (super)replication. This approach has been made popular by the work of F\"ollmer and Leukert \cite{FL99} who provided a closed form solution in a special setting.
 
 The first PDE characterisation was introduced by \cite{bouchard2009stochastic} in a possibly incomplete market setting with portfolio constraints.
 Various extensions have been considered since this work: to jump dynamics \cite{moreau2011stochastic}; to the Bermudan case \cite{BBC16} and American case \cite{dumitrescu2017bsdes}; to a non-Markovian setting \cite{BER15, dumitrescu2016bsdes}; and to a finite number of quantile constraints \cite{bouchard2012stochastic}.
 
 Except for \cite{BBC16, bouchard2012stochastic}, all the aforementioned works are of a theoretical nature.
 The lack of established numerical methods for these problems is a clear motivation for our study. We now present in more detail the quantile hedging problem and the new numerical method we introduce and study in this paper.

On a complete probability space $\left(\Omega,  \cF, \P\right)$, we consider  a
$d$-dimensional Brownian motion  $(W_t)_{t \in [0,T]}$ and denote by $(\cF_t)_{t \in [0,T]}$ its natural filtration. We suppose that all the randomness comes from the Brownian motion and assume that $\cF = \cF_T$.

\noindent  Let $\mu:\R^d \to \R^d$, $\sigma:\R^d\to\cM_d(\R)$, where $\cM_d(\R)$ is the set of $d \times d$ matrices with real entries, $\textcolor{black}{f}:[0,T]\times\R^d
\times\R\times\R^d \to \R$ be Lipschitz continuous functions, with Lipschitz constant $L$. 
 
 \noindent For $(t,x,y) \in [0,T]\times \R^d \times \R$ and $\nu \in \H^2$, which denotes the set of
predictable square-integrable processes, we consider the solution
$(X^{t,x}, Y^{t,x,y,\nu})$ to the following stochastic differential
equations:
\begin{eqnarray*}
  X_s &=& x + \int_t^s \mu(X_u)\ud u + \int_t^s \sigma(X_u)\ud W_u,\\
  Y_s &=& y - \int_t^s f(u,X_u,Y_u,\nu_u)\ud u + \int_t^s \nu_u \ud W_u, \qquad s \in [t,T].
\end{eqnarray*}
In the financial applications we are considering, $X$ will typically represent the log-price of risky assets, the control process $\nu$ is the amount invested in the risky assets, and the function $f$ is non-linear to allow to take into account some market imperfections in the model. 
A typical financial example, which will be investigated in the numerical section, is the following:
\begin{Example}
{The underlying diffusion $X$ is a one-dimensional Brownian motion with constant drift $\mu \in \R$ and volatility $\sigma > 0$. There is a constant borrowing rate $R$ and a lending rate $r$ with $R\ge r$. } In this situation, the function $f$ is given by:
  \begin{align*}
   f(t,x,y,z) = -r y - \sigma^{-1}\mu z + (R-r)(y- \sigma^{-1} z)^-.
  \end{align*}
\end{Example}

\noindent 
The quantile hedging problem corresponds to the following stochastic control problem:
{\color{black}  for $(t,x,p) \in [0,T] \times \textcolor{black}{\R^d} \times [0,1]$ find}
\begin{align} \label{eq:prob1} v(t,x,p) &:= \inf \left\{ y \ge 0:
    \exists \nu \in \H_2,\; \P\left(Y^{t,x,y,\nu}_T \ge
      g(X^{t,x}_T)\right) \ge p \right\}.
\end{align}

\noindent
{The main objective of this paper is to design a numerical procedure to approximate the function $v$ by discretizing an associated non-linear PDE
first derived in \cite{bouchard2009stochastic}. 
A key point in the derivation} of this PDE is to observe that the above problem can be reformulated as a classical stochastic target problem by introducing a new control process representing the conditional probability of success. To this end, for $\alpha \in \H^2$, we denote
\begin{align*}
P^{t,p,\alpha}_s := p + \int_t^s \alpha_s \ud W_s\;, \; t \le s \le T\;,
\end{align*}
and by $\cA^{t,p}$ the set of $\alpha$ such that $P^{t,p,\alpha}_\cdot \in [0,1] $. The problem \eqref{eq:prob1} can be rewritten as
\begin{align*}
v(t,x,p) &:= \inf \left\{ y \ge 0:
    \exists (\nu,\alpha) \in (\H_2)^2,\; Y^{t,x,y,\nu}_T \ge
      g(X^{t,x}_T)\1_{\set{P^{t,p,\alpha}_T>0}}  \right\}
\end{align*}
(see Proposition 3.1 in \cite{bouchard2009stochastic} for details).
In our framework, the above singular stochastic control problem admits a representation in terms of a non-linear expectation,
{generated by} a Backward Stochastic Differential Equation (BSDE), 
\begin{align}\label{eq proba rep v}
v(t,x,p) = \inf_{\alpha \in \cA^{t,p}} \cY_t^\alpha
\end{align}
where $(\cY^\alpha,\cZ^\alpha)$ is the solution to
\begin{align*}
\cY^\alpha_s = g(X^{t,x}_T)\1_{\set{P^{t,p,\alpha}_T>0}} + \int_{s}^T f(s,X_s,\cY^\alpha_s,\cZ^\alpha_s) \ud s - \int_s^T \cZ^\alpha_s \ud W_s\;, \; t \le s \le T\;.
\end{align*}
The article \cite{BER15} justifies the previous representation and proves a dynamic programming {principle} for the  control problem in a general setting. In the Markovian setting, this would lead naturally to the following PDE for $v$ in $[0,T)\times\R^d\times(0,1)$:
\begin{align}\label{eq natural pde}
``- \partial_t \varphi + \sup_{a \in \R^d} F^a(t,x,\varphi,D\varphi,D^2\varphi) = 0 \," 
\end{align}
 where for  $ (t,x,y)\in [0,T]\times\R^d\times\R^{+}$,
$q:=\left(\begin{matrix}q^{x}\\q^p\end{matrix}\right)\in\R^{d+1}$ and $A:=\left(\begin{matrix}A^{xx}&A^{xp}\\A^{xp^\top}&A^{pp}\end{matrix}\right)\in\mathbb{S}^{d+1}$, $A^{xx} \in \mathbb{S}^{d}$, denoting $\Xi :=(t,x,y,q,A)$, we define
\begin{align}
F^a(\Xi) &:= -f(t,x,y,\mathfrak{z}(x,q,a)) - \cL(x,q,A,a)\;, \label{eq de Fa}
\end{align}
with
\begin{align}
\mathfrak{z}(x,q,a) &:= q^x \sigma(x) + q^p a\;, \label{eq de za}
\\
\cL(x,q,A,a) &:= \mu(x)^\top q^x+ \frac12 \Tr\left[\sigma(x)\sigma(x)^\top A^{xx} \right] + \frac{|a|^2}{2}A^{pp} + a^\top \sigma(x)^\top A^{xp}. \label{eq de La}
\end{align}
The PDE formulation in \eqref{eq natural pde} is not entirely correct as the supremum part may degenerate and it would require using semi-limit relaxation to be mathematically rigourous. We refer to \cite{bouchard2009stochastic}, where it has been obtain in a more general context. 
We shall use an alternative PDE formulation to this ``natural'' one \eqref{eq natural pde}, which we give at the start of Section \ref{se-pcpt}.

\noindent Moreover, the value function $v$ {continuously}
satisfies the following boundary conditions in the $p$-variable:
\begin{align}
  v(t,x,0) &= 0 \;\text{ and }\;
  v(t,x,1) = V(t,x)\mbox{ on }[0,T] \times (0,\infty)^d, \label{eq bc in p}
  %
\end{align}
where $V$ is the super-replication price of the contingent
claim with payoff $g(\cdot)$. 

\noindent It is also known that $v$ has a discontinuity as $t \to T$. By definition, the terminal condition is
\begin{align*}
\R^d\times[0,1] \ni (x,p) \mapsto g(x)1_{p > 0} \in \R^+\;,
\end{align*}
but  {the values which are continuously attained are}
obtained by convexification \cite{bouchard2009stochastic}, namely 
\begin{align}
  v(T-,x,p) = pg(x) \mbox{ on } \R^d \times [0,1], \label{eq bc at T}
\end{align}
and we shall work with this terminal condition at $t=T$ from now on.


\vspace{4pt}
\noindent 
To design the numerical scheme to  approximate $v$, we use the following strategy:
\begin{enumerate}
\item Bound and discretise the set where the controls $\alpha$ take their values.
\item Consider an associated Piecewise Constant Policy Timestepping (PCPT) scheme for the {control processes
}.
\item Use a monotone finite difference scheme to approximate in time and space the PCPT {solution} resulting from 1.\ \& 2.
\end{enumerate} 

\noindent
{
The approximation of controlled diffusion processes by ones where policies are piecewise constant in time was first analysed by \cite{Kry99};
in \cite{K00}, this procedure is used in conjunction with Markov chain approximations to diffusion processes to
construct fully discrete approximation schemes to the associated Bellman equations and to derive their convergence order.
An improvement to the order of convergence from \cite{Kry99} was shown recently in \cite{jakobsen2019improved} using a refinement of Krylov's original, probabilistic techniques.

\noindent
Using purely viscosity solution arguments for PDEs, error bounds for such approximations are derived in \cite{BJ07}, which are weaker than those in \cite{Kry99} for the control approximation scheme, but improve the bounds in \cite{K00} for the fully discrete scheme.
In \cite{RF16}, using a switching system approximation introduced in \cite{BJ07}, convergence is proven for a generalised scheme where linear PDEs are solved piecewise in time on different meshes, and the control optimisation is carried out at the end of time intervals using possibly non-monotone, higher order interpolations.
An extension of the analysis in \cite{RF16} to jump-processes and non-linear expectations is given in \cite{DRZ18}.
}

\noindent Our first contribution is to prove that the approximations built in step 1.\ and 2.\ above are convergent {for the quantile hedging problem, which has substantial new difficulties compared to the settings considered in the aforementioned works}.
For this we rely heavily on the comparison theorem for the formulation in \eqref{eq rep pde} and we take advantage of the monotonicity property of the approximating sequences. The main new difficulties come from the non-linear form of the PDE {including unbounded controls}, and in particular the boundaries in the $p$-variable. {To deal with the latter especially, we} rely on some fine estimates for BSDEs to prove the consistency of the  {scheme including the strong boundary conditions} (see Lemma \ref{le perturb xi} and Lemma \ref{le cons 2}).

\noindent Our second contribution is to design the monotone scheme in step 3.\ and to prove its convergence. The main difficulties come here from the non-linearity of the {new term from the driver of the BSDE in the gradient} combined with
the degeneracy of the diffusion operator given in \eqref{eq de La}, and again the boundedness for the domain in $p$. In particular, a careful analysis of the consistency of the boundary condition is needed (see Proposition \ref{prop bound cond num}). 

\noindent To the best of our knowledge, this is the first numerical method for the quantile hedging problem in this non-linear market specification. In the linear market setting, using a dual approach, \cite{BBC16} combines the solution of a linear PDE with Fenchel-Legendre transforms to tackle the problem of Bermudan quantile hedging. Their approach cannot be directly  adapted here due to the presence of the non-linearity. The dual approach in the non-linear setting would impose some  convexity assumption on the $f$ parameter and would require to solve fully non-linear PDEs. Note that here $f$ is only required to be Lipschitz continuous in $(y,z)$.
We believe that an interesting alternative to our method would be to extend the work of \cite{BBMZ09} to the non-linear market setting we consider here.


\vspace{4mm}
The rest of the paper is organised as follows. In Section \ref{se-pcpt}, we derive the control approximation and PCPT scheme associated with items $1.$ \& $2.$ above and prove their convergence. In Section \ref{se-monotone}, we present a monotone finite difference approximation which is shown to convergs to the semi-discrete PCPT scheme. In Section \ref{se-numerics}, we present numerical results {for a specific application and analyse the
observed convergence}. Finally, the appendix {contains some of the longer, more technical proofs and} collects useful background results used in the paper.

\paragraph{Notations} 
\noindent
$\diag(x)$ is the diagonal matrix of size $d$, whose diagonal is given by $x$.

\noindent
Let us denote by $\mathcal{S}$ the sphere in $\R^{d+1}$ of radius $1$ and by $\cD$ the set of vectors $\eta \in \cS$  such that their first component $\eta^1 = 0$. 
For a vector $\eta \in \cS \setminus \cD$, we denote $\eta^\flat:=\frac1{\eta^1} (\eta^2,\dots,\eta^{d+1})^\top \in \R^d$.
By extension, we denote, for $\cZ \subset \cS \setminus \cD$,
$
\cZ^\flat := \set{ \eta^\flat \in \R^d \,|\, \eta \in \cZ}\;.
$

\vspace{4pt}
\noindent 
We denote by
$\cB\cC := L^\infty([0,T],\cC^0(\R^d\times[0,1]))$, namely the space of functions $u$ that are \textcolor{black}{ essentially bounded in time and continuous with respect to their space variable}.
  The convergence in $\cC^0([0,T]\times\R^d)$ considered here is the local uniform convergence.



\section{Convergence of a discrete-time scheme}
\label{se-pcpt}

In this section, we design a Piecewise Constant Policy Timestepping  (PCPT) scheme which is convergent to the value function $v$ defined in \ref{eq proba rep v}.


Following \cite{BBMZ09}, it has been shown in \cite{BC17}, that the function  $v$ is  equivalently a viscosity solution of the following PDE
(see Theorems 3.1 and 3.2 in \cite{BC17}):
\begin{align}\label{eq rep pde}
  \cH(t,x,\varphi,\partial_t \varphi,D \varphi,D^{2} \varphi) = 0
\end{align}
in $(0,T)\times\R^d\times(0,1)$, where $\cH$ is a continuous operator 
 \begin{align}\label{eq def H}
 \cH(\Theta) = \sup_{\eta \in \cS} H^\eta(\Theta)\;,
 \end{align}
where for  $ (t,x,y,b)\in [0,T]\times\R^d\times\R^{+}\times\R$,
$q:=\left(\begin{matrix}q^{x}\\q^p\end{matrix}\right)\in\R^{d+1}$ and $A:=\left(\begin{matrix}A^{xx}&A^{xp}\\A^{xp^\top}&A^{pp}\end{matrix}\right)\in \mathbb{S}^{d+1}$,
and $\Theta := (t,x,y,b,q,A)$, we define
\begin{align}
H^\eta(\Theta) = (\eta^1)^2\left( -b \textcolor{black}{-} f(t,x,y,\mathfrak{z}(x,q,\eta^\flat)) - \cL(x,q,A,\eta^\flat) \right)\;,\;\text{ for } \eta \in \cS\setminus \cD\;.
\end{align}
Recall also the definition of $\cL$ and $\mathfrak{z}$ in \eqref{eq de za} and \eqref{eq de La}.  

This representation and its properties  are key in the proof of convergence. Loosely speaking, it is obtained by ``compactifying'' the set $\set{1}\times \R^d$ to the unit sphere $\cS$. A comparison theorem is shown in Theorem 3.2 in \cite{BC17}.

\vspace{4pt}
\noindent As partially stated in the introduction, we will work under the following assumption: 
\begin{enumerate}
\item[\HYP{}(i)] The functions $b$, $\sigma$ are $L$-Lipschitz continuous and $g$ is bounded and $L_g$-Lipschitz continuous.
\item[(ii)] The function $f$ is measurable and for all $t\in[0,T]$, $f(t,\cdot,\cdot,\cdot)$ is $L$-Lipschitz continuous. For all $(t,x,z) \in[0,T]\times\R^d\times\R^d$, the function $y \mapsto f(t,x,y,z)$ is decreasing. Moreover,
\begin{align}\label{eq f00=0}
f(t,x,0,0) = 0.
\end{align}

\end{enumerate}
\noindent Under the above Lipschitz continuity assumption,  the mapping
$$ \cS\setminus \cD \ni \eta \mapsto H^\eta(\Theta) \in \R$$
extends continuously to $\cS$ by setting, for all $\eta \in \cD$,
\begin{align*}
H^\eta(\Theta) = -\frac12A^{pp}\;,
\end{align*}
see {Remark 3.1 in  \cite{BC17}}.

\begin{Remark}
(i) In \HYP{}(ii), the monotonicity assumption is not a restriction, as in a Lipschitz framework, the classical transformation $\tilde{v}(t,x,p) := e^{\lambda t}v(t,x,p)$ for $\lambda$ large enough allows to reach this setting; see Remark 3.3 in \cite{BC17} for details.
\\
(ii) The condition \eqref{eq f00=0} is a reasonable financial modelling assumption: It says that starting out in the market with zero initial wealth and making no investments will lead to a zero value of the wealth process.
\\
(iii)   Since  $f$ is decreasing and $g$ is bounded, it is easy to see that $|V|_\infty \le |g|_\infty$, where $V$ is the super-replication price.
\end{Remark}


\subsection{Discrete set of control}
\label{sub-se disc control}
In order to introduce a discrete-time scheme which approximates the solution $v$ of 
\eqref{eq rep pde}--\eqref{eq bc at T}, 
we start by discretizing the set of controls $\cS$.

\noindent Let $({\cR}_n)_{n\ge 1}$ be an increasing sequence of closed subsets of $\cS\setminus\cD$ such that
\begin{align}\label{eq density Sn}
\overline{\bigcup_{n\ge 1} \cR_n} = \cS\;. 
\end{align}
%
%
%
For $n \ge 1$, let $v_n: [0,T]\times\textcolor{black}{\R^d}\times[0,1] \to \R$ be
the unique continuous viscosity solution of the following PDE:
\begin{align}\label{pdeN-bis}
 \cH_n(t,x,\varphi,\partial_t\varphi,D\varphi,D^2\varphi)  = 0,
\end{align}
satisfying the boundary conditions
\eqref{eq bc in p}-\eqref{eq bc at T}, see Corollary \ref{co comp princ}. Above, the operator $\cH_n$ is naturally given by
\begin{align}
 \cH_n(\Theta) := \sup_{\eta \in \cR_n} H^\eta(\Theta)\;.
\end{align}

\begin{Proposition}\label{pr conv bounding}
The functions $v_n$ converge to $v$ in $\cC^0([0,T]\times\R^d)$.
\end{Proposition}

\begin{proof}
1. For $n' < n$, we observe that $v_{n'}$ is a super-solution of 
\eqref{pdeN-bis}
as $\cR_{n'} \subset \cR_n$. Using the comparison result of Proposition \ref{pr appendix comp princ}, we
obtain that $v_{n'}\ge v_n$. Similarly, using the comparison principle (\cite{BC17}, Theorem 3.2), we obtain that
$v_n \ge v$, for all $n\ge 1$.
\\
For all $(t,x,p) \in [0,T]\times(0,\infty)^d\times[0,1]$, let:
  \begin{align}
    \overline v(t,x,p) &= \lim_{j \to \infty} \sup \left\{v_n(s,y,q) : n \ge j \mbox{ and } \|(s,y,q)-(t,x,p)\| \le \frac 1 j \right\}, \\
    \underline v(t,x,p) &= \lim_{j \to \infty} \inf \left\{v_n(s,y,q) : n \ge j \mbox{ and } \|(s,y,q)-(t,x,p)\| \le \frac 1 j \right\}.
  \end{align} 
 From the above discussion, recalling that $v_1$ and $v$ are continuous, we have
 \begin{align*}
 v_1 \ge \overline{v} \ge \underline{v} \ge v\,,
 \end{align*}
 which shows that $\overline{v}$ and  $\underline{v}$ satisfy the boundary conditions \eqref{eq bc in p}-\eqref{eq bc at T}.
  
\noindent In order to prove the theorem, it is enough to show that
$\overline v$ is a viscosity subsolution
 of \eqref{eq rep pde} {and $\underline v$ is a viscosity supersolution (which follows similarly and is therefore omitted)}.
 The comparison principle (\cite{BC17}, Theorem 3.2) then implies that
$v = \overline v = \underline v$, and it follows from \cite{crandall1992user}, Remark 6.4 that the convergence 
$v_n \to v$ as ${n \to \infty}$ is uniform on every compact set.
Using  Theorem 6.2 in \cite{achdou2013hamilton}, we obtain that $\overline{v}$ is a subsolution to
\begin{align*}
\underline{\cH}(t,x,\varphi,\partial_t \varphi, D\varphi,D^2\varphi) = 0 \text{ on } (0,T)\times(0,\infty)^d\times(0,1)\;,
\end{align*}
where 
\begin{align*}
\underline{\cH}(\Theta)=\lim_{j \to \infty} \inf \left\{\cH_n(\Theta') : n \ge j \mbox{ and } \|\Theta-\Theta'\| \le \frac 1 j \right\}.
\end{align*}
In the next step, we prove that 
$\underline{\cH} = \cH$, which concludes the proof of the proposition.
\\
2. Let us denote by $\mathfrak{P}_n:\cS \rightarrow \tilde{\cS}_n$ the closest neighbour projection on the closed set  $\tilde{\cS}_n$. From \eqref{eq density Sn}, we have that $\lim_{n\rightarrow \infty}\mathfrak{P}_n(\eta)=\eta$, for all $\eta \in \cS$.
We also have that
\begin{align*}
\cH(\Theta) = H^{\eta^*}(\Theta)
\end{align*}
{for some $\eta^* \in {\rm argmax}_{\eta \in \cS} H^{\eta}(\Theta)$}
as $\cS$ is compact. Let us now introduce $\eta_n := \mathfrak{P}_n(\eta^*)$ and by continuity of $H$, we have
\begin{align*}
H^{\eta_n}(\Theta) \rightarrow \cH(\Theta)\;.
\end{align*}
We also observe that
\begin{align*}
H^{\eta_n}(\Theta) \le \cH_n(\Theta) \le \cH(\Theta)\;. 
\end{align*}
This proves the convergence $\cH_n(\Theta) \uparrow \cH(\Theta)$, for all $\Theta$.
As $\cH$ is continuous, we conclude by using Dini's Theorem that the convergence is uniform on compact subsets, leading to
$\underline{\cH}=\cH$.
\eproof
\end{proof}

\subsection{The PCPT scheme} 
\label{sub-se pcpt schme}
From now on, we fix $n\ge 1$ and  $\cR_n$ the associated discrete set of control. 
 For  $ (t,x,y)\in [0,T]\times\R^d\times\R^{+}$,
$q\in\R^{d+1}$ and $A\in\mathbb{S}^{d+1}$, denoting $\Xi :=(t,x,y,q,A)$, we define
\begin{align*}
\cF_n(\Xi) = \sup_{a \in \cR_n^\flat} F^a(\Xi) \text{ with }  F^a(\Xi) := -f(t,x,y,\mathfrak{z}(x,q,a)) - \cL(x,q,A,a)\;.
\end{align*}

\noindent Following the proof of Corollary \ref{co comp princ}
in the appendix, we easily observe that $v_n$ is also the unique viscosity solution to
\begin{align}
-\partial_t \varphi + \cF_n(t,x,\varphi, D\varphi,D^2 \varphi) = 0 \text{ on } [0,T)\times \R^d \times (0,1) 
\label{eq pde more classic}
\end{align}
with the same boundary conditions \eqref{eq bc in p}--\eqref{eq bc at T}. 
\textcolor{black}{The above PDE is written in a more classical way and we will mainly consider this form in the sequel.}
Let us observe in particular that $K := \cR_n^\flat$ is a discrete subset of $\R^d$, such that \eqref{eq pde more classic} appears as a natural discretisation of  \eqref{eq natural pde} and will be simpler to study. 


\vspace{2mm}
\noindent To approximate ${v}_n$, we  consider an adaptation of the PCPT scheme in  \cite{K00,BJ07}, and especially \cite{DRZ18}, to our setting, as described below.

\noindent
For $\kappa \in \mathbb{N}^*$, we consider  grids of the time interval $[0,T]$:
$$\pi = \{0 =: t_0< \dots<t_k<\dots< t_\kappa := T\},$$
and denote $|\pi|:=\max_{0\le k\le \kappa} (t_{k+1}-t_k)$.

\vspace{4pt}
\noindent For $0 \le t < s \le T$, $a \in K$ and
a continuous $\phi: \R^d \times [0,1] \to \R$, we denote by $S^a(s,t,\phi):[t,s]\times\R^d\times[0,1] \rightarrow \R$ the
unique solution of 
    \begin{align}
      \label{PDEaconst} -\partial_t \varphi  + F^a(t,x,\varphi,D \varphi, D^2 \varphi) = 0 \text{ on } [t,s) \times \R^d \times (0,1), \\
      \label{TermCond} \varphi(s,x,p) =  \phi(x,p) \quad &\mbox{ on } \R^d \times [0,1], \\
      \label{bCond1} \varphi(r,x,0) = B^0(t,s,\phi)(r,x), \quad \varphi(r,x,1) =B^1(t,s,\phi)(r,x)\quad &\mbox{ on } [t,s) \times \R^d.
    \end{align}
 
 \noindent The function  ${B}^p(t,s,\phi)$ for $p \in \set{0,1}$ is solution to
 \begin{align}\label{eq pde zero control}
 -\partial_t \varphi + F^0(r,x,\varphi, D\varphi,D^2 \varphi) = 0 \text{ on } [t,s)\times \R^d \;,
 \end{align}
 with terminal condition  ${B}^p(t,s,\phi)(r,x)(s,x)=\phi(x,p)$. 

\noindent The solution to the PCPT scheme associated with the grid $\pi$ is then the
function $v_{n,\pi}:[0,T]\times \R^d\times[0,1]$ such that
\begin{align}
  \mathfrak{S}(\pi, t, x, p, v_{n,\pi}(t,x,p), v_{n,\pi}) = 0,
\end{align}
where for a grid $\pi$, $(t,x,p,y)\in [0,T]\times\R^d\times[0,1]\times\R^+$ and a function $u \in \cB\cC$,
\begin{align}
\label{ThScheme} 
\mathfrak{S}(\pi,t,x,p,y,u) = 
  \left\{
  \begin{array}{rl}
  y - \min_{a \in K} S^a\left(t_\pi^+,t_\pi^-, u(t_\pi^+, \cdot)\right)(t,x,p) & \mbox { if } t < T, \\
 y - \hat{g}(x)p & \mbox{ otherwise},
  \end{array} \right. 
\end{align}
with 
\begin{align}\label{eq notation grid}
t_\pi^+ := \inf \{r \in \pi \, | \, r > t\} \;\text{ and }
\;t_\pi^- := \sup \{r \in \pi \,|\, r \le t\}.
\end{align}
We will drop the subscript $\pi$ for brevity whenever we consider a fixed mesh.

\vspace{4pt}
\noindent Let us observe that the function $v_{n,\pi}$ can be alternatively described by the following backward
algorithm:
\begin{enumerate}
\item Initialisation: set $v_{n,\pi}(T,x,p) := {g}(x)p$,
  $x \in \R^d \times [0,1]$.
\item Backward step: 
  For $k = \kappa-1,\dots,0$, compute $w^{k,a} := S^a(t_k,t_{k+1},v_{n,\pi}(t_{k+1},\cdot))$ and set
  \begin{align}\label{eq de w-pi-backward}
  v_{n,\pi}(\cdot) := \inf_{a \in K}w^{k,a}\;.  
  \end{align}
  
\end{enumerate}


\begin{Remark} In our setting, we can easily identify the boundary values (of the scheme):
\begin{enumerate}[(i)]
\item At $p=0$, the terminal condition is $\phi(T,x)=0$ (recall that $v(T,x,p)=g(x)\1_{\set{p>0}}$), and this propagates through the backward iteration, so that $v_{n,\pi}(t,x,0)=0$ for all $(t,x)\in [0,T]\times\R^d$.
\item At $p=1$, the terminal condition is  $\phi(T,x)=g(x)$ and the boundary condition is thus given by $v_{n,\pi}(t,x,0)=V(t,x)$ for all $(t,x)\in [0,T]\times\R^d$, where $V$ is the super-replication price.
\end{enumerate}
\end{Remark}

\noindent The main result of this section is the following.

\begin{Theorem}\label{th conv disc time}
  The function $v_{n,\pi}$ converges to $v_n$ in $\cC^0$ as
  $|\pi| \rightarrow 0$.
\end{Theorem}
\begin{proof}
1. We first check the consistency with the boundary condition. Let $\hat{a} \in K$ and $\hat{w}$ be the (continuous) solution of
\begin{align}\label{eq pde one a for F}
-\partial_t \varphi + F^{\hat{a}}(t,x,\varphi,D\varphi, D^2\varphi) = 0 \text{ on } [0,T)\times \R^d \times(0,1)
\end{align}
with boundary condition $v(t,x,p)=pV(t,x)$ on $[0,T]\times\R^d\times \set{0,1} \bigcup \set{T}\times \R^d \times [0,1]$.
\\
By backward induction on $\pi$, one gets that
\begin{align} \label{eq majo sup 1}
v_{n,\pi} \le \hat{w}\;. 
\end{align}
Indeed, we have $v_{n,\pi}(T,\cdot) =\hat{w}(T,\cdot)$. Now if the inequality is true at time $t_k$, $k\ge1$, we have,
using the comparison result for \eqref{eq pde one a for F}, recalling Proposition \ref{pr appendix comp princ}, that
$$w^{k,\hat{a}}(t,\cdot) \le \hat{w}(t,\cdot) \text{ for } t \in [t_{k-1},t_k]\;,$$
and thus \emph{a fortiori} $\hat{w}(t,\cdot) \ge v_{n,\pi}(t,\cdot)$, for $t \in [t_{k-1},t_k]$.
\\
We also obtain that  
\begin{align}\label{eq mino inf 1}
v_{n,\pi}(\cdot) \ge v_n(\cdot)
\end{align}
by backward induction. Indeed, we have $v_{n,\pi}(T,\cdot) = v_n(T,\cdot)$. Assume that the inequality is true at time $t_k$, $k\ge1$. We observe that $w^{k,a}$ is a supersolution of  \eqref{pdeN-bis}, namely the PDE satisfied by $v_n$. By the comparison result, this implies that $w^{k,a}(t,\cdot) \ge v_n(t,\cdot)$, for $t \in [t_{k-1},t_k]$. Taking the infimum over $a \in K$ yields then \eqref{eq mino inf 1}.
\\
Since 
\begin{align}
v_n \le \underline{w} \le \overline{w} \le \hat{w}\,,
\end{align}
where 
\begin{align*}
\overline{w}(t,x,p) &= \limsup_{(t',x',p',|\pi|) \rightarrow (t,x,p,0)} v_{n,\pi}(t',x',p')
\text{ and }
\underline{w}  = \liminf_{(t',x',p',|\pi|) \rightarrow (t,x,p,0)} v_{n,\pi}(t',x',p'),
\end{align*}
we obtain that $\underline{w}$ and $\overline{w}$ satisfy the boundary conditions \eqref{eq bc in p}--\eqref{eq bc at T}.  

\vspace{4pt}

\noindent 2.   We prove below that the scheme is
  \emph{monotone}, \emph{stable} and \emph{consistent}, see Proposition \ref{pr monotonicity}, Proposition \ref{pr stability} and Proposition \ref{cons-theo} respectively. Combining this with step 1.\ and Theorem 2.1 in
  \cite{BS91} then ensures the convergence in $\cC^0$ of $v_{n,\pi}$ to $v_n$ as
  $|\pi| \to 0$.  \eproof 
\end{proof}

\begin{Remark}
We prove the following properties by a combination of viscosity solution arguments and, mostly, BSDE arguments, where they appear more natural.
It should be possible to derive these results purely using PDE arguments using similar main steps as in \cite{BJ07} .
\end{Remark}

\begin{Proposition}[Monotonicity]\label{pr monotonicity}
  Let
  $u \ge v$ for $u,v \in \cB\cC$, $(t,x,p) \in [0,T]\times\R^d\times[0,1]$, $y \in
  \R$. We have:
  \begin{align}
    \mathfrak{S}(\pi,t,x,p,y,u) \le \mathfrak{S}(\pi,t,x,p,y,v).
  \end{align}
\end{Proposition}

\begin{proof}
  Let $t < T, x \in \R^d, p \in [0,1]$. By definition of $v_{n,\pi}$, recalling \eqref{eq de w-pi-backward}, it is sufficient to prove that,
  for any $a \in K$, we have:
  \begin{align}
    S^a(t^+, t^-, u(t^+, \cdot))(t,x,p) \ge S^a(t^+,t^-, v(t^+, \cdot))(t,x,p).
  \end{align}
with $t^+, t_-$ defined in \eqref{eq notation grid}\;.
  But this follows directly from the comparison result given in Proposition \ref{pr appendix comp princ}.
\eproof
\end{proof}

\vspace{4pt}
\noindent We now study the stability of the scheme. We first show that
the solution of the scheme $v_{n,\pi}$ is increasing in its third
variable. {\color{black} This is not only an interesting property in its own right
which the piecewise constant policy solution inherits from the solution to
the original problem (\ref{eq:prob1}), but it also}
allows us to obtain easily a uniform bound for $v_{n,\pi}$,
namely the boundary condition at $p=1$.
\begin{Lemma} \label{lem inc p}
The scheme \eqref{ThScheme} has the
  property, for all $t \in [0,T]$ and $x \in \R^d$:
  \begin{align}
    v_{n,\pi}(t,x,q) \le v_{n,\pi}(t,x,p) \mbox{ if } 0\le q \le p \le 1.
  \end{align}
\end{Lemma}

\begin{proof}
  We are going
  to prove the assertion by induction on $k \in \{0, \dots, \kappa\}$.
  \\
  For $t = T = t_\kappa$ and every $x \in \R^d$, we have
  $(x,p)\mapsto v_{n,\pi}(T,x,p) := g(x)p$, which is an increasing function of $p$.
  \\
  Let $1\le k < \kappa-1$. Assume now that $v_{n,\pi}(t,x,\cdot)$ is an increasing function for all
  $t \ge t_{k+1}$ and $x \in \R^d$. We show that $v_{n,\pi}(t,x,\cdot)$ is
  also increasing for $t \in [t_k, t_{k+1})$ and $x \in \R^d$.
  \\
  Let $0 \le q \le p \le 1$. By the definition of $v_{n,\pi}$ in \eqref{eq de w-pi-backward}, 
  it is sufficient to show that for each $a \in K$, we
  have, for $(t,x)\in[0,T]\times\R^d$,
$$ w^{k,a}(t,x,q) \le w^{k,a}(t,x,p)\;.$$
  \\
 From Lemma \ref{le prob rep wa}(i) in the appendix, these two quantities admit a
  probabilistic representation  with two different random terminal times
  \begin{align}
    \tau^q &= \inf \{s \ge t : P^{t,q,a}_s \in \{0,1\}\} \wedge t_{k+1}, \\
    \tau^p &= \inf \{s \ge t : P^{t,p,a}_s \in \{0,1\}\} \wedge t_{k+1}.
  \end{align}
  However, using Lemma \ref{le prob rep wa}(ii), we can write probabilistic representations with BSDEs with
  terminal time $t_{k+1}$: we have that
  $S^a(t_{k+1},t_k,w_\pi(t_{k+1},\cdot))(t,x,p) = \tilde
  Y^{t,x,p,a}_t$, where $\tilde Y^{t,x,p,a}_t$ is the first component
  of the solution of the following BSDE:
  \begin{align}
    Y_s = \textcolor{black}{v_{n,\pi}(t_{k+1}, X^{t,x}_{t_{k+1}}, \tilde P^{t,p,a}_{t_{k+1}})} + \int_s^{t_{k+1}} f(u, X^{t,x}_u, Y_u,  Z_u) \ud u - \int_s^{t_{k+1}} Z_u \ud W_u,
  \end{align}
  where $\tilde P^{t,p,a}$ is the process defined by:
  \begin{align}
    \tilde P^{t,p,a}_s = p + \int_t^s a 1_{\set{u \le \tau^p}} \ud W_u,
  \end{align} 
  and a similar representation holds for
  $S^a(t_{k+1},t_k,w_\pi(t_{k+1},\cdot))(t,x,q)$.
  \\
  It remains to show that
  \begin{align}
  \label{whatineq}
   v_{n,\pi}(t_{k+1},X^{t,x}_{t_{k+1}},\tilde P^{t,p,a}_{t_{k+1}})
  \ge v_{n,\pi}(t_{k+1},X^{t,x}_{t_{k+1}},\tilde
  P^{t,q,a}_{t_{k+1}}).
  \end{align}
  If this is true, the classical comparison theorem for BSDEs (see e.g.\ Theorem 2.2 in \cite{el1997backward}),
  concludes the proof.\\
  First, we observe that $P^{t,p,a}_{\tau_p} \ge P^{t,q,a}_{\tau_p}$.
  On $\set{\tau_p = T}$,  \eqref{whatineq} holds straightforwardly
  by the  induction hypothesis.
  On $\set{\tau_p < T}$, if $P^{t,p,a}_{\tau_p} =1$ then $P^{t,p,a}_{T} =1$ and
   \eqref{whatineq} holds 
  by induction hypothesis, as $P^{t,q,a}_{T} \le1$; 
  if $P^{t,p,a}_{\tau_p} =0$ then \emph{a fortiori}
$P^{t,q,a}_{\tau_p} =0$ and $P^{t,p,a}_{T} =P^{t,q,a}_{T} =0$, which concludes the proof.
\eproof
\end{proof}

\begin{Proposition}[Stability]  \label{pr stability}
  The solution to scheme \eqref{ThScheme} is bounded. 
\end{Proposition}

\begin{proof}
  For any $\pi$ and any $(t,x,p) \in [0,T]\times\R^d\times[0,1]$, we
  have $v_{n,\pi}(t,x,p) \le v_{n,\pi}(t,x,1) =  V(t,x)$.  \eproof
\end{proof}

\vspace{4pt}
\noindent To prove the consistency of the scheme, we will need the two following lemmata.

\begin{Lemma}\label{le perturb xi} For $0 \le \tau \le t \le \theta \le T$, $\xi \in \R$, and $\phi \in \cC^{\infty}([0,T]\times \R^d \times [0,1])$, the following holds
\begin{align*}
|S^a(\tau,\theta,\phi(\theta,\cdot)+\xi)(t,\cdot) - S^a(\tau,\theta,\phi(\theta,\cdot))(t,\cdot) - \xi |_\infty \le C|\theta-t| |\xi| \,.
\end{align*}
\end{Lemma}
\begin{proof} We denote $w=S^a(\tau,\theta,\phi(\cdot))$ and $\tilde{w}=S^a(\tau,\theta,\phi(\cdot)+\xi)$. Using Lemma \ref{le prob rep wa}, we have that, for
$(t,x,p) \in [\tau,\theta]\times\R^d\times[0,1]$,
\begin{align*}
w(t,x,p) = Y_t  \; \text{ and } \; \tilde{w}(t,x,p) = \hat{Y}_t 
\end{align*}
where $(Y,Z)$ and $(\hat{Y},\hat{Z})$ are solutions to, respectively,
\begin{align*}
Y_r &= \phi(X^{t,x}_\theta,\tilde{P}^{t,p,a}_\theta) + \int_r^T f(s,X^{t,x}_s, Y_s, Z_s) \ud s - \int_r^\theta Z_s \ud W_s\;,\; t\le r \le \theta\,,
\\
\hat{Y}_r &= \phi(X^{t,x}_\theta,\tilde{P}^{t,p,a}_\theta) + \xi + \int_r^T f(s,X^{t,x}_s, \hat{Y}_s, \hat{Z}_s) \ud s - \int_r^\theta \hat{Z}_s \ud W_s\;,\; t\le r \le \theta\,.
\end{align*}
Denoting $\Gamma := Y + \xi$ and $f_\xi(t,x,y,z)=f(t,x,y-\xi,z)$, one observes then that $(\Gamma,Z)$ is the solution to
\begin{align*}
\Gamma_r &= \phi(X^{t,x}_\theta,\tilde{P}^{t,p,a}_\theta) + \xi + \int_r^T f_\xi(s,X^{t,x}_s, \Gamma_s, {Z}_s) \ud s - \int_r^\theta {Z}_s \ud W_s\;,\; t\le r \le \theta\,.
\end{align*}
Let $\Delta := \Gamma - \hat{Y}$, $\delta Z = Z-\hat{Z}$ and
$ \delta f_s = f_\xi(s,X^{t,x}_s, \Gamma_s, {Z}_s) -  f(s,X^{t,x}_s, \Gamma_s, {Z}_s)$, for $s \in [t,\theta]$. We then get
\begin{align*}
\Delta_r := \int_r^\theta\left( f(s,X^{t,x}_s, \Gamma_s, {Z}_s)- f(s,X^{t,x}_s, Y_s, {Z}_s) + \delta f_s\right) \ud s - \int_r^\theta \delta Z_s \ud W_s\;.
\end{align*}
Classical energy estimates for BSDEs \cite{el1997backward,chassagneux2018obliquely} lead to
\begin{align}\label{eq nrj estim bsde}
\esp{\sup_{r \in [t,\theta]}|\Delta_r|^2} \le C\esp{\int_t^\theta |\Delta_s \delta f_s| \ud s}\;.
\end{align}
Next, we compute
\begin{align*}
\int_t^\theta |\Delta_s \delta f_s| \ud s \le \frac1{2C}\sup_{s \in [t,\theta] } |\Delta_s|^2 + 2C\left(\int_t^\theta |\delta f_s| \ud s\right)^2\;.
\end{align*}
Combining the previous inequality with \eqref{eq nrj estim bsde}, we obtain
\begin{align*}
\esp{\sup_{r \in [t,\theta]}|\Delta_r|^2} \le 4C^2\esp{\left(\int_t^\theta |\delta f_s| \right)^2}\;.
\end{align*}
Using the Lipschitz property of $f$, we get from the definition of $f_\xi$,
\begin{align*}
|\delta f_s| \le L \xi \,,
\end{align*}
which eventually leads to
\begin{align}
\esp{\sup_{r \in [t,\theta]}|\Delta_r|^2} \le C|\theta-t|^2 \xi^2\;
\end{align}
and concludes the proof. \eproof
\end{proof}

\begin{Lemma}\label{le cons 2}Let $0\le \tau < \theta \le T$ and $\phi \in \cC^{\infty}([0,T]\times \R^d \times [0,1])$.
For $(t,x,p)\in [\tau,\theta)\times \R^d \times (0,1)$,
\begin{align*}
\phi(t,x,p) - S^a(\tau,\theta,\phi(\theta,\cdot))(t,x,p)- 
(\theta-t)G^a\phi(t,x,p)
= o(\theta-t)\;.
\end{align*}
where $G^a\phi(t,x,p) := -\partial_t \phi(t,x,p) +F^a(t,x,p,\phi,D\phi,D^2\phi)$.
\end{Lemma}

\begin{proof}
We first observe that $S^a(\tau,\theta,\phi(\cdot))(t,x,p) = {Y}_t$, where $({Y}^a,{Z}^a)$ is solution to 
\begin{align*}
{Y}_r &= {\Phi}_\theta+ \int_r^\theta f(s,X^{t,x}_s, {Y}_s, {Z}_s) \ud s - \int_r^\theta {Z}_s \ud W_s
\end{align*}
with, for $t \le s \le \theta$,
\begin{align*}
\Phi_s = \phi(s,X^{t,x}_s,P_s^{t,p,\alpha})\, \text{ and } \alpha := a \1_{[0,\tau]} \,.
\end{align*}
By a direct application of Ito's formula, we observe that 
\begin{align*}
\Phi_r = \Phi_\theta - \int_r^\theta\set{\partial_t \phi + \cL^\alpha\phi}(s,X^{t,x}_s,P_s^{t,p,\alpha})\ud s - \int_r^\theta \mathfrak{Z}_s \ud W_s\,,\quad t \le r  \le \theta\,,
\end{align*}
where $\mathfrak{Z}_s:=\mathfrak{z}(X^{t,x}_s,D\phi(s,X^{t,x}_s,P_s^{t,p,a}),\alpha_s)$, $t\le s \le \theta$.
\\
For ease of exposition, we also introduce an ``intermediary"  process $(\hat{Y},\hat{Z})$ as the solution to
\begin{align*}
\hat{Y}_r = \Phi_\theta + \int_r^\theta f(s,X^{t,x}_s,\Phi_s,\mathfrak{Z}_s)\ud s 
- \int_r^\theta  \hat{Z}_s \ud W_s\,,\quad t \le r  \le \theta\,.
\end{align*}
Now, we compute
\begin{align*}
&\hat{Y}_t - \Phi_t +(\theta-t)G^a\phi(t,x,p) 
\\
&=
\\
&\esp{\int_t^\theta \left(\set{\partial_t \phi(s,X^{t,x}_s,P^{t,p,a}_{s}) - \partial_t \phi(t,x,p)}  + 
\set{F^a\phi(s,X^{t,x}_s,P^{t,p,a}_{s}) - F^a\phi(t,x,p)}
\right) \ud s}\,.
\end{align*}
Using the smoothness of $\phi$, the Lipschitz property of $f$ and the following control
\begin{align}
\esp{|X^{t,x}_s - x| + |P^{t,p,\alpha}_s-p|} \le C_a|\theta-t|^\frac12 \,,
\end{align}
we obtain
\begin{align}\label{eq cons temp 1}
|\hat{Y}_t - \Phi_t +(\theta-t)G^a\phi(t,x,p)|
\le
C_{a,\phi}(\theta-t)^\frac32\;.
\end{align}
We also have
\begin{align*}
\hat{Y}_r - \Phi_r = \int_r^\theta G^a\phi(s,X^{t,x}_s,P^{t,p,\alpha}_s) \ud s -\int_r^\theta (\hat{Z}_s-\mathfrak{Z}_s)\ud W_s
\end{align*}
Applying classical energy estimates for BSDEs,
we obtain
\begin{align}
\esp{\sup_{r \in [t,\theta]}|\hat{Y}_r - \Phi_r |^2
+\int_t^\theta |\hat{Z}_s-\mathfrak{Z}_s|^2 \ud s}
&\le C\esp{\left(\int_t^\theta |G^a\phi(s,X^{t,x}_s,P^{t,p,\alpha}_s| \ud s\right)^2 } \nonumber
\\
&\le C_{a,\phi}(\theta-t)^2\;, \label{eq useful cons 1}
\end{align}
where for the last inequality we used the smoothness of $\phi$ and the linear growth  of $f$ and $\sigma$.\\
We also observe that
\begin{align*}
\hat{Y}_r - Y_r = \int_r^\theta \set{\delta f_s + f(s,X^{t,x}_s,\hat{Y}_s,\hat{Z}_s)-
f(s,X^{t,x}_s,{Y}_s,{Z}_s)}\ud s - \int_r^\theta \set{\hat{Z}_s-Z_s} \ud W_s\,,
\end{align*}
where $\delta f_s := f(s,X^{t,x}_s,\Phi_s,\mathfrak{Z}_s) - f(s,X_s,\hat{Y}_s,\hat{Z}_s)$, for $t \le s \le \theta$. Once again, from classical energy estimates
 \cite{el1997backward,chassagneux2018obliquely}, we obtain
\begin{align*}
|\hat{Y}_t - Y_t|^2 & \le C \esp{\left(\int_t^\theta \delta f_s \ud s\right)^2}\;.
\end{align*}
Using the Cauchy-Schwarz inequality and the Lipschitz property of $f$,
\begin{align*}
|\hat{Y}_t - Y_t|^2&\le C(\theta-t)\esp{\sup_{r \in [t,\theta]}|\hat{Y}_r - \Phi_r |^2
+\int_t^\theta |\hat{Z}_s-\mathfrak{Z}_s|^2 \ud s}\;.
\end{align*}
This last inequality, combined with \eqref{eq useful cons 1}, leads to
\begin{align*}
|\hat{Y}_t - Y_t| \le C(\theta-t)^\frac32 \;.
\end{align*}
The proof is concluded by combining the above inequality with \eqref{eq cons temp 1}.
\eproof
\end{proof}

\vspace{4pt}

\noindent Finally, we can prove the following consistency property.  
  \begin{Proposition}[Consistency] \label{cons-theo} 
  Let $\phi \in \cC^{\infty}([0,T]\times \R^d \times [0,1])$.
  For $(t,x,p)\in [0,T)\times\R^d\times(0,1)$,
  \begin{align}
  \left|\frac1{t_\pi^+ - t} \mathfrak{S}\big(\pi,t,x,p,\phi(t,x,p)+\xi,\phi(\cdot)+\xi\big)
  +\partial_t \phi -\cF_n(t,x,p,\phi,D\phi,D^2 \phi)\right| \rightarrow 0
  \end{align}
  as $(|\pi|,\xi) \rightarrow 0$.
  \end{Proposition}
  \begin{proof}
  We first observe that by Lemma \ref{le perturb xi}, it is sufficient to prove
  \begin{align*}
    \left|\frac1{t_\pi^+ - t} \mathfrak{S}\big(\pi,t,x,p,\phi(t,x,p),\phi(\cdot)\big)
  +\partial_t \phi -\cF_n(t,x,p,\phi,D\phi,D^2 \phi)\right| \underset{|\pi| \downarrow 0}{ \longrightarrow} 0
  \end{align*}
  We have that
  \begin{align*}
  &\left|\frac1{t_\pi^+ - t} \mathfrak{S}\big(\pi,t,x,p,\phi(t,x,p),\phi(\cdot)\big)
  +\partial_t \phi -\cF_n(t,x,p,\phi,D\phi,D^2 \phi)\right|
  \\
  &=\left|
  \frac1{t_\pi^+ - t} \set{ \phi(t,x,p) - \min_{a \in K} S^a\left(t^-_\pi, t^+_\pi, \phi(t^+_\pi, \cdot)\right)(t,x,p) 
 }  - \max_{a \in K} G^a(t,x,p)  \phi \right|
   \\
   &\le \max_{a \in K} \left|
    \frac1{t_\pi^+ - t} \set{ \phi(t,x,p) - S^a\left(t^-_\pi, t^+_\pi, \phi(t^+_\pi, \cdot)\right)(t,x,p) }
   - G^a \phi \right|\,.
  \end{align*}
  The proof is then concluded by applying Lemma \ref{le cons 2}.
  \eproof
  \end{proof}

\vspace{4pt}
\noindent To conclude this section, let us observe that we obtain the following result, combining Proposition \ref{pr conv bounding} and Theorem \ref{th conv disc time} .
\begin{Corollary}
In the setting of this section, assuming (H), the following holds
\begin{align*}
\lim_{n \rightarrow \infty}\lim_{|\pi| \downarrow 0} v_{n,\pi} = v \;.
\end{align*}
\end{Corollary}

\begin{Remark} An important question, from numerical perspective, is to understand how to fix the parameters $n$ and $\pi$ in relation to each other.
The theoretical difficulty here is to obtain a precise rate of convergence for the approximations given in Proposition \ref{pr conv bounding} and Theorem \ref{th conv disc time}, along the lines of the continuous dependence estimates with respect to control discretisation in \cite{jakobsen2005continuous, DRZ18},
and estimates of the approximation by piecewise constant controls as in \cite{Kry99, jakobsen2019improved}. 
To answer this question in our general setting is a challenging task, extending also to error estimates for the full discretisation in the next section, which is left for further research.
\end{Remark}

\section{
Application to the Black-Scholes model: a fully discrete monotone scheme} 
  \label{se-monotone}

The goal of this section is to introduce a fully implementable scheme and to prove its convergence. The scheme is obtained by adding a finite difference approximation to the PCPT procedure described in Section \ref{sub-se pcpt schme}. Then in Section \ref{se-numerics}, we present numerical tests that demonstrate the practical feasibility of our numerical method.
From now on, we will assume that the log-price process $X$ is a one-dimensional Brownian motion with drift, for
$(t,x) \in [0,T]\times\R$:
\begin{align} \label{eq x bm drifted}
  X^{t,x}_s &= x + \mu(s-t) + \sigma(W_s - W_t), \qquad s \in [t,T],
\end{align}
with $\mu \in \R$ and $\sigma > 0$. \\
This restriction to Black-Scholes is not essential, as the main difficulty and nonlinearities are already present in this case
and the analysis technique can be extended straightforwardly to more general monotone schemes in the setting of more complex SDEs for $X$.
We take advantage of the specific dynamics to design a simple to implement numerical scheme, which also simplifies the notation.\\
We shall moreover work under the following hypothesis.
\begin{Assumption}\label{hypNum}
  The coefficient $\mu$ is non-negative.
\end{Assumption}
\begin{Remark} \label{boundSuperRep}
  This assumption is introduced without loss of generality in order to alleviate the notation in the scheme definition. We detail in Remark \ref{re scheme full def}(ii) how to  modify the schemefor non-positive drift $\mu$. The convergence properties are the same.
\end{Remark}
\noindent We now fix $n \ge 1$, $\cR_n$ the associated discrete set of controls (see Section \ref{sub-se disc control}). We denote $K = \cR^\flat_n$ assuming that $0 \notin K$ and  recall that $v_n$ is the solution to \eqref{eq pde more classic}. We consider the grid $\pi = \{0 =: t_0 < \cdots < t_k < \cdots < t_\kappa := T\}$ on $[0,T]$ and approximate $v_n$ by a PCPT scheme, extending Section \ref{sub-se pcpt schme}.

The main point here is that we introduce a finite difference approximation for the solution $S^a(\cdot)$, $a \in K$ to 
\eqref{PDEaconst}--\eqref{bCond1}.
This approximation, denoted by $S^a_\delta(\cdot)$ for a parameter $\delta > 0$, will be specified in Section \ref{sub-se de finite diff scheme} below. For $\delta > 0$ and $a \in K$, each approximation $S^a_\delta(\cdot)$ is defined on a spatial grid 
  \begin{align}
    \cG^a_\delta := \delta\Z \times \Gamma^a_\delta \subset \R\times[0,1].
  \end{align}
  where $\Gamma^a_\delta$ is a uniform grid of $[0,1]$, with $N^a_\delta+1$ points and mesh size $1/N^a_\delta$. 
  A typical element of $\cG^a_\delta$ is denoted 
  $(x_k,p_l) := (k\delta,  l/N^a_\delta)$,
  and an element of $\ell^\infty(\cG^a_\delta)$ is $u_{k,l} := u(x_k,p_l)$, for all $k \in \Z$ and $0 \le l \le N^a_\delta$.
  For $0 \le t < s \le T$, and $\varphi : \delta\Z \times [0,1] \to \R$ a bounded function, we have that $S^a_\delta(s,t,\varphi) \in \ell^\infty(\cG^a_\delta)$. 
  
 \noindent In order to define our approximation of $v_n$, it is not enough to replace $S^a(\cdot)$ in the minimisation \eqref{ThScheme}, or similarly \eqref{eq de w-pi-backward}, by $S^a_\delta(\cdot)$,  as the approximations are not defined on the same grid for the $p$-variable. (The flexibility of different grids will be important later on.) We thus have to consider a supplementary step which consists in a linear interpolation in the $p$-variable. Namely any mapping $u \in \ell^\infty(\cG^a_\delta)$ is extended into $\cI^a_\delta(u) : \delta\Z \times [0,1] \to \R$ by linear interpolation in the second variable: if $u \in \ell^\infty(\cG^a_\delta),k \in \Z$ and $p \in [p_l,p_{l+1})$ with $0 \le l < N^a_\delta$,
 \begin{align*}
   \cI^a_\delta(x_k,p) = \frac{p_{l+1}-p}{p_{l+1}-p_l} u_{k,l} + \frac{p-p_l}{p_{l+1}-p_l} u_{k,l+1},
 \end{align*}
 and obviously $\cI^a_\delta(x_k,1) = u_{k,N^a_\delta}$.\\
  The solution to the numerical scheme associated with $\pi, \delta$ is then $v_{n,\pi,\delta} : \pi\times\delta\Z\times[0,1] \to \R$ satisfying
  \begin{align} \label{def vnpidelta}
    \widehat{\mathfrak{S}}(\pi,\delta,t,x,p,v_{n,\pi,\delta}(t,x,p),v_{n,\pi,\delta}) = 0,
  \end{align}
  where, for any $0 \le t \in \pi, x \in \delta\Z, p \in [0,1], y \in \R^+$ and any bounded function $u:\pi\times\delta\Z\times[0,1]\to\R$:
  \begin{align}
    \label{DiscScheme}
    \widehat{\mathfrak{S}}(\pi,\delta,t,x,p,y,u) = 
    \left\{
    \begin{array}{rl}
      y - \min_{a \in K} \cI^a_\delta\left(S^a_\delta\left(t^+_\pi, t_k, u(t^+_\pi, \cdot)\right)\right)(t_k,x,p) & \mbox { if } k < \kappa, \\
      y - g(x)p & \mbox{ otherwise},
    \end{array} \right. 
  \end{align}
  where $t^+_\pi = \inf \{s \in \pi: s \ge t\}$.
  
  \vspace{4pt}
  \noindent Alternatively, the approximation $v_{n,\pi,\delta}$ is defined by the following backward induction: 
  \begin{enumerate}
  \item Initialisation: set $v_{n,\pi,\delta}(T,x,p) := {g}(x)p$,
    $x \in \R^d \times [0,1]$.
  \item Backward step: 
    For $k = \kappa-1,\dots,0$, compute $w^{k,a}_\delta := S^a_\delta(t_k,t_{k+1},v_{n,\pi,\delta}(t_{k+1},\cdot))$ and set, for $(x,p) \in \delta \Z \times [0,1]$,
    \begin{align}\label{eq de w-pi-backward delta}
      v_{n,\pi,\delta}(t_k,x,p) := \inf_{a \in K}\cI^a_\delta(w^{k,a}_\delta)(t_k,x,p)\;.  
    \end{align}
  \end{enumerate}
  
  \noindent Before stating the main convergence result of this section, see Theorem \ref{thm conv scheme} below, we give the precise definition of $S^a_\delta(\cdot)$ using finite difference operators.

\subsection{Finite difference scheme definition and convergence result}
\label{sub-se de finite diff scheme}

Let $0 \le t < s \le T, \delta > 0, \varphi : \delta\Z \times [0,1] \to \R$. We set $h := s-t > 0$. 

For $a \in K$, we will describe the grid $\cG^a_\delta = \delta\Z \times \Gamma^a_\delta \subset \delta\Z \times [0,1]$ and the finite difference scheme used to define
$S^a_\delta$. \\
First, we observe that for the model specification of this section, 
\eqref{eq pde more classic}
can be rewritten as
\begin{align}\label{pde num}
  \sup_{a \in K} \left(-D^a\varphi - \mu \nabla^a \varphi  - \frac{\sigma^2}2 \Delta^a\varphi - f(t,x,\varphi,\sigma \nabla^a\varphi) \right) = 0,
\end{align}
with:
\begin{align}
  \nabla^a \varphi &:= \partial_y\varphi + \frac{a}{\sigma}\partial_p \varphi, \\
  \label{diffOp} \Delta^a\varphi &:=  \partial^2_{yy} \varphi + 2 \frac{a}{\sigma}  \partial^2_{yp} \varphi + \frac{a^2}{\sigma^2} \partial^2_{pp} \varphi,\\
  D^a\varphi &:= \partial_t\varphi - \frac{a}{\sigma}\mu\partial_p\varphi.
\end{align}
Exploiting the degeneracy of the operators $\nabla^a$ and $\Delta^a$ in the direction $(a,-\sigma)$, we construct $\Gamma^a_\delta$ so that the solution to \eqref{pde num} is approximated by the solution of an implicit finite difference scheme with only one-directional derivatives.\\
To take into account the boundaries $p=0,p=1$, we set
\begin{align}\label{def N a delta}
  N^a_\delta := \min \left\{j \ge 1: j\frac{|a|}{\sigma}\delta \ge 1\right\} =
  \bigg\lceil{\frac{\sigma}{|a|\delta}\bigg\rceil}
\end{align}
and
\begin{align}\label{def a a delta}
  \mathfrak{a}(a,\delta) := \sgn(a)\frac{\sigma}{\delta N^a_\delta},
\end{align}
where $a \neq 0$. We have $N^a_\delta = \sigma/\delta|\mathfrak{a}(a,\delta)|$. 
We finally set:
\begin{align}
  \Gamma^a_\delta &:= \left\{0, \frac{|\mathfrak{a}(a,\delta)|}{\sigma}\delta, \dots, N^a_\delta\frac{|\mathfrak{a}(a,\delta)|}{\sigma}\delta=1\right\} = \left\{ \frac{j}{N^a_\delta}
  \; : \;j=0,\ldots,N^a_\delta \right\}.
\end{align}
We now define the finite difference scheme. To use the degeneracy of the operators $\nabla^{\mathfrak{a}(a,\delta)}$ and $\Delta^{\mathfrak{a}(a,\delta)}$ in the direction $(\mathfrak{a}(a,\delta),-\sigma)$, we define the following finite difference operators, for $v = (v_{k,l})_{k\in\Z,0\le l \le N^a_\delta} = (v(x_k,p_l))_{k\in\Z,0\le l\le N^a_\delta} \in \ell^\infty(\cG^a_\delta)$ and $w = (w_k)_{k\in\Z} \in \ell^\infty(k\Z)$: 
\begin{align*}
  \nabla^a_\delta v_{k,l} := \frac{1}{2\delta}\left(v_{k+1,l+\sgn(a)} - v_{k-1,l-\sgn(a)}\right), &\quad \nabla_\delta w_k := \frac{1}{2\delta}\left(w_{k+1}-w_{k-1}\right),\\
  \nabla^a_{+,\delta} v_{k,l} := \frac{1}{\delta}\left(v_{k+1,l+\sgn(a)}-v_{k,l}\right), &\quad \nabla_{+,\delta} w_k := \frac{1}{\delta}\left(w_{k+1}-w_k\right), \\
  \Delta^{\!a}_\delta v_{k,l} := \frac{1}{\delta^2}\left(v_{k+1,l+\sgn(a)} + v_{k-1,l-\sgn(a)} - 2v_{k,l}\right), &\quad \Delta_\delta w_k := \frac{1}{\delta^2}\left(w_{k+1} + w_{k-1} - 2w_k\right).
\end{align*}
Let $\theta > 0$ a parameter to be fixed later. We define, for $(t,x,y,q,q_+,A)\in[0,T]\times\R^5$.
\begin{align}
  \label{def H} F(t,x,y,q,A) &:= -\mu q - \frac{\sigma^2}{2} A - f(t,x,y,\sigma q), \mbox{ and} \\
  \label{def H hat}  \widehat F(t,x,y,q,q_+,A) &:= -\mu q_+ - \left(\frac{\sigma^2}2 + \theta\frac{\delta^2}{h} \right) A - f(t,x,y,\sigma q)\,.
\end{align}
Now, $S^a_\delta(s,t,\varphi) \in \ell^\infty(\cG^a_\delta)$ is defined as the unique solution to (see Proposition \ref{uniqueness picard scheme} below for the well-posedness of this definition)
\begin{align}
  \label{systu} S(k,l,v_{k,l}, \nabla^a_\delta v_{k,l},\nabla^a_{+,\delta}v_{k,l}, \Delta^{\!a}_\delta v_{k,l},\varphi) &= 0, \\
  \label{discBcond} v_{k,0} = \underline v_k, v_{k,N^a_\delta} &= \overline v_k,
\end{align}
where, for $k \in \Z, 0 < l < N^a_\delta, (v,v_+,v_-) \in \R^3$, and any bounded function $u:\delta\Z\times[0,1]\to\R$:
\begin{align}
  \label{S} S(k, l, v, q, q_+, A, u) = v - u\left(x_k, \textcolor{black}{\mathfrak p^a(p_l)}\right) + h \widehat F(t,k\delta,v,q,q_+,A),
\end{align}
with, for $p \in [0,1]$,
\begin{align} \label{p a p}
  \mathfrak p^a(p) := p - \mu\frac{\mathfrak a(a,\delta)}{\sigma}h,
\end{align}
and where $(\underline v_k)_{k \in \Z}$ (resp. $(\overline v_k)_{k \in \Z}$) is the solution to
\begin{align}
  \label{eq def bound 0} S_b(k,\underline v_k,\nabla_\delta \underline v_k, \nabla_{+,\delta} \underline v_k, \Delta_{\delta} \underline v_k,(\underline \varphi_k)_{k\in\Z}) &= 0, \\
  \label{eq def bound 1} (\mbox{resp. } S_b(k,\overline v_k,\nabla_\delta\overline v_k,\nabla_{+,\delta}\overline v_k,\Delta_{\delta}\overline v_k, (\overline \varphi_k)_{k\in\Z}) &= 0)
\end{align}
with $\underline \varphi_k = \varphi(k\delta,0)$ (resp. $\overline \varphi_k = \varphi(k\delta,N^a_\delta)$) and, for $k \in \Z, (v,v_+,v_-) \in \R^3, u \in \ell^\infty(\Z)$:
\begin{align}
  \label{Sb} S_b(k, v, q,q_+,A, u) = v - u_k+ h \widehat F(t,k\delta,v,q,q_+,A).
\end{align}
\begin{Remark}\label{re scheme full def}
  (i) Here, as stated before, we have assumed $\mu \ge 0$. If the opposite is true, one has to consider $\nabla^a_-(\delta)v_{k,l} := \frac{1}{\delta}\left(v_{k,l}-v_{k-1,l-\sgn(a)}\right)$ (resp.\ $\nabla_-(\delta) w_k := \frac{1}{\delta}\left(w_k-w_{k-1}\right)$) instead of $\nabla^a_{+,\delta}v_{k,l}$ (resp. $\nabla_{+,\delta}w_k$), in the definition of $S^a_\delta(s,t,\varphi)$ (resp.\ $\underline v_k, \overline v_k$), to obtain a monotone scheme.
\\
(ii)  For the nonlinearity $f$, we used the Lax-Friedrichs scheme \cite{CL84,DRZ18}, adding the term $\theta(v_++v_--2v)$ term in the definition of $\widehat F$ to
enforce monotonicity.
\end{Remark}
We now assume that the following conditions on the parameters are satisfied:
\begin{align}
  \label{CFL1} \delta &\le 1, \\
  \label{CFL2} \frac{hL}{2\delta} \le \theta &< \frac14,\\
  \label{CFL3} \mu h \le \delta &\le Mh,
\end{align}
for a constant $M > 0$.
Under these conditions, we prove that $S^a_\delta(s,t,\varphi)$ is uniquely defined, and that it can be obtained by Picard iteration.
\begin{Remark}
  Since $\mu h \le \delta$, we have $|\mu\frac{\mathfrak a(a,\delta)}{\sigma}h| \le \frac{|\mathfrak a(a,\delta)|}{\sigma}\delta$, which ensures that from
  \eqref{p a p}, $\mathfrak p^a(p_l) \in [0,1]$ for all $0 < l < N^a_\delta$.
\end{Remark}
\begin{Proposition} \label{uniqueness picard scheme}
  For every bounded function $\varphi : \delta\Z \times [0,1] \to \R$, there exists a unique solution to \eqref{systu}--\eqref{discBcond}.
\end{Proposition}
\begin{proof} 
  First, $\underline v \in \ell(\delta\Z)$ (resp. $\overline v \in \ell(\delta\Z)$) is uniquely defined by \eqref{eq def bound 0} (resp. \eqref{eq def bound 1}), see Proposition \ref{comp thm bound}.\\
  We consider the following map:
  \begin{align*}
    \ell^\infty(\cG^a_\delta) &\to \ell^\infty(\cG^a_\delta), \\
    v &\mapsto \psi(v),
  \end{align*}
  where $\psi(v)$ is defined by, for $k \in \Z$ and
  $l \in \{1,\dots,N_a-1\}$:
  \begin{align}
    \psi(v)_{k,l} \ =\ &\frac{1}{1 + \frac{h}{\delta}\mu + \sigma^2\frac{h}{\delta^2} +2\theta}\left(\varphi\left(k\delta, \textcolor{black}{\mathfrak p^a(p_l)}\right) + \right.\\
    \nonumber  &\hspace{-1.2 cm} \left. \frac{h}{\delta}\mu v_{k+1,l+\sgn(a)}+ \frac{\sigma^2}{2}\frac{h}{\delta^2}(v_{k+1,l+\sgn(a)}+v_{k-1,l-\sgn(a)}) + \right.\\
    \nonumber &\hspace{-1.2 cm} \left. hf \left(t^-,\textcolor{black}{k\delta},v_{k,l},\frac{\textcolor{black}{\sigma}}{2\delta}(v_{k+1,l+\sgn(a)} - v_{k-1,l-\sgn(a)})\right) + \theta(v_{k+1,l+\sgn(a)}+v_{k-1,l-\sgn(a)})\right),
  \end{align}
  \begin{align}
    \psi(v)_{k,0} = \underline v_k, \psi(v)_{k,N_a} = \overline v_k.
  \end{align}
  Notice that $v$ is a solution to \eqref{systu}-\eqref{discBcond} if and only if $v$ is a fixed point of $\psi$. It is now enough to show that $\psi$ maps $\ell^\infty(\cG^a_\delta)$ into $\ell^\infty(\cG^a_\delta)$  and is contracting. \\
  If $v \in \ell^\infty(\cG^a_\delta)$, by boundedness of $\varphi, \underline v$ and $\overline v$, it is clear that $\psi(v)$ is bounded. \\
  If $v^1, v^2 \in \ell^\infty(\cG^a_\delta)^2$, we
  have, for all $k \in \Z$ and $1 \le l \le N_a - 1$:
  \begin{align}
    |\psi(v^1)_{k,l} - \psi(v^2)_{k,l}| \le \frac{ \frac{h}{\delta}\mu+\sigma^2\frac{h}{\delta^2}+2\theta+hL+\frac{hL}{\delta}}{1+\frac{h}{\delta}\mu + \sigma^2\frac{h}{\delta^2} +2\theta}|v^1-v^2|_\infty.
  \end{align}
  Since $\delta \le 1$ by assumption (\ref{CFL1}), one has
  $hL + \frac{hL}{\delta} \le 2\frac{hL}{\delta} \le 4\theta$, thus:
  \begin{align}
    |\psi(v^1) - \psi(v^2)|_\infty \le \frac{ 4\theta + \frac{h}{\delta}\mu+\sigma^2\frac{h}{\delta^2}+2\theta}{1+\frac{h}{\delta}\mu + \sigma^2\frac{h}{\delta^2} +2\theta}|v^1-v^2|_\infty.
  \end{align}
  Since $4\theta < 1$ by assumption (\ref{CFL2}) and the function
  $x \mapsto \frac{4\theta + x}{1 + x}$ is increasing on $[0,\infty)$
  with limit $1$ when $x \to +\infty$, this proves that $\psi$ is a
  contracting mapping.  \eproof
\end{proof}
For this scheme, we have the following strong uniqueness result:
\begin{Proposition} \label{prop comp thm scheme}
  Let $\varphi^1, \varphi^2 : \delta\Z \times [0,1] \to \R$ two bounded functions satisfying $\varphi^1 \le \varphi^2$ on $\delta\Z \times [0,1]$.
  \begin{enumerate}
  \item (Monotonicity) For all $k \in \Z$, $1 \le l \le N_a$,
    $(v,q,q_+,A) \in \R^4$, we have:
    \begin{align} \label{monotone} S(k,l,v,q,q_+,A,\varphi^2) \le
      S(k,l,v,q,q_+,A,\varphi^1).
    \end{align}
  \item (Comparison theorem) Let
    $(v^1,v^2) \in \ell^\infty(\cG^a_\delta)^2$ satisfy,
    for all $k \in \Z$ and $1 \le l \le N^a_\delta - 1$:
    \begin{align}
      \nonumber S(k,l,v^1_{k,l},\nabla^a_\delta v^{\textcolor{black}{1}}_{k,l},&\nabla^a_{+,\delta}v^{\textcolor{black}{1}}_{k,l},\Delta^a_\delta v^{\textcolor{black}{1}}_{k,l},\varphi^2) \\
      \label{eq diff} &\hspace{0cm}\le S(k,l,v^2_{k,l},\nabla^a_\delta v^2_{k,l},\nabla^a_{+,\delta}v^2_{k,l},\Delta^a_\delta v^2_{k,l},\varphi^2)\\
      v^1_{k,0} &\le v^2_{k,0},\\
      v^1_{k,N^a_\delta} &\le v^2_{k,N^a_\delta}.
    \end{align}
    Then $v^1 \le v^2$.
  \item We have $S^a_\delta(s,t,\varphi^1)_{k,l} \le S^a_\delta(s,t,\varphi^2)_{k,l}$ for all $k \in \Z$ and $0 \le l \le N^a_\delta$.
  \end{enumerate}
\end{Proposition}
\begin{proof}
  Let $\varphi^1, \varphi^2$ as stated in the proposition.
  \begin{enumerate}
  \item We have, for $k \in \Z$ and $0 < l < N^a_\delta$:
    \begin{align*}
      &S(k,l,v,q,q_+,A,\varphi^2) - S(k,l,v,q,q_+,A,\varphi^1) \\
      &= (\varphi^1-\varphi^2)\left(x_k, \textcolor{black}{\mathfrak p^a(p_l)}\right) \le 0.
    \end{align*}
  \item We assume here that $a > 0$. For $k \in \Z$, let
    $M_k = \max_{0 \le l \le N^a_\delta} ( v^1_{k+l,l} - v^2_{k+l,l}) < \infty$ (if $a < 0$, we have to consider
    $\max_{0 \le l \le N^a_\delta} ( v^1_{k-l,l} - v^2_{k-l,l}$)). We want to prove that $M_k \le 0$ for all $k$. Assume
    to the contrary that there exists $k \in \Z$ such that $M_k >
    0$. Then there exists $0 \le l \le N^a_\delta$ such that
    \begin{align} \label{ptMax} v^1_{k+l,l} - v^2_{k+l,l} = M_k > 0.
    \end{align}
    First, we have $v^1_{k,0} \le v^2_{k,0}$ and $v^1_{k+N^a_\delta,N^a_\delta} \le v^2_{k+N^a_\delta,N^a_\delta}$. Thus $0 < l < N^a_\delta$. \\
    Moreover, using \eqref{eq diff}, re-arranging the terms, using the fact that $f$ is non-increasing with
    respect to its third variable and Lipschitz-continuous, by
    \eqref{ptMax}, 
    \begin{align}
      \nonumber (1 + 2\theta)M_k &\le \frac{hL}{2\delta}\left|v^2_{k+l+1,l+1}-v^1_{k+l+1,l+1}\right| - \theta(v^2_{k+l+1,l+1} - v^1_{k+l+1,l+1}) +\\
      \label{eqComp} &\frac{hL}{2\delta}\left|v^2_{k+l-1,l-1}-v^1_{k+l-1,l-1}\right| - \theta(v^2_{k+l-1,l-1} - v^1_{k+l-1,l-1}).
    \end{align}
    For $j \in \set{l-1,l+1}$, we observe that
    \begin{align} \label{eq comp temp 1}
 \frac{hL}{2\delta}|v^2_{k+j,j}-v^1_{k+j,j}| &- \theta(v^2_{k+j,j}-v^1_{k+j,j}) 
 \le  \left(\frac{hL}{2\delta} + \theta\right) M _k.
    \end{align}
    Indeed, if $v^2_{k+j,j} \ge v^1_{k+j,j}$ then
    \begin{align*}
      \frac{hL}{2\delta}|v^2_{k+j,j}-v^1_{k+j,j}| &- \theta(v^2_{k+j,j}-v^1_{k+j,j}) = \left(\frac{hL}{2\delta} - \theta\right)(v^2_{k+j,j}-v^1_{k+j,j}) \le 0,
    \end{align*}
    since $\frac{hL}{2\delta} \le \theta$. Otherwise, if
    $v^2_{k+j,j} < v^1_{k+j,j}$
    \begin{align*}
      \frac{hL}{2\delta}|v^2_{k+j,j}-v^1_{k+j,j}| - \theta(v^2_{k+j,j}-v^1_{k+j,j}) &= \left(\frac{hL}{2\delta} + \theta\right)(v^1_{k+j,j}-v^2_{k+j,j}) \\
                                                                                              &\le \left(\frac{hL}{2\delta} + \theta\right) M_k.
    \end{align*}
    Inserting \eqref{eq comp temp 1} into \eqref{eqComp}, we get
    \begin{align}
      (1 + 2\theta)M_k &\le 2 \left(\frac{hL}{2\delta} + \theta\right) M_k.
    \end{align}
    Thus,
    \begin{align}
      \left(1 - \frac{hL}{\delta}\right)M_k \le 0,
    \end{align}
    which is a contradiction to $M_k > 0$ since
    $\frac{hL}{\delta} \le 2\theta < \frac12$.
  \item Let $v^i = S^a_\delta(s,t,\varphi^i)$ for $i=1,2$.
    Since $\underline \varphi^1 \le \underline \varphi^2$ and $\overline \varphi^1 \le \overline \varphi^2$, we get by Proposition \ref{comp thm bound} that $v^1_{k,0} \le v^2_{k,0}$ and $v^1_{k,N^a_\delta} \le v^2_{k,N^a_\delta}$ for all $k \in \Z$.\\
    By monotonicity, we get, for all $k \in \Z$ and $0 < l < N^a_\delta$,
    \begin{align*}
      S(k,l,v^1_{k,l},\nabla^a_\delta v^{1}_{k,l},&\nabla^a_{+,\delta}v^{1}_{k,l},\Delta^a_\delta v^{1}_{k,l},\varphi^2)  
      \\
      &\le S(k,l,v^1_{k,l},\nabla^a_\delta v^{1}_{k,l},\nabla^a_{+,\delta}v^{1}_{k,l},\Delta^a_\delta v^{1}_{k,l},\varphi^1)
      \end{align*}
      Moreover,
      \begin{align*}
      S(k,l,v^1_{k,l},\nabla^a_\delta v^{1}_{k,l},&\nabla^a_{+,\delta}v^{1}_{k,l},\Delta^a_\delta v^{1}_{k,l},\varphi^1) 
      \\&  = S(k,l,v^2_{k,l},\nabla^a_\delta v^2_{k,l},\nabla^a_{+,\delta}v^2_{k,l},\Delta^a_\delta v^2_{k,l},\varphi^2) = 0
    \end{align*}
    So that,
     \begin{align*}
      S(k,l,v^1_{k,l},\nabla^a_\delta v^{1}_{k,l},&\nabla^a_{+,\delta}v^{1}_{k,l},\Delta^a_\delta v^{1}_{k,l},\varphi^2)  
      \\ &\le S(k,l,v^2_{k,l},\nabla^a_\delta v^2_{k,l},\nabla^a_{+,\delta}v^2_{k,l},\Delta^a_\delta v^2_{k,l},\varphi^2)
      \end{align*}   
    and the proof is concluded applying the previous point. \eproof
  \end{enumerate}
\end{proof}
\noindent We last give a refinement of the comparison theorem, which will be useful in the sequel.
\begin{Proposition} \label{diffsupersubsol}
  Let $u :\delta\Z \times [0,1] \to \R$ be a bounded function, and let
  $v^1, v^2 \in \ell^\infty(\cG^a_\delta)$. Assume that, for all $k \in \Z$ and $0 < l < N^a_\delta$, we have
  \begin{align*}
    S(k,l,v^1_{k,l},\nabla^a_\delta v^{1}_{k,l},&\nabla^a_{+,\delta}v^{1}_{k,l},\Delta^a_\delta v^{1}_{k,l},u) 
    \\
    &\le 0 \le S(k,l,v^2_{k,l},\nabla^a_\delta v^{2}_{k,l},\nabla^a_{+,\delta}v^{2}_{k,l},\Delta^a_\delta v^{2}_{k,l},u).
  \end{align*}
  Then:
  \begin{align}
    v^1_{k,l} - v^2_{k,l} \le e^{-4\frac{\mathfrak{a}(a,\delta)^2}{\sigma^2}C(h,\delta)l(N^a_\delta-l)}\left(|(v^1_{\cdot,0} - v^2_{\cdot,0})^+|_\infty + |(v^1_{\cdot,N^a_\delta} - v^2_{\cdot,N^a_\delta})^+|_\infty\right),
  \end{align}
  where
  \begin{align}
    \label{constC} C(h,\delta) := \frac{ \ln\left(\frac{1+\frac{h}{\delta}\mu + \sigma^2\frac{h}{\delta^2}+2\theta+\frac{hL}{2\delta}}{\frac{h}{\delta}\mu + \sigma^2\frac{h}{\delta^2}+2\theta+\frac{hL}{2\delta}}\right)}{\delta^2}.
  \end{align}
  Moreover,
  \begin{align}
    \label{minconstC} C(h,\delta) \ge \frac{1}{\left((\mu+\frac{L}{2})M+2\theta M^2\right)h^2 + \sigma^2h} - \frac{M^2}{2\sigma^4}.
  \end{align}
\end{Proposition}

\textcolor{black}{\begin{Remark}
  (i) To prove the consistency of the scheme, we define in Lemma \ref{lemsupersolum} smooth functions $w^\pm$ so that $(w^\pm(x_k,p_l)) \in l^\infty(\cG^a_\delta)$ satisfy $S \ge 0$ or $S \le 0$, but we cannot use the comparison theorem as the values at the boundary cannot be controlled. The previous proposition will be used in Lemma \ref{lem interp} to show that the difference between $w^\pm$ and the linear interpolant of a solution of $S = 0$ is small.
  \\
  (ii) The coefficient $\exp\!\left(-4\frac{\mathfrak a(a,\delta)}{\sigma^2}C(h,\delta)l(N^a_\delta-l\right)$ that appears in the first equation of the previous proposition shows that the dependance on the boundary values decays exponentially with the distance to the boundary. This was to be expected and was already observed in similar situations, see for example Lemma 3.2 in \cite{BJ07} for Hamilton-Jacobi-Bellman equations.
\end{Remark}}
We now can state the main result of this section.
  \begin{Theorem} \label{thm conv scheme}
    The function $v_{n,\pi,\delta}$ converges to $v_n$ uniformly on compact sets, as $|\pi|,\delta \to 0$ satisfying conditions \eqref{CFL1}--\eqref{CFL3} 
    for all $h = t_{i+1}-t_i$, where $\pi = \{0 = t_0 < t_1 < \dots < t_\kappa = T\}$.
  \end{Theorem}
  We prove below that the scheme is \emph{monotone} (see Proposition \ref{monotone scheme}), \emph{stable} (see Proposition \ref{stable scheme}), \emph{consistent} with \eqref{eq pde more classic} in $[0,T) \times \R \times (0,1)$ (see Proposition \ref{consistent scheme}) and with the boundary conditions (see Proposition \ref{prop bound cond num}).
  The theorem then follows by identical arguments to \cite{BS91}.

\subsection{Proof of Theorem \ref{thm conv scheme}} \label{subsec proof}
We first show that the numerical scheme is consistent with the boundary conditions. For any discretisation parameters $\pi,\delta$, we define $V_{\pi,\delta} : \pi \times \delta\Z \to \R$ as the solution to the following system:
\begin{align} \label{eq de V_pi_delta}
  S_b(k,v^j_k,\nabla_\delta v^j_k,\nabla_{+,\delta}v^j_k,\Delta_{\delta}v^j_k, v^{j+1}_k) &= 0, k \in \Z, 0 \le j < \kappa \\
  v^\kappa_k &= g(x_k), k \in \Z,
\end{align}
where $v^j_k := v(t_j,x_k)$ for $0 \le j \le \kappa$ and $k \in \Z$. We set $(U_{\pi,\delta})^j_k := \nabla_\delta(V_{\pi,\delta})^j_k = \frac{1}{2\delta}((V_{\pi,\delta})^j_{k+1} - (V_{\pi,\delta})^j_{k-1})$. We recall from 
Proposition \ref{bound}
that $V_{\pi,\delta}$ and $U_{\pi,\delta}$ are bounded, uniformly in $\pi,\delta$, and, by \cite{BS91}, that $V_{\pi,\delta}$ converges to $V$ uniformly on compact sets as $|\pi|\to 0$ and $\delta \to 0$.
\begin{Proposition}\label{prop bound cond num} There exists constants $K_1,K_2,K_3 > 0$ such that, for all discretisation parameters $\pi,\delta$ with $|\pi|$ small enough, we have, for $(t_j,x_k,p) \in \pi \times \delta\Z \times [0,1]$:
  \begin{align*}
    p V_{\pi,\delta}(t_j, x_k) - K_1(T-t_j) &\le v_{n,\pi,\delta}(t_j,x_k,p) \le p V_{\pi,\delta}(t_j,x_k) + K_1(T-t_j), \\
    pV_{\pi,\delta}(t_j,x_k) - (1 - &e^{-K_2p})(1-e^{-K_2(1-p)}) \le v_{n,\pi,\delta}(t_j,x_k,p) \\ \le &pV_{\pi,\delta}(t_j,x_k) + (1-e^{-K_2p})(1-e^{-K_2(1-p)}).
  \end{align*}
\end{Proposition}
\begin{proof}
    We only prove, by backward induction, the lower bounds, while the proof of the upper bounds is similar. We need to introduce first some notation.
    For $0 \le j \le \kappa$, $k \in \delta\Z$ and $0 \le l \le N^a_\delta$, we set $V^j_k := V_{\pi,\delta}(t_j,x_k)$ and $U^j_k := U_{\pi,\delta}(t_j,x_k)$. For $\epsilon \in \set{0,1}$, we define:
    \begin{align}
      \we(t_j,x_k,p) :=   p V^j_k -  \ce(t_j,p),
    \end{align}
    with
    \begin{align}
      \ce(t_j,p) := \epsilon K_1(T-t_j) + (1-\epsilon)(1-e^{-K_2 p})(1-e^{-K_2(1-p)}),
    \end{align}  
    and
    $\we^j_{k,l} = \we(t_j,x_k,p_l), \ce^j_l = \ce(t_j,p_l)$, $p_l \in \Gamma^a_\delta$.
    The proof now procedes in two steps.\\
    \noindent 1. 
    First, we have $\we(T,x_k,p) \le p V_{\pi,\delta}(T,x_k) = pg(x_k) = v_{n,\pi,\delta}(T,x_k,p)$ on $\delta\Z\times[0,1]$.\\
  Suppose that, for $0 \le j < \kappa$, on $\delta\Z\times[0,1]$, we have
  \begin{align*}
    \we(t_{j+1},x_k,p)\le v_{n,\pi,\delta}(t_{j+1},x_k,p).
  \end{align*}
  We want to prove on $\delta\Z\times[0,1]$
  \begin{align*}
    \we(t_j,x_k,p) \le v_{n,\pi,\delta}(t_j,x_k,p).
  \end{align*}
  Since $\we$ is convex in $p$, $\we(t_j,x_k,\cdot) \le \cI^a_\delta(\we^j_{k,\cdot})$ on $[0,1]$. By definition, we have $v_{n,\pi,\delta}(t_j,x_k,p) = \min_{a \in K} \cI^a_\delta(S^a_\delta(t_{j+1},t_j,v_{n,\pi,\delta}(t_{j+1},\cdot)))(x_k,p)$, we are thus going to prove
  \begin{align}
    \we^j_{k,l} \le S^a_\delta(t_{j+1},t_j,v_{n,\pi,\delta}(t_{j+1},\cdot))(t_j,x_k,p_l) \label{eq disc comp a}
  \end{align}
  for all $a \in K$ and all $k \in \Z, 0 \le l \le N^a_\delta$. 
  \\
For $a \in K$,
by induction hypothesis, $\we(t_{j+1},\cdot) \le v_{n,\pi,\delta}(t_{j+1},\cdot)$, so if we are able to get  
  \begin{align}
    S_b(k,\underline \we_k, \nabla_\delta\underline \we_k, \nabla_{+,\delta}\underline \we_k, \Delta_{\delta} \underline \we_k, \we^{j+1}_{k,0}) &\le 0, k \in \Z,  \label{eq scheme up}\\
    S_b(k,\overline \we_k, \nabla_\delta\overline \we_k, \nabla_{+,\delta}\overline \we_k, \Delta_{\delta} \overline \we_k, \we^{j+1}_{k,N^a_\delta}) &\le 0, k \in \Z,\label{eq scheme down} 
    \\
    S(k,l,\we_{k,l},\nabla^a_\delta \we_{k,l}, \nabla^a_{+,\delta} \we_{k,l}, \Delta^a_\delta \we_{k,l}, \we(t_{j+1},\cdot)) &\le 0, k \in \Z, 0 < l < N^a_\delta, \label{eq scheme in}
  \end{align}
  where $\underline \we^j_k = \we(t_j,x_k,0), \overline \we^j_k =  \we(t_j,x_k,1)$, we obtain that \eqref{eq disc comp a} holds true by the comparison result in Proposition \ref{prop comp thm scheme}, which concludes the proof. We now proceed with the proof of \eqref{eq scheme up}, \eqref{eq scheme down} and \eqref{eq scheme in}.\\
 \noindent 2.a Now, observe that $\underline \we^j_k =-\epsilon K_1(T-t_j)$, for $k \in \Z$. We have, since $f(t_j,x_k,0,0) = 0$ and $f$ is non-increasing in its third variable,
    \begin{align*}
      S_b(k,\underline \we_k, \nabla_\delta\underline \we_k, \nabla_{+,\delta}\underline \we_k, \Delta_{\delta} \underline \we_k, \we^{j+1}_{k,0}) = -\epsilon Kh - h f(t_j,x_k,-\epsilon K(T-t_j),0) \le 0.
    \end{align*}
 2.b We have that $\overline \we^j_k =V^j_k-\epsilon K_1(T-t_j)$, for $k \in \Z$. Since 
 $$f(t_j,x_k,V^j_k - \epsilon K_1(T-t_j), U^j_k) \ge f(t_j,x_k,V^j_k,U^j_k)\,,$$ and by definition of $V_{\pi,\delta}$: 
    \begin{align*}
       S_b(k,\overline \we_k, \nabla_\delta\overline \we_k, \nabla_{+,\delta}\overline \we_k, \Delta_{\delta} \overline \we_k, \we^{j+1}_{k,N^a_\delta}) &= -\epsilon Kh + S_b(k,V^j_k,\nabla_\delta V^j_k, \nabla_{+,\delta}V^j_k, \Delta_{\delta}V^j_k)\\ &\le -\epsilon Kh \le 0.
    \end{align*}
    2.c We now prove \eqref{eq scheme in}. 
    Let $k \in \Z$, $0 < l < N^a_\delta$. We have, by definition \eqref{S} of $S$:
    \begin{align*}
      S(k,l,\we^j_{k,l},&\nabla^a_\delta\we^j_{k,l},\nabla^a_{+,\delta}\we^j_{k,l},\Delta^a_\delta\we^j_{k,l},\we(t_{j+1},\cdot)) \\
                      &= \we^j_{k,l} - \we(t_{j+1},x_k,\textcolor{black}{\mathfrak p^a(p_l)})\\
                        &\hspace{0.5cm} + h \widehat F(t,k\delta,\we^j_{k,l},\nabla^a_\delta\we^j_{k,l},\nabla^a_{+,\delta}\we^j_{k,l},\Delta^a_\delta\we^j_{k,l}) \\
                        &\le - \ce^j_{k,l} + \mu\frac{\mathfrak a(a,\delta)}{\sigma}hV^{j+1}_k + \ce(t_{j+1},x_k,\mathfrak p^a(p_l))\\
                        &\hspace{0.5cm} - p_l h \widehat F(t,x_k, V^j_k, U^j_k, \nabla_{+,\delta}V^j_k,\Delta_{\delta}V^j_k)\\
                        &\hspace{0.5cm} + h\widehat F(t,x_k,p_lV^j_k,\nabla^a_\delta\we^j_{k,l},\nabla^a_{+,\delta}\we^j_{k,l},\Delta^a_\delta\we^j_{k,l}),
    \end{align*}
    where we have used \eqref{eq de V_pi_delta} and $f(t,x_k,\we^j_{k,l},\sigma\nabla^a_\delta\we^j_{k,l}) \ge f(t,x_k,p_l V^j_k,\sigma\nabla^a_\delta\we^j_{k,l})$.\\
    By adding $\pm p_l h f(t_j,x_k,p_l V^j_k, \sigma \nabla^a_\delta\we^j_{k,l})$, using the Lipschitz continuity of $f$ and 
    \begin{align*}\nabla^a_\delta\we^j_{k,l} = p_l U^j_k + \frac{\mathfrak{a}(a,\delta)}{2\sigma}(V^j_{k+1}+V^j_{k-1}) + \frac{1}{2\delta}\left(\ce^j_{l-\sgn(a)}-\ce^j_{l+\sgn(a)}\right),
    \end{align*}
    we get, by definition \eqref{def H hat} of $\widehat F$,
    \begin{align*}
      S(k,l,\we_{k,l},&\nabla^a_\delta\we^j_{k,l},\nabla^a_{+,\delta}\we^j_{k,l},\Delta^a_\delta\we^j_{k,l},\we^j(t_{j+1},\cdot))\\
      \le & h\frac{\mathfrak{a}(a,\delta)}{\sigma}\mu(V^{j+1}_k-V^j_{k+1}) - h\sigma\mathfrak{a}(a,\delta)U^j_k - 2\theta\frac{\mathfrak{a}(a,\delta)}{\sigma}\delta^2U^j_k\\ &+ 2hLp_l(1-p_l)(V^j_k + |U^j_k|) + hL\frac{|\mathfrak{a}(a,\delta)|}{2\sigma}(V^j_{k+1}\textcolor{black}{+}V^j_{k-1})
      \\ & -\left(\ce^j_l - \ce(t_{j+1},\mathfrak p^a(p_l)) - \mu h \nabla^a_{+,\delta}\ce^j_l - \left(\frac{\sigma^2}{2} h + \theta \delta^2\right) \Delta^a_\delta\ce^j_l \right)
      \\ & + hL |\nabla^a_\delta\ce^j_l|.
    \end{align*}
    Since $|\mathfrak a(a,\delta)| \le \max \{|a|, a \in K\} \le n$ and $V$ and $U$ are bounded uniformly in $h,\delta$ (see Proposition \ref{bound} in the appendix), there exists a constant $K_{n,\theta,M,L} > 0$ such that
    \begin{align*}
      &h\frac{\mathfrak{a}(a,\delta)}{\sigma}\mu(V^{j+1}_k-V^j_{k+1}) - h\sigma\mathfrak{a}(a,\delta)U^j_k - 2\theta\frac{\mathfrak{a}(a,\delta)}{\sigma}\delta^2U^j_k\\ &+ 2hLp_l(1-p_l)(V^j_k + |U^j_k|) + hL\frac{|\mathfrak{a}(a,\delta)|}{2\sigma}(V^j_{k+1}\textcolor{black}{+}V^j_{k-1}) \le h K_{n,\theta,M,L}.
    \end{align*}
    When $\epsilon = 1$, the terms of the last three lines all vanish except the first one, and $c^j_l - c(t_{j+1},p_l-\mu\frac{\mathfrak{a}(a,\delta)}{\sigma}h) = K_1h$. Thus we get:
    \begin{align*}
      S(k,l,\we_{k,l},&\nabla^a_\delta\we_{k,l},\nabla^a_{+,\delta}\we_{k,l},\Delta^a_\delta\we_{k,l},\we(t_{j+1},\cdot))
      \le  h(-K_1 + K_{n,\theta,M,L}).
    \end{align*}
    Hence, chosing $K_1$ large enough gives the result.\\
    We now deal with the case $\epsilon = 0$. By Taylor expansions of $\ce$ around $(t_j,p_l)$, we get:
    \begin{align*}
      &S(k,l,\we_{k,l},\nabla^a_\delta\we^j_{k,l},\nabla^a_{+,\delta}\we^j_{k,l},\Delta^a_\delta\we^j_{k,l},\we^j(t_{j+1},\cdot))\\
     \hspace{0 cm} &\le h K_{n,\theta,M,L} + hL\frac{|\mathfrak a(a,\delta)|}{\sigma}|\partial_p\,\ce(t_j,p_l)| + h \partial_t\,\ce(t_j,p_l) + h \frac{\mathfrak a(a,\delta)^2}{2}\partial^2_{pp}\,\ce(t_j,p_l) + h \varepsilon(h;K_2),
    \end{align*}
    with $\lim_{h \to 0} \varepsilon(h;K_2) = 0$. By definition of $\ce$, we get, for $h_0 > 0$ to be fixed later on and $h \in [0,h_0]$:
    \begin{align*}
      S(k,l,\we_{k,l},&\nabla^a_\delta\we^j_{k,l},\nabla^a_{+,\delta}\we^j_{k,l},\Delta^a_\delta\we^j_{k,l},\we^j(t_{j+1},\cdot))\\
      \le &h\left[ K_{n,\theta,M,L} +  K_2 L \frac{|\mathfrak a(a,\delta)|}{\sigma}e^{-K_2p_l} + K_2 L\frac{\mathfrak a(a,\delta)}{\sigma}e^{-K_2(1-p_l)} \right. \\ 
      &\left. - K_2^2\frac{\mathfrak a(a,\delta)^2}{2}e^{-K_2 p_l} - K_2^2\frac{\mathfrak a(a,\delta)^2}{2}e^{-K_2(1-p_l)} + |\varepsilon(h;K_2)|\right]\\
  &\hspace{-2cm} \le  h\left[ \max_{h \in [0,h_0]} |\varepsilon(h;K_2)| + K_{n,\theta,M,L} + K_2|\mathfrak a(a,\delta)|(e^{-K_2p_l}+e^{-K_2(1-p_l)})\left(\frac{L}{\sigma}-\frac{|\mathfrak a(a,\delta)|}{2}K_2\right)\right].
    \end{align*}
    To conclude, one can choose $K_2$ large enough so that $K_{n,\theta,M,L} + K_2|\mathfrak a(a,\delta)|(e^{-K_2p_l} + e^{-K_2(1-p_l)})(\frac{L}{\sigma} - \frac{|\mathfrak a(a,\delta)}{2}K_2) \le -\eta < 0$, and then consider $h_0 > 0$ small enough so that $|\varepsilon(h;K_2)| \le \eta$ for $h \in [0,h_0]$.
    \eproof
\end{proof}

\begin{Proposition}[Monotonicity] \label{monotone scheme}
  Let $\pi$ be a grid of $[0,T]$ and $\delta > 0$ satisfying \eqref{CFL1}--\eqref{CFL3}.
  Let $y \in \R, 0 \le k \le \kappa, j \in \Z$ and $p \in [0,1]$, and let  $ \cU, \cV :\pi\times\delta\Z\times[0,1] \to \R$ be two bounded functions such that $\cU \le \cV$. Then:
  \begin{align}
    \widehat{\mathfrak{S}}(\pi,\delta,k,j,p,y,\cU) \ge \widehat{\mathfrak{S}}(\pi,\delta,k,j,p,y,\cV).
  \end{align}
\end{Proposition}
\begin{proof}
  The result is clear for $k = \kappa$. If $k < \kappa$, it is sufficient to show that:
  \begin{align}\nonumber
    \cI^a_\delta(S^a_\delta(t_{k+1},t_k,\cU(t_{k+1},\cdot))) \le \cI^a_\delta(S^a_\delta(t_{k+1},t_k, \cV(t_{k+1},\cdot))),
  \end{align}
   for all $a \in K$,
  recalling \eqref{DiscScheme}. This is a consequence of the comparison result in Proposition \ref{prop comp thm scheme} and the monotonicity of the linear interpolator.
  \eproof
\end{proof}
We now prove the stability of the scheme. Here, in contrast to Lemma \ref{lem inc p}, we are not able to prove that the solution of the scheme is increasing in $p$. However, due to the boundedness of the terminal condition, we obtain uniform bounds for $v_{n,\pi,\delta}$.
\begin{Proposition}[Stability] \label{stable scheme}
  For all $\pi$ and $\delta > 0$, there
  exists a unique solution $v_{n,\pi,\delta}$ to \eqref{DiscScheme}, which
  satisfies:
  \begin{align}
    0 \le v_{n,\pi,\delta} \le |g|_\infty \mbox{ on } \pi \times \delta\Z \times [0,1].
  \end{align}
\end{Proposition}

\begin{proof}
  We prove the proposition by backward induction. \\
  First, since $v_{n,\pi,\delta}$ is a solution to \eqref{DiscScheme}, $v_{n,\pi,\delta}(T,x,p) = pg(x)$ on $\delta\Z\times[0,1]$, and we have $0 \le v_{n,\pi,\delta}(T,x,p) \le |g|_\infty$ for all $(x,p) \in \delta\Z \times [0,1]$.\\
  Let $0 \le j \le \kappa - 1$ and assume that $v_{n,\pi,\delta}(t_k,\cdot)$ is uniquely determined for $k > j$, and that $0 \le v_{n,\pi,\delta}(t_{j+1},\cdot) \le |g|_\infty$. Since $v_{n,\pi,\delta}$ is a solution to \eqref{DiscScheme}, we have
    \begin{align*}
      v_{n,\pi,\delta}(t_j,x,p) = \min_{a \in K} \cI^a_\delta(S^a_\delta(t_{j+1},t_j,v_{n,\pi,\delta}(t_{j+1},\cdot))),
    \end{align*}
    and for each $a \in K$, $S^a_\delta(t_{j+1},t_j,v_{n,\pi,\delta}(t_{j+1},\cdot))$ is uniquely determined by Proposition \ref{uniqueness picard scheme}, so $v_{n,\pi,\delta}(t_j,\cdot)$ is uniquely determined. Next, we show
  that, for all $k \in \Z$ and $0 \le l \le N_a$:
  \begin{align*}
    0 \le S^a_\delta(t_{j+1},t_j,v_{n,\pi,\delta}(t_{j+1},\cdot)) \le |g|_\infty.
  \end{align*}
  Then it is easy to conclude that $0 \le v_{n,\pi,\delta}(t_j,\cdot)  \le e^{LT}|g|_\infty$ on $ \R \times [0,1]$, by properties of the linear interpolation and the minimisation. \\
  First, it is straightforward that $\check{u}$ defined by $\check{u}_{k,l} = 0$ for all $k \in \Z$ and $0 \le l \le N_a$ satisfies $\check{u} = S^a_\delta(t_{j+1},t_j,0)$.\\ 
  The comparison theorem gives $0 \le S^a_\delta(t_{j+1},t_j,v_{n,\pi,\delta}(t_{j+1},\cdot))$,  since $0 \le v_{n,\pi,\delta}(t_{j+1},\cdot)$. \\
  To obtain the upper bound, we notice that $\hat{u}$ defined by $\hat{u}_{k,l} := |g|_\infty$ for all $k \in \Z$ and $0 \le l \le N_a$ satisfies 
  \begin{align*}
S(k,l,\hat{u}_{k,l},&\nabla^a_\delta\hat{u}_{k,l},\nabla^a_{+,\delta}\hat{u}_{k,l},\Delta^a_\delta\hat{u}_{k,l},\hat{u}) = -hf(t_j,x_k,\hat{u},0) \ge -hf(t_j,x_k,0,0) \ge 0.
\end{align*}
  Hence the comparison result in Proposition \ref{prop comp thm scheme} yields
  $S^a_\delta(t_{j+1},t_j,v_{n,\pi,\delta}(t_{j+1},\cdot)) \le |g|_\infty$.  \eproof
\end{proof}
We now prove the consistency. The proof requires several lemmata. First, we show that the perturbation induced by the change of controls vanishes as $\delta \to 0$.
\begin{Lemma}\label{diffDelta}
  For all $a \in K$, $a$ and $\mathfrak{a}(a,\delta)$ have the same sign, and:
  \begin{align} \label{diff a}
    0 \le |a| - |\mathfrak{a}(a,\delta)| \le \frac{n^2}{\sigma}\delta.
  \end{align}
  Moreover, there exists $c > 0$ such that for all $a \in K$ and $\delta > 0$, $|\mathfrak a(a,\delta)| \ge c > 0$.
\end{Lemma}
\begin{proof}
  By definition of $N^a_\delta$,
  \begin{align*}
    (N^a_\delta - 1)\frac{|a|}{\sigma}\delta < 1 = N^a_\delta\frac{|\mathfrak{a}(a,\delta)|}{\sigma}\delta \le N^a_\delta\frac{|a|}{\sigma}\delta,
  \end{align*}
  thus
  \begin{align*}
    |a| - \frac{|a|}{N^a_\delta} < |\mathfrak{a}(a,\delta)| \le |a|.
  \end{align*}
  Also, we observe
  \begin{align*}
    \frac{|a|}{N^a_\delta} = \frac{|a|}{\lceil{\frac{\sigma}{|a|\delta}\rceil}} \le \frac{|a|}{\frac{\sigma}{|a|\delta}} = \frac{a^2 \delta}{\sigma} \le \frac{n^2}{\sigma}\delta,
  \end{align*}
  which concludes the proof of \eqref{diff a}.\\
  By \eqref{def N a delta}, we have:
  \begin{align*}
    N^a_\delta \le \frac{\sigma}{|a|\delta} + 1 \le \frac{\sigma}{a_m\delta}+1 \le \frac{\sigma}{c\delta}
  \end{align*}
  where $a_m = \min \{ |a| : a \in K \}$ and $c > 0$ is independant of $a,\delta$.\\
  Now, by \eqref{def a a delta}, we get:
  \begin{align*}
    |\mathfrak a(a,\delta)| &= \frac{\sigma}{\delta N^a_\delta} \ge \frac{c\delta\sigma}{\delta\sigma} = c.
  \end{align*} \eproof
\end{proof}

  Last, we give explicit supersolutions and subsolutions satisfying appropriate conditions.
  Let $0 \le t < s \le T$, $\delta > 0$ and $a \in K$ be fixed.
  For $\epsilon > 0$, we set
  \begin{align*}
    f_\epsilon(t,x,y,\nu) &:= (f(t,\cdot,\cdot,\cdot) * \rho_\epsilon)(t,y,\nu) \\ &:= \int_{\R\times\R\times\R} f(t,x-u,y-z,\nu-\eta)\rho_\epsilon(u,z,\eta)\ud u \ud z \ud \eta,
  \end{align*}
  where $*$ is the convolution operator and, for $\epsilon > 0$, $\rho_\epsilon(x) := \epsilon^{-3}\rho(x/\epsilon)$ with $\rho:\R^3\to\R$ is a mollifier, i.e.\ a smooth function supported on $[-1,1]^3$ satisfying $\int_{\R} \rho = 1$. We set
  \begin{align*}
    F_\epsilon(t,x,y,q,A) = \left(\frac12\sigma^2-\mu\right)q - \frac{\sigma^2}2A - f_\epsilon(t,x,y,\sigma q).
  \end{align*}
  \begin{Remark}
    Since $f$ is L-Lipschitz continuous with respect to its three last variables, we have $|f_\epsilon - f|_\infty \le L \epsilon$.
  \end{Remark}
  
  The lengthy proof of the following lemma by insertion is given in the appendix.
  \begin{Lemma} \label{lemsupersolum}
    Let $0 \le t < s \le T, \varphi \in \cC^\infty_b(\R\times\R,\R), a \in K$. We set $h := s - t$. Let $\epsilon > 0$ such that $\epsilon \to 0$ and $\frac{\delta}{\epsilon^2} \to 0$ as $h \to 0$, observing \eqref{CFL3}.\\
    Then there exist bounded functions
    $S^{a,\pm}_\delta(s,t,\varphi):\delta\Z\times[0,1]\to\R$ of the
    form
    \begin{align}
      \label{def S a pm}S^{a,\pm}_{\delta,\epsilon}(s,t,\varphi)(x,p) = &\varphi(x,\textcolor{black}{\mathfrak p^a(p)})\\
      \nonumber &- h F_\epsilon(t,x,\varphi(x,\textcolor{black}{\mathfrak p^a(p)}),\nabla^{\mathfrak{a}(a,\delta)}\varphi(x,\textcolor{black}{\mathfrak p^a(p)}),\Delta^{\mathfrak{a}(a,\delta)}\varphi(x,\textcolor{black}{\mathfrak p^a(p)}))\\
      \nonumber &\pm C_{\varphi,n}(h,\epsilon),
    \end{align}
    where $\mathfrak p^a$ is defined in \eqref{p a p}, and where $C_{\varphi,n}(h,\epsilon) > 0$ satisfies $\frac{C_{\varphi,n}(h,\epsilon)}{h} \to 0$ as $h \to 0$,
    such that 
    $w^\pm := (S^{a,\pm}_{\delta,\epsilon}(s,t,\varphi)(x_k,p_l))_{k
      \in \Z, 0 \le l \le N^a_\delta} \in \ell^\infty(\cG^a_\delta))$ satisfy
    \begin{align}
      \label{supersolinside} S(k,l,w^+_{k,l},\nabla^a_\delta w^+_{k,l},\nabla^a_{+,\delta}w^+_{k,l},\Delta^a_\delta w^+_{k,l})  &\ge 0, \\
      \label{subsolinside} S(k,l,w^-_{k,l},\nabla^a_\delta w^-_{k,l},\nabla^a_{+,\delta}w^-_{k,l},\Delta^a_\delta w^-_{k,l})&\le 0,
    \end{align}
    for all $k \in \Z$ and $0< l < N^a_\delta$.\\
    Furthermore, for all $x \in \delta\Z$,
    $S^{a,\pm}_{\delta,\epsilon}(s,t,\varphi)(x,\cdot) \in \cC^2([0,1],\R)$, and $|\partial_{pp}^2S^{a,\pm}_{\delta,\epsilon}(s,t,\varphi)|_\infty \le \frac{C_{\varphi}(h)}{\epsilon^2}$ for some constant $C_{\varphi}(h) > 0$ independent of $\epsilon$.
  \end{Lemma}

\begin{Lemma} \label{lem interp}
  Let
  $0 \le t < s \le T, \delta > 0, a \in K, \varphi \in
  \cC^\infty_b(\R\times\R)$ be fixed. Let $h = s-t$,
  $k \in \Z, x_k \in \delta\Z, p \in (0,1)$, and assume that $h$ is sufficiently small so that $p \in [p_1, p_{N^a_\delta-1}]$, observing \eqref{CFL3}.
  Let $\epsilon > 0$ such that $\epsilon \to 0$ and $\frac{\delta}{\epsilon^2} \to 0$ as $h \to 0$.
  Then we have:
  \begin{align}
    S^{a,-}_{\delta,\epsilon}(s,t,\varphi)(x_k,p)-\cI^a_\delta(S^a_\delta(s,t,\varphi))(x_k,p) &\le C'_{\varphi,n}(h,\epsilon), \\
    \cI^a_\delta(S^a_\delta(s,t,\varphi))(x_k,p)-S^{a,+}_{\delta,\epsilon}(s,t,\varphi)(x_k,p)  &\le C'_{\varphi,n}(h,\epsilon),
  \end{align}
  where $C'_{\varphi,n}(h,\epsilon) > 0$ satisfies $\frac{C'_{\varphi,n}(h,\epsilon)}{h} \to 0$, as $h \to 0$ and where the functions 
  $S^{a,\pm}_{\delta,\epsilon}(s,t,\varphi)$ are introduced in Lemma \ref{lemsupersolum}.
\end{Lemma}
\begin{proof}
  We prove the first identity, the second one is similar. \\
  Set $w := S^a_\delta(s,t,\varphi)$ and $w^- := S^{a,-}_{\delta,\epsilon}(s,t,\varphi)$.
  By definition of $w$ and by \eqref{subsolinside}, one can apply
  Proposition \ref{diffsupersubsol}. For all $k \in \Z$ and
  $0 < l < N^a_\delta$:
  \begin{align}
    \label{subsolcontrol} w^-_{k,l} - w_{k,l} \le
    Be^{-4\frac{\mathfrak{a}(a,\delta)^2}{\sigma^2}C(h,\delta)l(N^a_\delta-l)} \le
    Be^{-4\frac{\mathfrak{a}(a,\delta)^2}{\sigma^2}C(h,\delta)(N^a_\delta-1)},
  \end{align}
  with $B = |(w^-_{\cdot,0} - w_{\cdot,0})^+|_\infty + |(w^-_{\cdot,N^a_\delta}-w_{\cdot,N^a_\delta})^+|_\infty$ and $C(h,\delta)$ is defined in \eqref{constC}. By Lemma \ref{diffDelta}, there exists a constant $c > 0$ such that $|\mathfrak a(a,\delta)| \ge c$. In addition, using \eqref{minconstC}, we get:
  \begin{align*}
  \frac{B}{h}e^{-4C(h,\delta)\frac{\mathfrak{a}(a,\delta)^2}{\sigma}(N^a_\delta -1)} &\le \frac{B}{h} e^{-4 \frac{c^2}{\sigma^2}\left(\frac{1}{\left((\mu+\frac{L}{2})M+2\theta M^2\right)h^2 + \sigma^2h} - \frac{M^2}{2\sigma^4}\right)}\\ &= B e^{4\frac{c^2M^2}{2\sigma^6}} \frac{e^{-4\frac{c^2}{\sigma^2}\frac{1}{\left((\mu+\frac{L}{2})M+2\theta M^2\right)h^2 + \sigma^2h}}}{h}\to 0,
  \end{align*}
  as $h \to 0$. \\
  Now, let $p \in [p_1,p_{N^a_\delta-1})$ and $k \in \Z$. By definition of
  $\cI^a_\delta$, one has:
  \begin{align}
    \cI^a_\delta(S^a_\delta(s,t,\varphi))(x_k,p) = \lambda w_{k,l} + (1-\lambda) w_{k,l+1},
  \end{align}
  where $p \in [p_l, p_{l+1})$ with $0 < l < N^a_\delta - 1$, and
  $\lambda = \frac{p_{l+1}-p}{p_{l+1}-p_l}$. Thus:
  \begin{align}
    \nonumber &S^{a,-}_{\delta,\epsilon}(s,t,\varphi)(x_k,p) - \cI^a_\delta(w)(x_k,p)\\
              &= S^{a,-}_{\delta,\epsilon}(s,t,\varphi)(x_k,p) - \cI^a_\delta(w^-)(x_k,p) + \cI^a_\delta(w^-)(x_k,p) - \cI^a_\delta(w)(x_k,p) \\
    \nonumber &= S^{a,-}_{\delta,\epsilon}(s,t,\varphi)(x_k,p) - \cI^a_\delta(w^-)(x_k,p) + \lambda (w^-_{k,l} - w_{k,l}) + (1-\lambda) (w^-_{k,l+1} - w_{k,l+1}).
  \end{align}
  The two last terms are controlled using \eqref{subsolcontrol}, and, by properties of linear interpolation of the function $p \mapsto S^{a,-}_{\delta,\epsilon}(t^+,t^-,\varphi)(x_k,p) \in \cC^2([0,1],\R)$ with $|\partial^2_{pp}S^{a,-}_{\delta,\epsilon}(t^+,t^-,\varphi)|_\infty \le \frac{C_{\varphi}(h)}{\epsilon^2}$ (recall the previous Lemma)
  the first term is of order $\frac{\delta^2}{\epsilon^{2}} = o(h)$
  since \eqref{CFL3} is in force and
  $\frac{\delta}{\epsilon^2} \to 0$.
  \eproof
\end{proof}
\begin{Lemma} \label{term phi}
  For $0 \le t < s \le T$ such that $L(s-t) \le 1, \xi > 0, \varphi : \delta\Z \times [0,1] \to \R$ a bounded function, the following holds for all $a \in K$:
  \begin{align*}
    S^a_\delta(s,t,\varphi) + \xi - L(s-t)\xi \le S^a_\delta(s,t,\varphi+\xi) \le S^a_\delta(s,t,\varphi) + \xi,
  \end{align*}
  where $L$ is the Lipschitz constant of $f$.
\end{Lemma}
\begin{proof}
  Let $v = S^a_\delta(s,t,\varphi), w = v+\xi-L(s-t)\xi$. Since $v$ satisfies \eqref{systu}, we have, for $k \in \Z$ and $0 < l < N^a_\delta$,
  \begin{align*}
    &S(k,l,w_{k,l},\nabla^a_\delta w_{k,l},\nabla^a_{+,\delta}w_{k,l},\Delta^a_\delta w_{k,l},\varphi+\xi) = -L(s-t)\xi \\ &+ (s-t)\left(f(t,x_k,v_{k,l},\nabla^a_\delta v_{k,l})-f(t,x_k,w_{k,l},\nabla^a_\delta v_{k,l})\right).
  \end{align*}
  Since $f$ is non-increasing in its third variable and Lipschitz continuous, we get:
  \begin{align*}
    S(k,l,w_{k,l},\nabla^a_\delta w_{k,l},\nabla^a_{+,\delta}w_{k,l},\Delta^a_\delta w_{k,l},\varphi+\xi) \le 0.
  \end{align*}
  The same computation with $l=0$ or $l=N^a_\delta$ and $S_b$ instead of $S$ gives
  \begin{align*}
    S_b(k,l,w_{k,l},\nabla_\delta w_{k,l},\nabla_{+,\delta}w_{k,l},\Delta_{\delta}w_{k,l}, \varphi+\xi) \le 0,
  \end{align*}
  and the comparison theorem given in Proposition \ref{comp thm bound} gives $w_{k,l}\le S^a_\delta(s,t,\varphi+\xi)_{k,l}$ for $k \in \Z$ and $l \in \{0, N^a_\delta\}$.\\
  The comparison result from Proposition \ref{prop comp thm scheme} gives the first inequality of the lemma. The second one is proved similarly.
  \eproof
\end{proof}
\begin{Proposition}[Consistency] \label{consistent scheme}
  Let $\varphi \in \cC^\infty_b([0,T]\times\R\times\R,\R), (t,x,p) \in [0,T)\times\R\times(0,1)$. We have, with the notation in \eqref{pde num}:
  \begin{align*}
    &\left| \frac{1}{t_{j+1}-t_j}\widehat{\mathfrak{S}}(\pi,\delta,t_j,x_k,q,\varphi(t_j,x_k,q) + \xi,\varphi+\xi) \right. \\
    &\hspace{1cm} \left. - \sup_{a \in K} \left[- D^a\varphi(t,x,p) + F(t,x,\varphi(t,x,p),\nabla^a\varphi(t,x,p),\Delta^a\varphi(t,x,p)) \right] \right| \to 0,
  \end{align*}
  as $\delta, |\pi| \to 0$ satisfying 
  \eqref{CFL1}--\eqref{CFL3}, 
  $\pi\times\delta\Z\times[0,1] \ni (t_j,x_k,q) \to (t,x,p), \xi \to 0$.
\end{Proposition}
\begin{proof}
  Let $\varphi,j,k,p,l$ as in the statement of the Proposition.\\
  Without loss of generality, we can consider $\pi,\delta,t_j,x_k,q$ such that, for all $a \in K$:
  \begin{align*}
    0 \le \mathfrak p^a(q) \le 1,
  \end{align*}
  Since $\varphi$ is smooth and $(t_k,x_j,p_l) \to (t,x,p)$, we have
  \begin{align*}|&\sup_{a \in K} \left[ - D^a\varphi(t,x,p) + F(t,x,p,\varphi(t,x,p),\nabla^a\varphi(t,x,p),\Delta^a\varphi(t,x,p))\right]\\
                                &- \sup_{a \in K}\left[- D^a\varphi(t_j,x_k,p_l) + F(t_j,x_k,p_l,\varphi(t_j,x_k,p_l),\nabla^a\varphi(t_j,x_k,p_l),\Delta^a\varphi(t_j,x_k,p_l))) \right] | \to 0.
  \end{align*}
  Thanks to Lemma \ref{term phi}, it suffices to prove:
  \begin{align*}
    &\left| \frac{1}{t_{j+1}-t_j}\widehat{\mathfrak{S}}(\pi,\delta,t_j,x_k,q,\varphi(t_j,x_k,q),\varphi) \right. \\
    &\hspace{0.5cm} \left. - \sup_{a \in K} \left( -D^a \varphi(t_j,x_k,q) + F(t_j,x_k,\varphi(t_j,x_k,q),\nabla^a\varphi(t_j,x_k,q),\Delta^a\varphi(t_j,x_k,q))\right) \right| \to 0,
  \end{align*} as $|\pi| \to 0$ and $\pi\times\delta\Z\times(0,1) \ni (t_j,x_k,q) \to (t,x,p)$.\\
  Let $\epsilon > 0$ such that $\epsilon \to 0$ and $\frac{\delta}{\epsilon^2} \to 0$.\\
  Using $|\inf - \sup| \le \sup | \cdot - \cdot |$, adding
  \begin{align*}
    \pm &\left(\frac{1}{t_{j+1}-t_j}\varphi(t_{j+1},x_k,\mathfrak p^a(q))\right.\\ &\left. + F_\epsilon(t_j,x_k,\varphi(t_{j+1},x_k,\mathfrak p^a(q)),\nabla^{\mathfrak{a}(a,\delta)}\varphi(t_{j+1},x_k,\mathfrak p^a(q)),\Delta^{\mathfrak{a}(a,\delta)}\varphi(t_{j+1},x_k,\mathfrak p^a(q)))\right)
  \end{align*}
  and using Lemma \ref{diffDelta}, it is enough to show that, for all $a \in K$,
  \begin{align}\label{last eq}
    &\left|\frac{1}{t_{j+1}-t_j}\left(\varphi(t_{j+1},x_k,\mathfrak p^a(q))-\cI^a_\delta(S^a_\delta(t_{j+1},t_j,\varphi(t_{j+1},\cdot)))(x_k,q)\right) \right. \\
    \nonumber &\left. - F_\epsilon(t_j,x_k,\varphi(t_{j+1},x_k,\mathfrak p^a(q)),\nabla^{\mathfrak{a}(a,\delta)}\varphi(t_{j+1},x_k,\mathfrak p^a(q)),\Delta^{\mathfrak{a}(a,\delta)}\varphi(t_{j+1},x_k,\mathfrak p^a(q)))\right|\to 0.
  \end{align}
  The proof is concluded using the equality $|\cdot| = \max(\cdot,-\cdot)$ and the two following inequalities, obtained by Lemma \ref{lem interp}, and by definition \eqref{supersolinside}-\eqref{subsolinside} of $S^{a,\pm}_{\delta,\epsilon}$:
  \begin{align*}
    &\frac{1}{t_{j+1}-t_j}\left(\varphi(t_{j+1},x_k,\mathfrak p^a(q))-\cI^a_\delta(S^a_\delta(t_{j+1},t_j,\varphi(t_{j+1},\cdot)))(x_k,q)\right)  \\
    & - F_\epsilon(t_j,x_k,\varphi(t_{j+1},x_k,\mathfrak p^a(q)),\nabla^{\mathfrak{a}(a,\delta)}\varphi(t_{j+1},x_k,\mathfrak p^a(q)),\Delta^{\mathfrak{a}(a,\delta)}\varphi(t_{j+1},x_k,\mathfrak p^a(q))) \\
    &\le \frac{1}{t_{j+1}-t_j}(\varphi(t_{j+1},x_k,\mathfrak p^a(q))-S^{a,-}_{\delta,\epsilon}(t_{j+1},t_j,\varphi(t_{j+1},\cdot)))(x_k,q) + o(t_{j+1}-t_j)) \\
    &- F_\epsilon(t_j,x_k,\varphi(t_{j+1},x_k,\mathfrak p^a(q)),\nabla^{\mathfrak{a}(a,\delta)}\varphi(t_{j+1},x_k,\mathfrak p^a(q)),\Delta^{\mathfrak{a}(a,\delta)}\varphi(t_{j+1},x_k,\mathfrak p^a(q))),
  \end{align*}
  and
  \begin{align*}
    &F_\epsilon(t_j,x_k,\varphi(t_{j+1},x_k,\mathfrak p^a(q)),\nabla^{\mathfrak{a}(a,\delta)}\varphi(t_{j+1},x_k,\mathfrak p^a(q)),\Delta^{\mathfrak{a}(a,\delta)}\varphi(t_{j+1},x_k,\mathfrak p^a(q)))\\
    &- \frac{1}{t_{j+1}-t_j}\left(\varphi(t_{j+1},x_k,\mathfrak p^a(q))-\cI^a_\delta(S^a_\delta(t_{j+1},t_j,\varphi(t_{j+1},\cdot)))(x_k,q)\right) \\
    &\le F_\epsilon(t_j,x_k,\varphi(t_{j+1},x_k,\mathfrak p^a(q)),\nabla^{\mathfrak{a}(a,\delta)}\varphi(t_{j+1},x_k,\mathfrak p^a(q)),\Delta^{\mathfrak{a}(a,\delta)}\varphi(t_{j+1},x_k,\mathfrak p^a(q))) \\
    &- \frac{1}{t_{j+1}-t_j}(\varphi(t_{j+1},x_k,\mathfrak p^a(q))-S^{a,+}_{\delta,\epsilon}(t_{j+1},t_j,\varphi(t_{j+1},\cdot)))(x_k,q) + o(t_{j+1}-t_j)).
  \end{align*}
  \eproof
\end{proof}


\section{Numerical studies}
\label{se-numerics}

We now present a numerical application of the previous algorithm.\\

\subsection{Model}
We keep the notation of the previous section: the process $X$ is a Brownian motion with drift. In this numerical example, the drift of the process $Y$ is given by the following functions:
\begin{align*}
  f_1(t,x,y,z) &:= -\sigma^{-1} \mu z, \mbox{ and } \\
  f_2(t,x,y,z) &:= -\sigma^{-1} \mu z + R(y- \sigma^{-1} z)^-,
\end{align*}
where, for $x \in \R$, $x^-=\max(-x,0)$ denotes the negative part of $x$. 
The function $f_1$ corresponds to pricing in a linear complete Black \& Scholes market. It is well-known that there are explicit formulae for the quantile hedging price for a vanilla put or call, see \cite{FL99}.\\
In both cases, we compute the quantile hedging price of a put option with strike $K = 30$ and maturity $T = 1$,
i.e.\ $g(x) = \max(K-\exp(x),0)$.\\
The parameters of $X$ are $\sigma = 0.25$ and $\mu = 0.01875$
(this corresponds to a parameter $b = 0.05$ in the dynamics of the associated geometric Brownian motion, where $\mu = b-\sigma^2/2$).\\
In the rest of this section, we present some numerical experiments.
  First, using the non-linear driver $f_2$, we observe the convergence of $v_{\pi,\delta}$ towards $v_n$ for a fixed discrete control set, and we estimate the rate of convergence with respect to $\delta$. 
  Second, we show that the conditions \eqref{CFL1} to \eqref{CFL3} are not only theoretically important, but also numerically.
  Last, we use the fact that the analytical solution to the quantile hedging problem with driver $f_1$ is known (see \cite{FL99}) to assess the convergence (order) of the scheme more precisely. 
  We observe that a judicious choice of control discretisation, time and space discretisation leads to convergence of $v_{\pi,\delta}$ to $v$. However, the unboundedness of the optimal control as $t \to T$ leads to expensive computations.

The scheme obtained in the previous section deals with an infinite domain in the $x$ variable. In practice, one needs to consider a bounded interval $[B_1,B_2]$ and to add some boundary conditions. Here, we choose $B_1 = \log(10)$ and $B_2 = \log(45)$, and the approximate Dirichlet boundary values for $v(t,B_i,p)$ are the limits $\lim v_{th}(t,x,p)$ as $x \to 0$ or $x \to +\infty$, where $v_{th}$ is the analytical solution obtained in \cite{FL99} for the linear driver $f_1$. Since the non-linearity in $f_2$ is small
 for realistic parameters (we choose $R = 0.05$ in our tests), it is expected that the prices are close (see also \cite{GP15}).
 Furthermore, we will consider values obtained for points $(t,x,p)$ with $x$ far from to the boundary. In this situation, the influence of our choice of boundary condition should be small, as noticed for example in Proposition \ref{diffsupersubsol}. This was studied more systematically, for example, in \cite{BDR95}.

\subsection{Convergence towards $v_n$ with the non-linear driver}

In this section, we consider the non-linear driver $f_2$ defined above, where there is no known analytical expression for the quantile heding price.
 We now fix a discrete control set, and we compute the value function $v_{\pi,\delta}$ for various discretisation parameters $\pi,\delta$ satisfying \eqref{CFL1} to \eqref{CFL3}.
We consider the following control set with $22$ controls:
\begin{align}
\label{contr-ex}
  &\left([-2,2] \cap \frac{\Z - \{0\}}{2}\right) \bigcup \left([-3,3] \cap \frac{\Z - \{0\}}{3}\right)\\ &= \left\{-2, -1.5, \dots, 1.5, 2\right\} \cup \left\{-3, -3+\frac13, \dots, 3-\frac13,3\right\},
  \nonumber
\end{align}
and $\delta \in \{0.05, 0.03, 0.005\}$.
For a fixed $\delta$, we set $h = C \delta$ with $C:= \min(1, 2\frac{\theta}{L}, \frac{1}{|\sigma^2-\mu|})$, $\theta = \frac15$ and $L = |\mu| + R$, so that \eqref{CFL1} to \eqref{CFL3} are satisfied.\\
We get the graphs shown in Figure \ref{fig:1} for the function $p \mapsto v_{\pi,\delta}(t,x,p)$, where $(t,x) = (0,30), (0,37)$.

\begin{figure}[!h]
  \centering
  \includegraphics[scale=0.3]{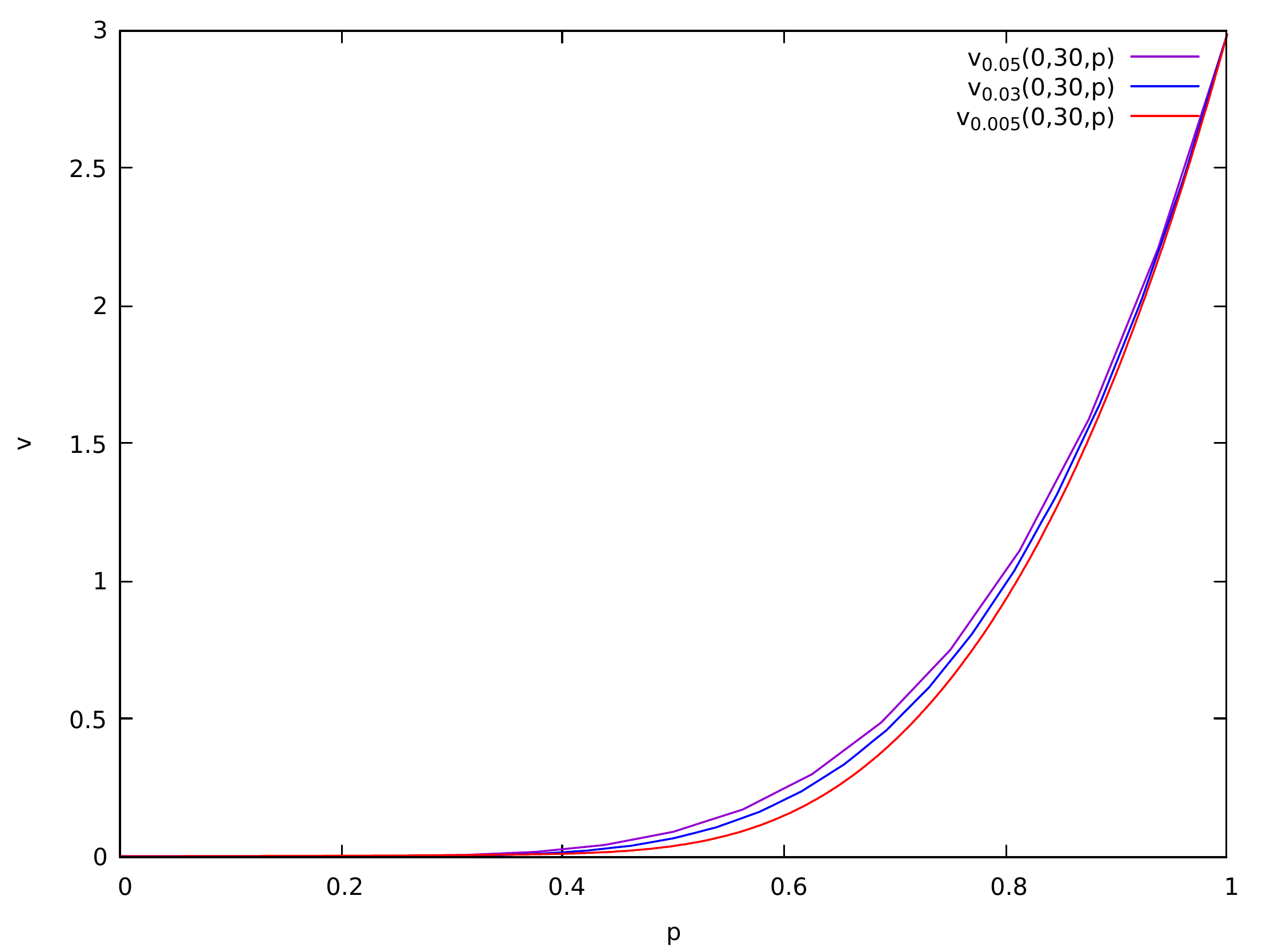}
  \includegraphics[scale=0.3]{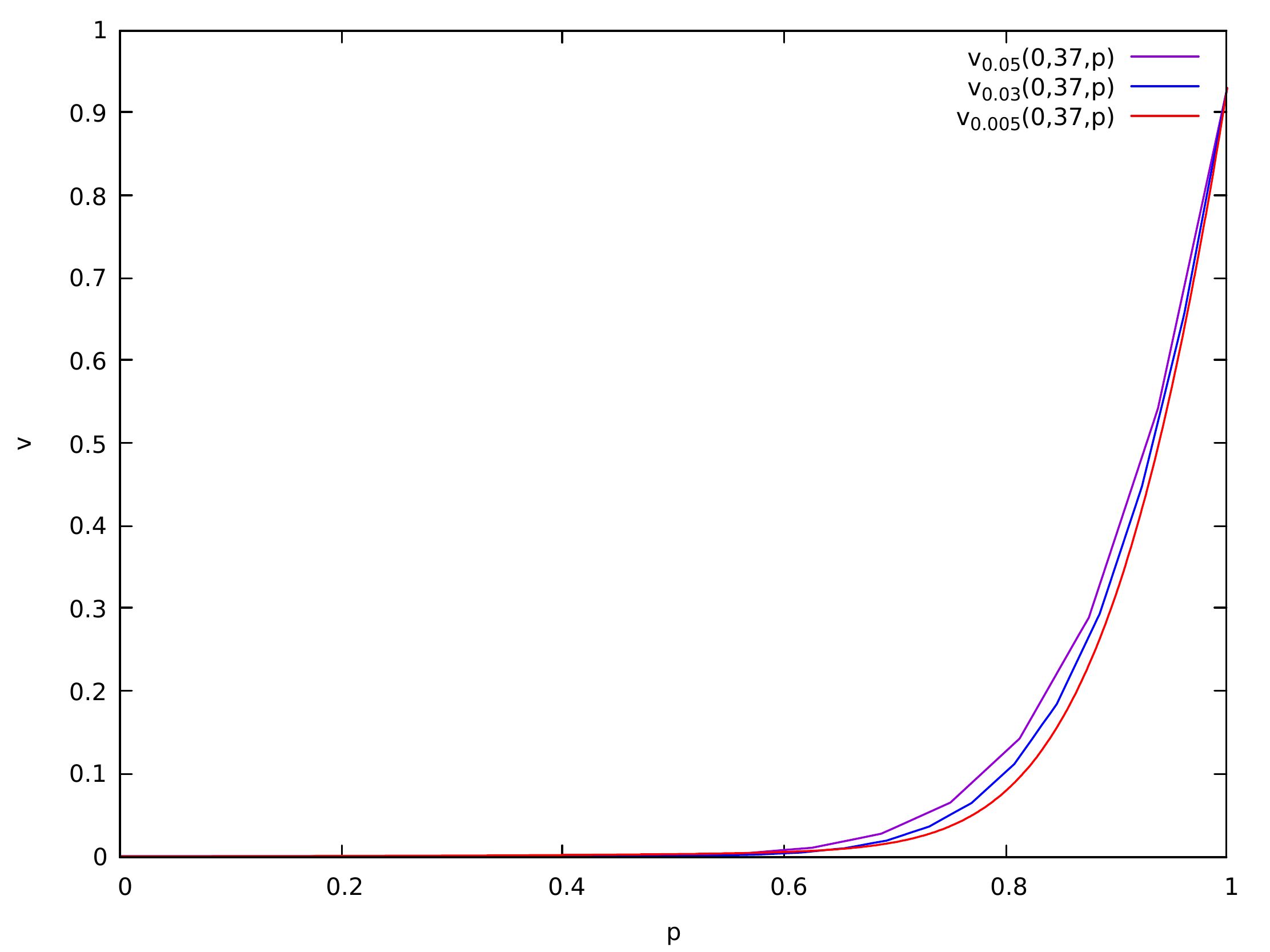}
  \caption{The functions $v_{\pi,\delta}(t,x,\cdot)$, $t = 0$ and $x \in \{30,37\}$.}
  \label{fig:1}
\end{figure}

We observe, while not proved, that the numerical approximation always gives an upper bound for $v_n$, which is itself greater than the quantile hedging price $v$. This is a practically useful feature of this numerical method.

The scheme preserves a key feature of the exact solution, namely that the quantile hedging price is $0$ exactly for $p$ below a certain threshold, depending on $t,x$. This is a consequence of the diffusion stencil $\Delta^a_\delta$ respecting the degeneracy of the diffusion operator $\Delta ^a$ in \eqref{diffOp}, which acts only in direction $(1,p)$, and by the specific construction of the meshes.

In Table \ref{tab:1}, we show some discretisation parameters obtained by this construction with selected values of $\delta$. Here, $N_x$ is the number of points for the $x$-variable, $N_c$ the number of controls, and $N_p$ the total number of points for the $p$ variable (i.e., for all meshes combined).
Moreover, $a_{\max}$ (resp.\ $a_{\min}$) is the greatest (resp.\ smallest) control obtained,
using the modification of the control set (\ref{contr-ex}) as described in Section \ref{sub-se de finite diff scheme}.
Also recall that different meshes are applied in each step of the PCPT algorithm for different $a$, hence we also report
$N_p(a_{\max})$ (resp.\ $N_p(a_{\min})$), the number of points for the $p$ variable for the control $a_{\max}$ (resp.\ $a_{\min}$).
With our choice of parameters, we have $h = \delta$, so the number of time steps is always $\frac 1 \delta$.

\begin{table}[!h]
  \begin{center}
    \begin{tabular}{|c|c|c|c|c|c|c|c|c|c|} 
      \hline
      $\delta$ & $N_t$ &$N_x$ & $N_c$ & $N_p$ & $a_{\max}$ & $N_p(a_{\max})$ & $a_{\min}$ & $N_p(a_{\min})$ & time (in seconds)\\
      \hline
      $0.1$ & $10$ & $18$ & $12$ & $15$ & $2.5$ & $2$ & $0.3125$ & $9$ & $0.04$ \\
      \hline
      $0.05$ & $20$ & $33$ & $15$ & $38$ & $2.5$ & $3$ & $0.3125$ & $17$ & $0.15$ \\
      \hline
      $0.01$ & $100$ & $153$ & $22$ & $255$ & $2.77$ & $10$ & $0.33$ & $76$ & $166$ \\
      \hline
      $0.005$ & $200$ & $304$ & $22$ & $585$ & $2.94$ & $18$ & $0.33$ & $151$ & $4608$ \\
      \hline
    \end{tabular}
    \caption{Parameters for selected values of $\delta$.}
    \label{tab:1}
  \end{center}
\end{table}

\subsection{CFL conditions}

Using the same discrete control set as above, we now fix $h = 0.1$ and compute $v_{\pi,\delta}$ for $\delta$ chosen as above. The conditions \eqref{CFL1} to \eqref{CFL3} are then not satisfied anymore. \\
First, while $\pi$ is coarse, we observe that the computational time to get $v_{\pi,\delta}(t_j,\cdot)$ from $v_{\pi,\delta}(t_{j+1},\cdot)$ is larger. In fact, since the conditions are not satisfied anymore, the results of Proposition \ref{uniqueness picard scheme} are not valid anymore. While convergence to a fixed point is still observed, many more Picard iterations are needed. For example, for $\delta = 0.005$ and $h = 0.1$, we observe that $3000$ Picard iterations are needed, while in the example where \eqref{CFL1} to \eqref{CFL3} were satisfied, $250$ iterations sufficed to obtain convergence (with a tolerance parameter of $10^{-5}$). \\
The second observation is that, while we observe convergence to some limit (at least with our choice of $\delta$: it might start to diverge for smaller $\delta$, as seen for the case $\delta$ fixed and varying $h$ below), it is not the limit observed in the previous subsection. We show in Figure \ref{fig:2} the difference between the solution obtained with $\delta = 0.005, h = 0.1$, and $\delta = 0.005, h = C\delta$. When the conditions are not met, we are dealing with a non-monotone scheme, and convergence to the unique viscosity solution of the PDE, which equals the value function of the stochastic target problem, is not guaranteed.

Conversely, when $\delta$ is fixed and we vary $h$, the situation is different. There is no issue with the Picard iterations, as the conditions needed for Proposition \ref{uniqueness picard scheme} are still satisfied. The issue here is that the consistency hypothesis is not satisfied, and convergence is not observed: when $h$ is too close to $0$, the value $v_{\pi,\delta}$ gets bigger, as seen in Figure \ref{fig:3}. Here, $\delta$ is fixed to $0.05$ and $h$ goes from $0.025$ to $1.2 \times 10^{-5}$.

\begin{figure}[!h]
  \begin{minipage}{.5\textwidth}
    \includegraphics[scale=0.33]{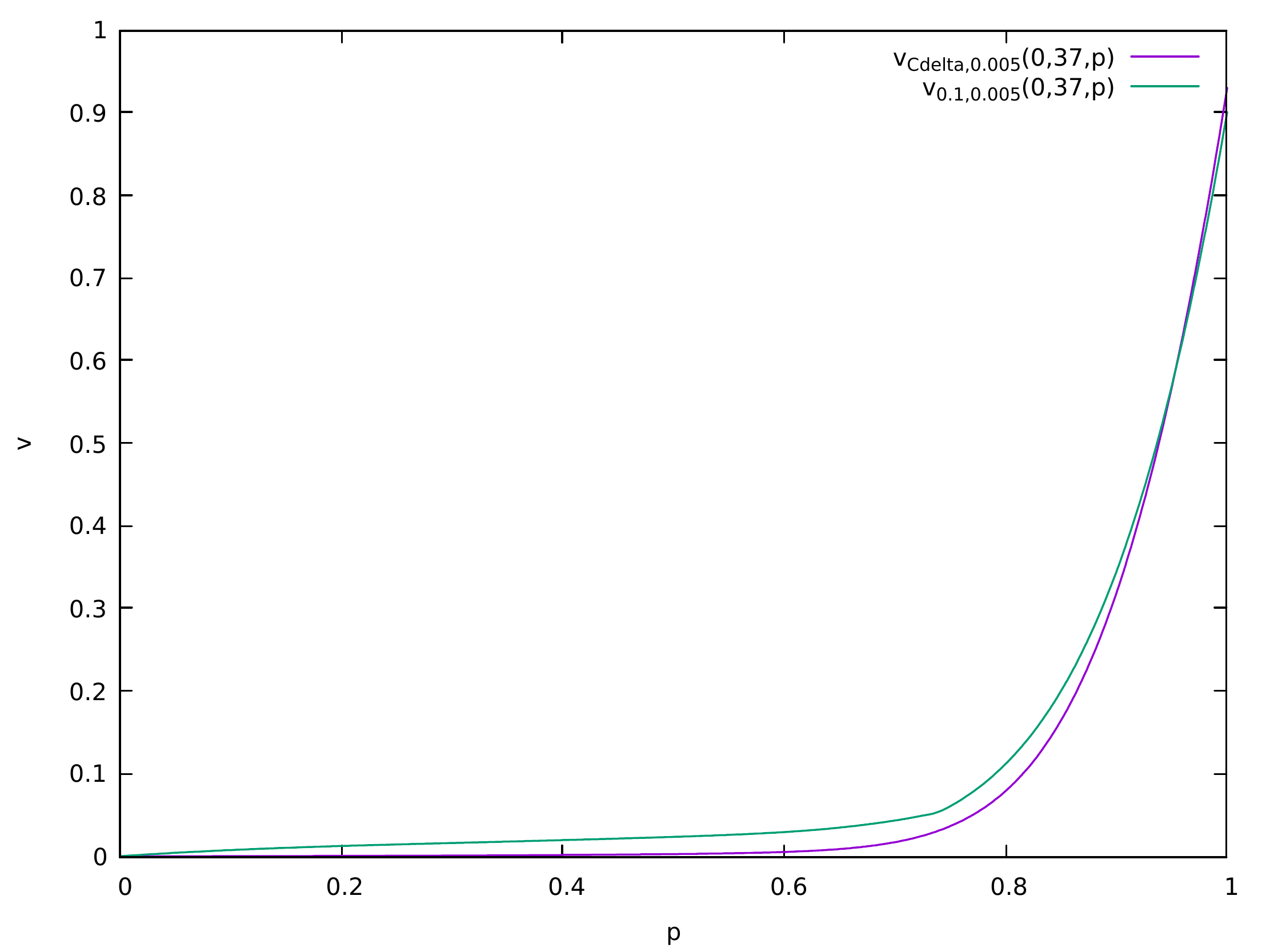}
    \caption{Comparison of $v_{0.1,0.005}(0,37,\cdot)$ and $v_{C\delta,0.005}(0,37,\cdot)$.}
    \label{fig:2}
  \end{minipage}%
  \begin{minipage}{.5\textwidth}
    \includegraphics[scale=0.33]{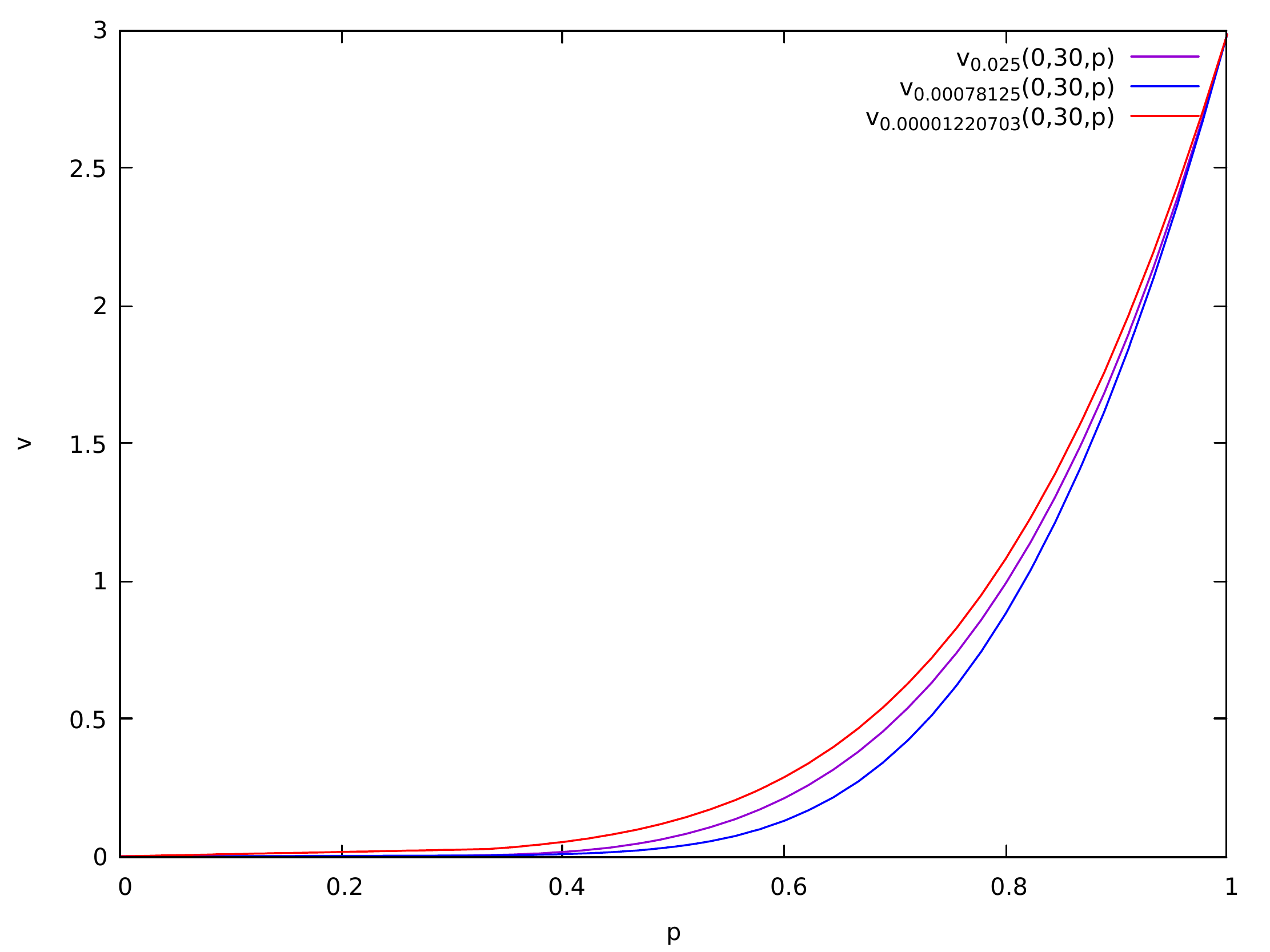}
    \caption{Comparison of $v_{h,0.05}(0,30,\cdot)$ for some $h$.}
    \label{fig:3}
  \end{minipage}
\end{figure}


\subsection{Convergence to the analytical solution with linear driver}

We now consider the linear driver $f_1$. In that case, the quantile hedging price can be found explicitly (see \cite{FL99}). For each $(t,x,p)$, the optimal control can also be computed explicitly:
\begin{align*}
  \alpha(t,x,p) = \frac{1}{\sqrt{2\pi(T-t)}} \exp \left( - \frac{N^{-1}(p)^2}{2} \right),
\end{align*}
where $N$ is the cumulative distribution function of the standard normal distribution.\\
In particular, if the uniform grid $\pi = \{0, h, \dots, \kappa h = T\}$ is fixed, one obtains that the optimal controls are contained in the interval $[0, \frac{1}{\sqrt{2\pi h}}]$.\\
On the other hand, if $\delta$ is fixed, one sees from \eqref{def a a delta} that the greatest control one can reach (with a non-trivial grid for the $p$ variable) is $\frac{\sigma}{2\delta}$.\\
We set our parameters as follows: we first choose $n \ge 2$, we pick $\delta$ such that $\frac{\sigma}{n\delta} \ge \frac{1}{\sqrt{2\pi C\delta}}$, and we set $h = C\delta$. It is easy to see that $\delta$ is proportional to $n^{-2}$.\\
We now pick the controls in $\{ \frac{\sigma}{m\delta}, m \ge n \}$ to obtain $K_n := \{ a_i := \frac{\sigma}{m_i \delta}, i = 1, \dots, M\}$ as follows: 
let $m_1 = n$ so that $a_1 = a^n_{\max} = \frac{\sigma}{n\delta}$. If $m_1,\dots,m_i$ are constructed, we set $m_{i+1} = \inf \{ m \ge m_i,  \frac{\sigma}{m_i\delta} - \frac{\sigma}{m\delta} \ge \frac{1}{n} \}$ and $a_{i+1} = \frac{\sigma}{m_{i+1}\delta}$. If $m_{i+1} < n^{-1}$, then we set $M = i+1$ and we are done.\\
In Figure \ref{fig:4}, we observe convergence towards the quantile hedging price. 
Moreover, Figure \ref{fig:5} demonstrates that the pointwise error, here for $(t,x,p) = (0,30,0.8)$,
has a convergence rate of about $1.3$ with respect to $n$ in the construction described previously.
Last, in Table \ref{tab:2}, we report the values of $\delta$ and $a_{\max}$ obtained for different choices of $n$.

\begin{figure}[!h]
  \begin{minipage}{.5\textwidth}
    \includegraphics[scale=0.33]{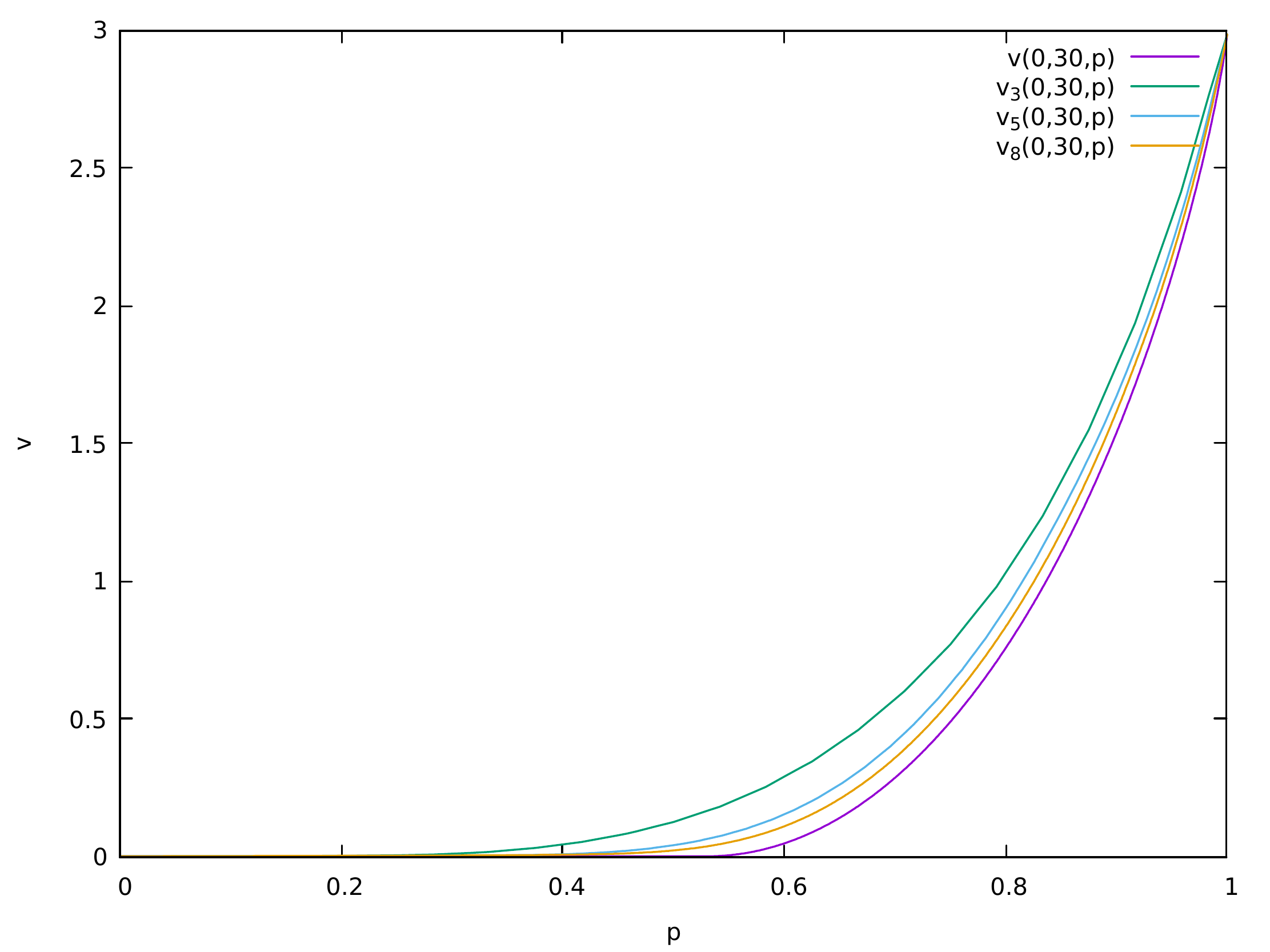}
    \caption{$v_n(0,30,\cdot)$ and $v(0,30,\cdot)$ \\ for $n = 3, 5, 8$.}
    \label{fig:4}
  \end{minipage}%
  \begin{minipage}{.5\textwidth}
    \includegraphics[scale=0.33]{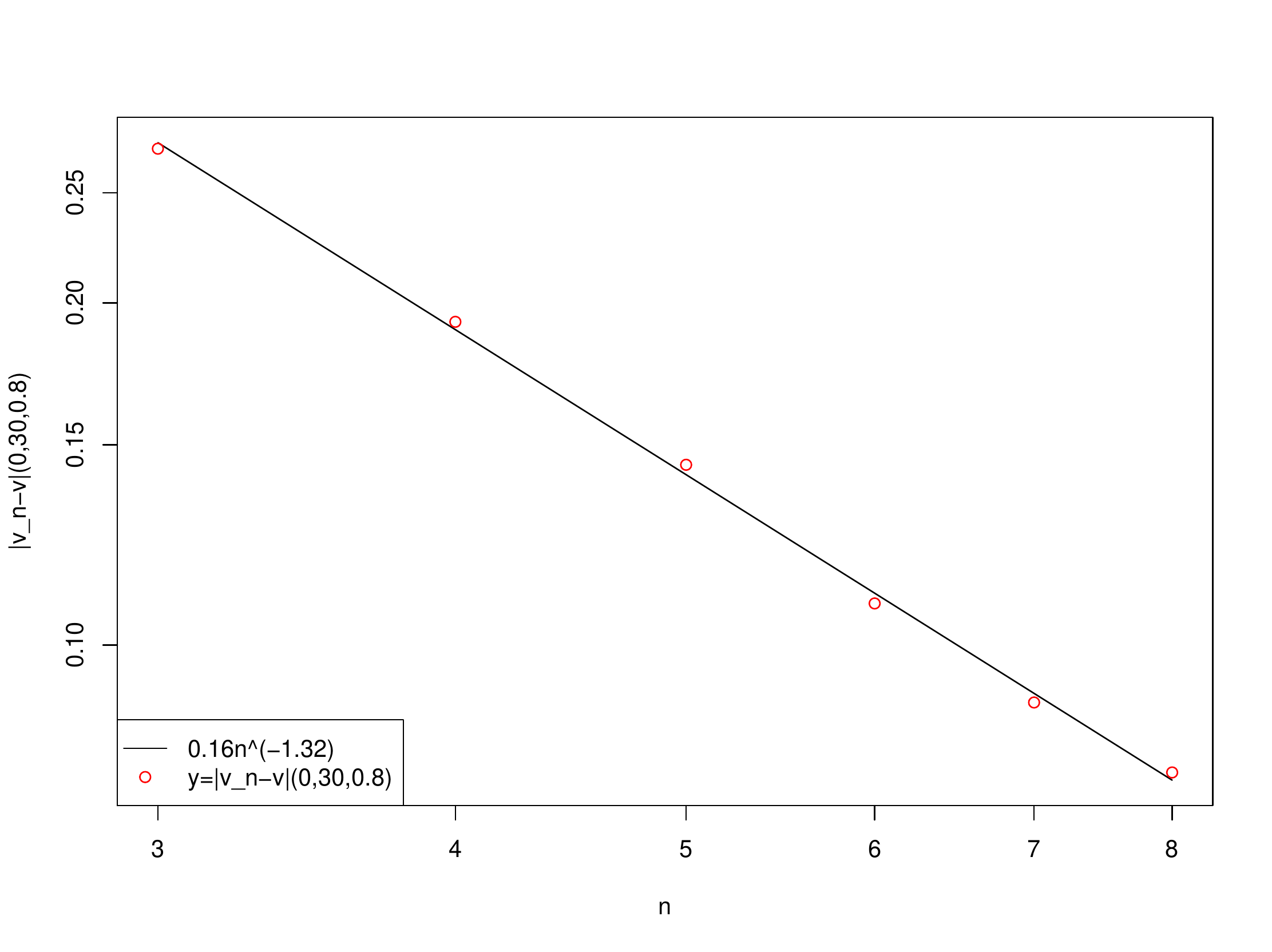}
    \caption{Convergence rate estimation of $v_n(0,30,0.8)$ to $v(0,30,0.8)$ and log-log plot.}
    \label{fig:5}
  \end{minipage}
\end{figure}



\begin{table}[!h]
  \begin{center}
    \begin{tabular}{|c|c|c|c|c|c|} 
      \hline
      $n$ & $\delta$ & $a_{\max}$ & $N_x$ & $N_c$ & $N_p$\\
      \hline
      $3$ & $0.04$ & $1.91$ & $37$ & $5$ & $33$\\
      \hline
      $5$ & $0.01$ & $3.18$ & $97$ & $12$ & $244$\\
      \hline
      $7$ & $0.006$ & $5.09$ & $248$ & $26$ & $1138$\\
      \hline
    \end{tabular}
    \caption{Discretisation parameters for selected values of $n$.}
    \label{tab:2}
  \end{center}
\end{table}

\section{Conclusions}

We have introduced semi-discrete and discrete schemes for the quantile hedging problem, proven their convergence, and illustrated their behaviour in a numerical test.

The scheme, based on piecewise constant policy time-stepping, has the attractive feature that semi-linear PDEs for individual controls can be solved independently on adapted meshes. In the example of the Black-Scholes dynamics this had the effect that in spite of the degeneracy of the diffusion operator it was possible to construct on each mesh a local scheme, i.e.\ one where only neighbouring points are involved in the discretisation. This does not contradict known results on the necessity of non-local stencils for monotone consistent schemes in this degenerate situation (see e.g.\ \cite{reisinger2018non}), because of the superposition of different highly anisotropic meshes to arrive at a scheme which is consistent overall.

A more accurate scheme could be constructed by exploiting higher order, limited interpolation in the $p$-variable, such as in \cite{RF16}. It should be possible to deduce  convergence from the results of this paper and the properties of the interpolator using the techniques in \cite{warin2016some}.
  


\section{Appendix}
\label{se-app}

\subsection{Proofs}

\begin{proof}[Proposition \ref{diffsupersubsol}]
  For ease of notation, we set,
  \begin{align*}
    e(l) &:= e^{-4\frac{\mathfrak{a}(a,\delta)^2}{\sigma^2}C(h,\delta)l(N^a-l)},\\ 
    e^\star &:= \min_{x \in [0,N^a_\delta]} e(x) = e(N^a_\delta/2)  
    = e^{-\frac{\mathfrak{a}(a,\delta)^2}{\sigma^2}C(h,\delta)(N^a_\delta)^2} = e^{-\frac{C(h,\delta)}{\delta^2}}, \\
    B &:= |(v^1_{\cdot,0} - v^2_{\cdot,0})^+|_\infty + |(v^1_{\cdot,N^a_\delta} - v^2_{\cdot,N^a_\delta})^+|_\infty.
  \end{align*}
  By the comparison theorem, it is enough to show that
  $w \in \ell^\infty(\cG^a_\delta)$ defined by
  \begin{align}
    w_{k,l} := v^2_{k,l} + Be(l)
  \end{align}
  satisfies $w_{k,0} \ge v^1_{k,0}$, $w_{k,N^a_\delta} \ge v^1_{k,N^a_\delta}$ and $S(k,l,w_{k,l},\nabla^a_{\delta}w_{k,l},\nabla^a_{+,\delta}w_{k,l},\Delta^a_{\delta}w_{k,l},u) \ge 0$, for all $k \in \Z$ and $0 < l < N^a_\delta$.\\
  The boundary conditions are easily checked: if $k \in \Z$ and
  $l \in \{0,N_a\}$, we have, since $e(0) = e(N_a) = 1$:
  \begin{align*}
    w_{k,l} = v^2_{k,l} + B \ge v^2_{k,l} + (v^1_{k,l}-v^2_{k,l})^+ \ge v^1_{k,l}.
  \end{align*}
For $k \in \Z$, $1 \le l \le N^a_\delta - 1$, we prove
  $S(k,l,w_{k,l},\nabla^a_{\delta}w_{k,l},\nabla^a_{+,\delta}w_{k,l},\Delta^a_{\delta}w_{k,l},u) \ge 0$. By
  definition \eqref{S}, inserting
  $\pm hf(t^-, e^{k\delta}, v^2_{k,l},
  \frac{1}{2\delta}(w_{k+1,l+\sgn(a)} - w_{k-1,l-\sgn(a)}))$, since
  $S(k,l,v^2_{k,l},\nabla^a_{\delta}v^{2}_{k,l},\nabla^a_{+,\delta}v^{2}_{k,l},\Delta^a_{\delta}v^{2}_{k,l},u) \ge 0$
  and since $f$ is non-increasing with respect to its third variable
  and Lipschitz continuous with respect to its fourth variable, we
  have:
  \begin{align}
    &S(k,l,w_{k,l},\nabla^a_{\delta}w_{k,l},\nabla^a_{+,\delta}w_{k,l},\Delta^a_{\delta}w_{k,l},u) \\
    \nonumber &\ge B\left[\left(1 + \mu\frac{h}{\delta} + \sigma^2\frac{h}{\delta^2}+2\theta\right)e(l) - \left(\mu\frac{h}{\delta} + \frac{\sigma^2}{2}\frac{h}{\delta^2} + \theta\right)e(l+\sgn(a))\right.\\ 
    &\left. - \left(\frac{\sigma^2}{2}\frac{h}{\delta^2} + \theta\right)e(l-\sgn(a)) - \frac{hL}{2\delta}|e(l+\sgn(a))-e(l-\sgn(a))|\right].
  \end{align}
  We have $|e(l+\sgn(a)) - e(l-\sgn(a))| \le 1 - e^\star$, thus:
  \begin{align}
    &S(k,l,w_{k,l},\nabla^a_{\delta}w_{k,l},\nabla^a_{+,\delta}w_{k,l},\Delta^a_{\delta}w_{k,l},u) \\
    \nonumber &\ge B\left[\left(1+  \mu\frac{h}{\delta} + \sigma^2\frac{h}{\delta^2}+2\theta+\frac{hL}{2\delta}\right)e^\star - \mu\frac{h}{\delta} - \sigma^2\frac{h}{\delta^2} - 2\theta - \frac{hL}{2\delta}\right].
  \end{align}
  It is thus enough to have
  \begin{align}
    \left(1+ \mu\frac{h}{\delta} + \sigma^2\frac{h}{\delta^2}+2\theta+\frac{hL}{2\delta}\right)e^\star - \mu\frac{h}{\delta} - \sigma^2\frac{h}{\delta^2} - 2\theta - \frac{hL}{2\delta} \ge 0,
  \end{align}
  and one can easily check that this is the case with our choice of $C(h,\delta)$.\\
  It remains to prove \eqref{minconstC}. Since $\ln(1+x) > x - \frac{x^2}{2}$ for all $x > 0$, we have, by \eqref{CFL3}:
  \begin{align*}
    C(h,\delta) &> \frac{1}{\delta^2}\left(\frac{1}{\mu\frac{h}{\delta}+\sigma^2\frac{h}{\delta^2}+2\theta+\frac{hL}{2\delta}} - \frac12\frac{1}{\left(\mu\frac{h}{\delta}+\sigma^2\frac{h}{\delta^2}+2\theta+\frac{hL}{2\delta}\right)^2}\right) \\
                &= \frac{1}{(\mu + \frac{L}{2}) \delta h + \sigma^2 h + 2 \theta \delta^2} - \frac{1}{2\left( (\mu + \frac{L}{2})h + \sigma^2 \frac{h}{\delta} + 2 \delta \theta\right)^2} \\
                &\ge \frac{1}{\left((\mu+\frac{L}{2})M+2\theta M^2\right)h^2 + \sigma^2h} - \frac{1}{2\left(((1+2\theta)\mu+\frac{L}{2})h+\frac{\sigma^2}{M}\right)^2} \\
                &\ge \frac{1}{\left((\mu+\frac{L}{2})M+2\theta M^2\right)h^2 + \sigma^2h} - \frac{M^2}{2\sigma^4}.
  \end{align*}
\end{proof}

  \begin{proof}[Lemma \ref{lemsupersolum}]
    We show the result for $S^{a,+}_{\delta,\epsilon}$, the proof is similar for $S^{a,-}_{\delta,\epsilon}$.\\
    For $k \in \Z$ and $0 \le l \le N^a_\delta$, let
    \begin{align}
      \label{z}
      z_{k,l} := \varphi(x_k,\textcolor{black}{\mathfrak p^a(p_l)}) - h F_\epsilon(t,x_k,\varphi(x,\textcolor{black}{\mathfrak p^a(p_l)}),\nabla^{\mathfrak{a}(a,\delta)}\varphi(x_k,\textcolor{black}{\mathfrak p^a(p_l)}),\Delta^{\mathfrak{a}(a,\delta)}\varphi(x_k,\textcolor{black}{\mathfrak p^a(p_l)})).
    \end{align}
    It is enough to show that, for all $k \in \Z$ and $0 < l < N^a_\delta$,
    \begin{align*}
      S(k,l,z_{k,l},\nabla^a_{\delta}z_{k,l},\nabla^a_{+,\delta}z_{k,l},\Delta^a_{\delta}z_{k,l}, \varphi) \ge - C_{\varphi,n}(h,\epsilon).
    \end{align*}
    Then, since $f$ is non-increasing in its third variable, it is then easy to show that $w^+ = z+C_{\varphi,n}(h,\epsilon)$ satisfies \eqref{supersolinside}.\\
    Let $k \in \Z$ and $1 \le l \le N_a - 1$. We have, by definition
    \eqref{S}:
    \begin{align*}
      &S(k,l,z_{k,l},\nabla^a_{\delta}z_{k,l},\nabla^a_{+,\delta}z_{k,l},\Delta^a_{\delta}z_{k,l},\varphi)\\
      &= h\left( \widehat F(t,x_k,z_{k,l},\nabla^a_{\delta}z_{k,l},\nabla^a_{+,\delta}z_{k,l},\Delta^a_{\delta}z_{k,l})\right.\\
      & \left. -F_\epsilon(t,x_k,\varphi(x_k,\textcolor{black}{\mathfrak p^a(p_l)}),\nabla^{\mathfrak{a}(a,\delta)}\varphi(x_k,\textcolor{black}{\mathfrak p^a(p_l)}),\Delta^{\mathfrak{a}(a,\delta)}\varphi(x_k,\textcolor{black}{\mathfrak p^a(p_l)})) \right),
    \end{align*}
    so it is enough to show
    \begin{align*}
      \widehat F(t,x_k,&z_{k,l},\nabla^a_{\delta}z_{k,l},\nabla^a_{+,\delta}z_{k,l},\Delta^a_{\delta}z_{k,l}) \\ &- F_\epsilon(t,x_k,\varphi(x_k,\textcolor{black}{\mathfrak p^a(p_l)}),\nabla^{\mathfrak{a}(a,\delta)}\varphi(x_k,\textcolor{black}{\mathfrak p^a(p_l)}),\Delta^{\mathfrak{a}(a,\delta)}\varphi(x_k,\textcolor{black}{\mathfrak p^a(p_l)})) \ge - \frac{C_{\phi,n}(h,\epsilon)}{h}.
    \end{align*}
    We split the sum into three terms:
    \begin{align*}
      A &= \widehat F(t,x_k,z_{k,l},\nabla^a_{\delta}z_{k,l},\nabla^a_{+,\delta}z_{k,l},\Delta^a_{\delta}z_{k,l}) \\
        &\hspace{0.5cm}- \widehat F(t,x_k,\varphi(x_k,\textcolor{black}{\mathfrak p^a(p_l)}),\nabla^a_{\delta}\varphi(x_k,\textcolor{black}{\mathfrak p^a(p_l)}),\nabla^a_{+,\delta}\varphi(x_k,\textcolor{black}{\mathfrak p^a(p_l)}),\Delta^a_{\delta}\varphi(x_k,\textcolor{black}{\mathfrak p^a(p_l)})), \\
      B &= \widehat F(t,x_k,\varphi(x_k,\textcolor{black}{\mathfrak p^a(p_l)}),\nabla^a_{\delta}\varphi(x_k,\textcolor{black}{\mathfrak p^a(p_l)}),\nabla^a_{+,\delta}\varphi(x_k,\textcolor{black}{\mathfrak p^a(p_l)}),\Delta^a_{\delta}\varphi(x_k,\textcolor{black}{\mathfrak p^a(p_l)}))\\ &\hspace{0.5cm}- F(t,x_k,\varphi(x_k,\textcolor{black}{\mathfrak p^a(p_l)}),\nabla^{\mathfrak{a}(a,\delta)}\varphi(x_k,\textcolor{black}{\mathfrak p^a(p_l)}),\Delta^{\mathfrak{a}(a,\delta)}\varphi(x_k,\textcolor{black}{\mathfrak p^a(p_l)})),\\
      C &= F(t,x_k,\varphi(x_k,\textcolor{black}{\mathfrak p^a(p_l)}),\nabla^{\mathfrak{a}(a,\delta)}\varphi(x_k,\textcolor{black}{\mathfrak p^a(p_l)}),\Delta^{\mathfrak{a}(a,\delta)}\varphi(x_k,\textcolor{black}{\mathfrak p^a(p_l)}))\\ &\hspace{0.5cm} - F_\epsilon(t,x_k,\varphi(x_k,\textcolor{black}{\mathfrak p^a(p_l)}),\nabla^{\mathfrak{a}(a,\delta)}\varphi(x_k,\textcolor{black}{\mathfrak p^a(p_l)}),\Delta^{\mathfrak{a}(a,\delta)}\varphi(x_k,\textcolor{black}{\mathfrak p^a(p_l)})).
    \end{align*}
    First, we have
    \begin{align*}
      C &= f_\epsilon(t,x_k,\varphi(x_k,\textcolor{black}{\mathfrak p^a(p_l)}),\nabla^{\mathfrak{a}(a,\delta)}\varphi(x_k,\textcolor{black}{\mathfrak p^a(p_l)})) - f(t,x_k,\varphi(x_k,\textcolor{black}{\mathfrak p^a(p_l)}),\nabla^{\mathfrak{a}(a,\delta)}\varphi(x_k,\textcolor{black}{\mathfrak p^a(p_l)})) \\
      &\ge -|f_\epsilon-f|_\infty.
    \end{align*}
    Secondly, by \eqref{def H}-\eqref{def H hat}, we have,
    \begin{align*}
      B = &- \theta \frac{\delta^2}{h} \Delta^a_\delta\varphi(x_k,\textcolor{black}{\mathfrak p^a(p_l)}) \\
          &+ \mu( \nabla^{\mathfrak{a}(a,\delta)}\varphi(x_k,\textcolor{black}{\mathfrak p^a(p_l)}) - \nabla^a_{+,\delta}\varphi(x_k,\textcolor{black}{\mathfrak p^a(p_l)}) ) + \frac{\sigma^2}{2}( \Delta^{\mathfrak{a}(a,\delta)}\varphi(x_k,\textcolor{black}{\mathfrak p^a(p_l)}) - \Delta^a_{\delta}\varphi(x_k,\textcolor{black}{\mathfrak p^a(p_l)}) )\\
          &+ ( f(t,x_k,\varphi(x_k,\textcolor{black}{\mathfrak p^a(p_l)}),\sigma\nabla^{\mathfrak{a}(a,\delta)}\varphi(x_k,\textcolor{black}{\mathfrak p^a(p_l)})) - f(t,x_k,\varphi(x_k,\textcolor{black}{\mathfrak p^a(p_l)}),\sigma\nabla^a_{\delta}\varphi(x_k,\textcolor{black}{\mathfrak p^a(p_l)})) )\\
      \ge &- \theta \frac{\delta^2}{h} |\Delta^a_\delta\varphi(x_k,\textcolor{black}{\mathfrak p^a(p_l)})|\\ &-\mu| \nabla^{\mathfrak{a}(a,\delta)}\varphi(x_k,\textcolor{black}{\mathfrak p^a(p_l)}) - \nabla^a_{+,\delta}\varphi(x_k,\textcolor{black}{\mathfrak p^a(p_l)}) ) - \frac{\sigma^2}{2}| \Delta^{\mathfrak{a}(a,\delta)}\varphi(x_k,\textcolor{black}{\mathfrak p^a(p_l)}) - \Delta^a_{\delta}\varphi(x_k,\textcolor{black}{\mathfrak p^a(p_l)}) |\\
          &- \sigma L |\nabla^{\mathfrak{a}(a,\delta)}\varphi(x_k,\textcolor{black}{\mathfrak p^a(p_l)}) - \nabla^a_{\delta}\varphi(x_k,\textcolor{black}{\mathfrak p^a(p_l)})| 
    \end{align*}
    The first term goes to $0$ since $\frac{\delta^2}{h} \to 0$ as $h \to 0$ and $\Delta^a_\delta\varphi(x_k,\textcolor{black}{\mathfrak p^a(p_l)})$ is bounded. The last three terms go to $0$ by Taylor expansion and Lemma \ref{diffDelta}, since $\varphi$ is smooth.\\
    Finally, by \eqref{def H}-\eqref{def H hat}, using the linearity of the discrete differential operators and \eqref{z}, and since $f$ is Lipschitz-continuous, we have,
    \begin{align*}
      A \ge &-h\mu|\nabla^a_{+,\delta}F_\epsilon(t,x_k,\varphi(x_k,\textcolor{black}{\mathfrak p^a(p_l)}),\nabla^{\mathfrak{a}(a,\delta)}\varphi(x_k,\textcolor{black}{\mathfrak p^a(p_l)}),\Delta^{\mathfrak{a}(a,\delta)}\varphi(x_k,\textcolor{black}{\mathfrak p^a(p_l)}))|\\ &- h(\frac{\sigma^2}{2} + \theta \frac{\delta^2}{h})|\Delta^a_{\delta}F_\epsilon(t,x_k,\varphi(x_k,\textcolor{black}{\mathfrak p^a(p_l)}),\nabla^{\mathfrak{a}(a,\delta)}\varphi(x_k,\textcolor{black}{\mathfrak p^a(p_l)}),\Delta^{\mathfrak{a}(a,\delta)}\varphi(x_k,\textcolor{black}{\mathfrak p^a(p_l)}))| \\ &- Lh|F_\epsilon(t,x_k,\varphi(x_k,\textcolor{black}{\mathfrak p^a(p_l)}),\nabla^{\mathfrak{a}(a,\delta)}\varphi(x_k,\textcolor{black}{\mathfrak p^a(p_l)}),\Delta^{\mathfrak{a}(a,\delta)}\varphi(x_k,\textcolor{black}{\mathfrak p^a(p_l)}))|\\ &- L\sigma h|\nabla^a_{\delta}F_\epsilon(t,x_k,\varphi(x_k,\textcolor{black}{\mathfrak p^a(p_l)}),\nabla^{\mathfrak{a}(a,\delta)}\varphi(x_k,\textcolor{black}{\mathfrak p^a(p_l)}),\Delta^{\mathfrak{a}(a,\delta)}\varphi(x_k,\textcolor{black}{\mathfrak p^a(p_l)}))|.
    \end{align*}
    We can show that each term goes to $0$ as $h \to 0$. By example:
    \begin{align*}
      h\frac{\sigma^2}{2}&|\Delta^a_{\delta}F_\epsilon(t,x_k,\varphi(x_k,\textcolor{black}{\mathfrak p^a(p_l)}),\nabla^{\mathfrak{a}(a,\delta)}\varphi(x_k,\textcolor{black}{\mathfrak p^a(p_l)}),\Delta^{\mathfrak{a}(a,\delta)}\varphi(x_k,\textcolor{black}{\mathfrak p^a(p_l)}))|\\ &\ge h\frac{\sigma^2}{2}\mu|\Delta^a_{\delta}\nabla^{\mathfrak{a}(a,\delta)}\varphi(x_k,\textcolor{black}{\mathfrak p^a(p_l)})|\\ &- h\frac{\sigma^4}{4}|\Delta^a_{\delta}\Delta^{\mathfrak{a}(a,\delta)}\varphi(x_k,\textcolor{black}{\mathfrak p^a(p_l)})|\\ &- h\frac{\sigma^2}{2}|\Delta^a_{\delta}f_\epsilon(t,x_k,\varphi(x_k,\textcolor{black}{\mathfrak p^a(p_l)}),\sigma\nabla^a_{\delta}\varphi(x_k,\textcolor{black}{\mathfrak p^a(p_l)}))|.
    \end{align*}
    The first two terms go to $0$ with $h$ since $|\Delta^a_{\delta}\nabla^{\mathfrak{a}(a,\delta)}\varphi(x_k,\textcolor{black}{\mathfrak p^a(p_l)})|$ and $|\Delta^a_{\delta}\Delta^{\mathfrak{a}(a,\delta)}\varphi(x_k,\textcolor{black}{\mathfrak p^a(p_l)})|$ are bounded, by smoothness of $\varphi$ and by Lemma \ref{diffDelta}.\\
    We can control the derivatives of $\mathfrak f_\epsilon:(x,p) \mapsto f_\epsilon(t,x,\varphi(x,p),\sigma\varphi(x,p))$ with respect to $\epsilon$: for any $\alpha = (\alpha_1,\alpha_2) \in \mathbb N^2$, we have
    \begin{align}
      \label{control deriv epsilon} |D^\alpha \mathfrak f_\epsilon|_\infty \le \frac{C_{\varphi,\alpha}}{\epsilon^{\alpha_1+\alpha_2}},
    \end{align}
    for a constant $C_{\varphi,\alpha} > 0$. \\
    By the triangle inequality and Taylor expansion, we get:
    \begin{align*}
      -h\frac{\sigma^2}{2}&|\Delta^a_{\delta} f_\epsilon(t,x_k,\varphi(x_k,\textcolor{black}{\mathfrak p^a(p_l)}),\sigma\nabla^a_{\delta}\varphi(x_k,\textcolor{black}{\mathfrak p^a(p_l)}))|\\
      &\ge -h\frac{\sigma^2}{2}|(\Delta^a_{\delta}-\Delta^{\mathfrak{a}(a,\delta)})f_\epsilon(t,x_k,\varphi(x_k,\textcolor{black}{\mathfrak p^a(p_l)}),\sigma\nabla^a_{\delta}\varphi(x_k,\textcolor{black}{\mathfrak p^a(p_l)}))| \\
                          &- h\frac{\sigma^2}{2}|\Delta^{\mathfrak{a}(a,\delta)}f_\epsilon(t,x_k,\varphi(x_k,\textcolor{black}{\mathfrak p^a(p_l)}),\sigma\nabla^a_{\delta}\varphi(x_k,\textcolor{black}{\mathfrak p^a(p_l)}))|\\
      &\ge -C_1h\frac{\sigma^2}{2}\frac{\delta^2}{\epsilon^4} - C_2h\frac{\sigma^2}{2}\frac{1}{\epsilon^2},
    \end{align*}
    where $C_1,C_2 > 0$, and this quantity goes to $0$ by our choice of $\epsilon$.\\
    Last, the smoothness of $S^{a,\pm}_{\delta,\epsilon}$ is straightforward by \eqref{def S a pm} and the control on its second derivative with respect to $p$ is obtained by \eqref{control deriv epsilon}. \eproof
  \end{proof}


\subsection{Representation and comparison results}

 For  $ (t,x,y)\in [0,T]\times\R^d\times\R^{+}$,
$q:=\left(\begin{matrix}q^{x}\\q^p\end{matrix}\right)\in\R^{d+1}$ and $A:=\left(\begin{matrix}A^{xx}&A^{xp}\\A^{xp^\top}&A^{pp}\end{matrix}\right)\in\mathbb{S}^{d+1}$, $A^{xx} \in \mathbb{S}^{d}$, denoting $\Xi :=(t,x,y,q,A)$, we define, recalling \eqref{eq de Fa}-\eqref{eq de za}-\eqref{eq de La},
\begin{align*}
\cF(\Xi) = \sup_{a \in \cR^\flat} F^a(\Xi) \text{ with }  F^a(\Xi) := - f(t,x,y,\mathfrak{z}(x,q,a)) - \cL(x,q,A,a)\;,
\end{align*}
where $\cR \subset \cS\setminus\cD$ with a \emph{finite} number of elements.

\begin{Proposition}\label{pr appendix comp princ}
Let $0\le \tau < \theta \le T$ and
 $u_1$ (resp. $u_2$) be a lower semi-continuous super-solution (resp. upper semi-continuous sub-solution)  with polynomial growth, of
\begin{align}\label{eq pde for comparison}
-\partial_t \varphi + \cF(t,x,\varphi,D\varphi, D^2\varphi) = 0\;\text{ on }\;[\tau,\theta)\times\R^d\times (0,1) 
\end{align}
with $u_1 \ge u_2$ on $[\tau,\theta]\times\R^d\times \set{0,1} \bigcup \set{\theta}\times\R^d\times[0,1]$, then
$u_1 \ge u_2$ on $[\tau,\theta]\times\R^d\times [0,1]$.
\end{Proposition}

\begin{Corollary}\label{co comp princ}
There exists a unique continuous solution $w$ to
\eqref{eq pde for comparison} 
or
equivalently
\begin{align} \label{eq equiv pde for comparison}
\sup_{\eta \in \cR} H^\eta(t,x,\varphi,\partial_t \varphi, D\varphi, D^2\varphi) = 0 \;\text{ on }\;[\tau,\theta)\times\R^d\times (0,1) 
\end{align}
satisfying $w(\cdot) = \Psi(\cdot)$ on $ [\tau,\theta]\times\R^d\times\set{0,1}\bigcup\set{\theta}\times\R^d\times[0,1]$, 
where $\Psi \in \cC^0$ 
\end{Corollary}
\begin{proof}
This is a direct application of the comparison principles. The equivalence between
\eqref{eq pde for comparison}  and  \eqref{eq equiv pde for comparison}, comes from the fact
that
$
H^\eta(\Theta)
$
and
$
-b - F^{\eta^\flat}(\Xi)
$
have the same sign. \eproof
\end{proof}

\begin{Lemma}\label{le prob rep wa}
(i) Let $a \in \R^d$ and $w^a$ be the unique solution to  
\begin{align}
-\partial_t \varphi + F^a(t,x,\varphi,D\varphi, D^2\varphi) = 0 \;\text{ on }\;[\tau,\theta)\times\R^d\times (0,1) 
\end{align}
satisfying $w^a(\cdot) = \Psi(\cdot)$ on $ [\tau,\theta]\times\R^d\times\set{0,1}\bigcup\set{\theta}\times\R^d\times[0,1]$, 
where $\Psi \in \cC^0$.
Then it admits the following probabilistic representation:
\begin{align} \label{eq rep bsde basic}
    w^a(t,x,p) = Y_t,
  \end{align}
  where $Y$ is the first component of the solution $(Y,Z)$ to the
  following BSDE with random terminal time
  \begin{align} \label{eq probrepmon} 
  Y_\cdot &= \Psi(\cT, X^{t,x}_\cT,P^{t,p,a}_\cT)  + \int_\cdot^\cT f(s,
    X^{t,x}_s, Y_s, Z_s) \ud s -
    \int_\cdot^\cT Z_s \ud W_s,
  \end{align}
  with
  $\cT := \inf \{s \ge t : P^{t,p,a}_s \in \{0, 1\} \} \wedge \theta$ and
  \begin{align*}
  \begin{cases}
  P^{t,p,a}_\cdot = p + a(W_\cdot - W_t)
  \\
  X^{t,x}_\cdot = x + \int_t^\cdot \mu(X^{t,x}_s) \ud s + \int_t^\cdot \sigma(X^{t,x}_s)\ud W_s \;.
  \end{cases}
  \end{align*}
  (ii) Assume moreover, that $\Psi(T,\cdot) = \phi(\cdot)$ and $\Psi(\cdot,1)=B^1(\cdot,\phi)$,
  $\Psi(\cdot,0)=B^0(\cdot,\phi)$, with the notation of \eqref{eq pde zero control}. Then the solution $(\tilde Y, \tilde Z)$ to
  \begin{align} \label{eq probrepmon bis} 
    Y_\cdot &= \phi(X^{t,x}_\theta,\tilde{P}^{t,p,a}_\theta)   + \int_\cdot^\theta f(s,
    X^{t,x}_s, Y_s, Z_s) \ud s -
    \int_\cdot^\theta Z_s \ud W_s\;,
  \end{align}
  where $\tilde{P}^{t,p,a}_\cdot := {P}^{t,p,a}_{\cdot \wedge \cT}$, satisfies
  \begin{align*}
    Y = \tilde Y \mbox{ on } [t,\cT].
  \end{align*}
\end{Lemma}

\begin{proof}
(i) The probabilistic representation is proved in \cite{DP97}. Note that uniqueness to the PDE comes from the previous lemma in the special case where $\cR$ is reduced to one element.
\\
(ii) Let $A := \{\cT = \theta\}, B := \{\cT < \theta, P^{p,a}_{\cT} = 1\}$, and $C := \{\cT < \theta, P^{p,a}_{\cT} = 0\}$, so that $\Omega = A\cup B\cup C$. For $e \in \set{0,1}$, let $({}^eY^{t,x},{}^eZ^{t,x})$ the solution to
\begin{align*}
  Y_\cdot &= \phi(X^{t,x}_\theta,e) + \int_\cdot^\theta f(s,X^{t,x}_s,Y_s,Z_s) \ud s - \int^\theta_\cdot Z_s \ud W_s\;.
\end{align*}
By \eqref{eq pde zero control}, we have $B^e(\tau,X^{t,x}_\tau,\phi) = {}^e Y_\tau$ for $e \in \set{0,1}$.\\
We introduce the following auxiliary processes, for $s \in [t,\theta]$,
\begin{align*}
  \check Y_s &:= Y_s 1_{t \le s \le \cT} + {}^1 Y_s 1_{s > \cT} 1_B + {}^0 Y_s 1_{s > \cT} 1_C,\\
  \check Z_s &:= Z_t 1_{t \le s \le \cT} + {}^1 Z_t 1_{s > \cT} 1_B + {}^0 Z_t 1_{s > \cT} 1_C.
\end{align*}
First, by construction, we have $Y = \check Y$ on $[t,\cT]$. To prove the proposition it is thus sufficient to show that $\tilde Y = \check Y$ on $[t,\theta]$. To this effect, we show that $(\check Y, \check Z)$ is solution of \eqref{eq probrepmon bis}.\\
We have, for all $s \in [t,\theta]$,
\begin{align*}
  \check Y_s = &\left[\Psi(\cT,X^{t,x}_\cT,P^{t,p,a}_\cT) + \int_s^\cT f(u,X^{t,x}_u,Y_u,Z_u) \ud u - \int_s^\cT Z_u \ud W_u \right] 1_{s \le \cT} \\
               &+ \left[\phi(X^{t,x}_\theta,1) + \int_s^\theta f(u,X^{t,x}_u,{}^1Y_u,{}^1Z_u) \ud u - \int_s^\theta {}^1Z_u \ud W_u \right] 1_{\cT < s}1_B \\
               &+ \left[\phi(X^{t,x}_\theta,0) + \int_s^\theta f(u,X^{t,x}_u,{}^0Y_u,{}^0Z_u) \ud u - \int_s^\theta {}^0Z_u \ud W_u \right] 1_{\cT < s}1_C.
\end{align*}
By our hypotheses and by the definition of ${}^p Y_\cdot$, since $\tilde P^{t,p,a}_\theta = 1$ on $B$ and $\tilde P^{t,p,a}_\theta = 0$ on $C$, we have
\begin{align*}
  \Psi(\cT,X^{t,x}_\cT,P^{t,p,a}_\cT) &= \phi(X^{t,x}_\theta,P^{t,p,a}_\theta) 1_A + B^1(\cT,X^{t,x}_\cT,\phi) 1_B + B^0(\cT,X^{t,x}_\cT,\phi) 1_C \\
                                         &= \phi(X^{t,x}_\theta,P^{t,p,a}_\theta) 1_A + {}^1 Y_\cT 1_B + {}^0 Y_\cT 1_C \\
                                         &= \phi(X^{t,x}_\theta,\tilde P^{t,p,a}_\theta) \\
                                         &\hspace{0.2cm} + \left[ \int_\cT^\theta f(u,X^{t,x}_u,{}^1Y_u,{}^1Z_u)\ud u - \int_\cT^\theta {}^1Z_u \ud W_u \right] 1_B \\
                                         &\hspace{0.2cm} + \left[ \int_\cT^\theta f(u,X^{t,x}_u,{}^0Y_u,{}^0Z_u)\ud u - \int_\cT^\theta {}^0Z_u \ud W_u \right] 1_C.
\end{align*}
Thus, since $A \cap \{\cT < s\} = \emptyset$, we deduce
\begin{align*}
  \check Y_s = &\phi(X^{t,x}_\theta,\tilde P^{t,p,a}_\theta) \\
               &+ \left[ \int_\cT^\theta f(u,X^{t,x}_u,{}^1Y_u,{}^1Z_u)\ud u 1_{s \le \cT} 1_B + \int_\cT^\theta f(u,X^{t,x}_u,{}^0Y_u,{}^0Z_u)\ud u 1_{s \le \cT} 1_C \right] \\
               &+ \left[ \int_s^\cT f(u,X^{t,x}_u,Y_u,Z_u) \ud u 1_{s \le \tau} + \int_s^\cT f(u,X^{t,x}_u,{}^1Y_u,{}^1Z_u) \ud u 1_{\cT < s} 1_B \right. \\
               &\left. \hspace{5.4cm}+ \int_s^\cT f(u,X^{t,x}_u,{}^0Y_u,{}^0Z_u) \ud u 1_{\cT < s} 1_C \right] \\
               &- \left[ \int_\cT^\theta {}^1Z_u \ud W_u 1_{s \le \cT} 1_B + \int_\cT^\theta {}^0 Z_u \ud W_u 1_{s \le \cT} 1_C \right] \\
               &- \left[ \int_s^\cT Z_u \ud W_u 1_{s \le \tau} + \int_s^\cT {}^1 Z_u \ud W_u 1_{\cT < s} 1_B + \int_s^\cT {}^0 Z_u \ud W_u 1_{\cT < s} 1_C \right].
\end{align*}
Now, since $({}^1Y_u,{}^1Z_u) = (\check Y_u,\check Z_u)$ on $(\cT,\theta] \cap B$ and $({}^0Y_u,{}^0Z_u) = (\check Y_u,\check Z_u)$ on $(\cT,\theta] \cap B$, and $\int_\cT^\theta f(u,X^{t,x}_u,\check Y_u,\check Z_u) \ud u 1_A = 0$, we get
\begin{align*}
  \int_\cT^\theta f(u,X^{t,x}_u,{}^1Y_u,{}^1Z_u)&\ud u 1_{s \le \cT} 1_B + \int_\cT^\theta f(u,X^{t,x}_u,{}^0Y_u,{}^0Z_u)\ud u 1_{s \le \cT} 1_C\\ &= \int_\cT^\theta f(u,X^{t,x}_u,\check Y_u,\check Z_u) \ud u 1_{s \le \cT}.
\end{align*}
A similar analysis for the other terms shows that $(\check Y, \check Z)$ is a solution to \eqref{eq probrepmon bis} and concludes the proof. \eproof
\end{proof}


\subsection{Finite differences operator $S_b$}

We recall here the main results about the finite difference approximation defined by the operator $S_b$, see \eqref{Sb}.

\begin{Proposition}[Comparison theorem] \label{comp thm bound}
  Let $0 \le t < s \le T, \delta > 0, h = s-t$ such that \eqref{CFL1}-\eqref{CFL2}-\eqref{CFL3} is satisfied.\\
  Let $(u^1, u^2, v^1, v^2) \in \ell^\infty(\delta\Z)^4$ such that $u^1_k \le u^2_k$ for all $k \in \Z$.
  \begin{enumerate}
  \item For all $k \in \Z, (v, \nabla,\nabla_+,\Delta) \in \R^4$:
    \begin{align*}
      S_b(k,v,\nabla,\nabla_+,\Delta,u^2_k) \le S_b(k,v,\nabla,\nabla_+,\Delta,u^1_k).
    \end{align*}
  \item Assume that, for all $k \in \Z$:
    \begin{align*}
      S_b(k,v^1_k,\nabla_\delta v^1_k,\nabla_{+,\delta}v^1_k,\Delta_\delta v^1_k,u^1_k) \le 0,
      S_b(k,v^2_k,\nabla_\delta v^2_k,\nabla_{+,\delta}v^2_k,\Delta_\delta v^2_k,u^1_k) \ge 0.
    \end{align*}
    Then $v^1_k \le v^2_k$ for all $k \in \Z$.
  \item Assume that, for all $k \in \Z$:
    \begin{align*}
      S_b(k,v^1_k,\nabla_\delta v^1_k,\nabla_{+,\delta}v^1_k,\Delta_\delta v^1_k,u^1_k) = 0,
    S_b(k,v^2_k,\nabla_\delta v^2_k,\nabla_{+,\delta}v^2_k,\Delta_\delta v^2_k,u^2_k) = 0.
    \end{align*}
    Then $v^1_k \le v^2_k$ for all $k \in \Z$.
  \end{enumerate}
\end{Proposition}

\begin{proof} The proof is similar to the proof of Proposition \ref{prop comp thm scheme} and is ommited.\eproof \end{proof}

\begin{Proposition} \label{bound}
  Let $\pi,\delta > 0$ such that \eqref{CFL1}-\eqref{CFL2}-\eqref{CFL3} is satisfied for all $h = t_{j+1}-t_j, j = 0, \dots, \kappa-1$.\\
  Let $V_{\pi,\delta} : \pi \times \delta\Z \to \R$ the solution to:
  \begin{align*}
    v^\kappa_k &= g(x_k), k \in \Z,\\
    S_b(k, V^j_k, \nabla_\delta v^j_k, \nabla_{+,\delta}v^j_k, \Delta_\delta v^j_k, v^{j+1}_k) &= 0, k \in \Z, 0 \le j < \kappa.
  \end{align*}
  For all $k \in \Z$, let $U_{\pi,\delta})_k := \frac{(V_{\pi,\delta})_{k+1}-(V_{\pi,\delta})_{k-1}}{2\delta}$.
  Then:
  \begin{enumerate}
  \item $(V_{\pi,\delta}, U_{\pi,\delta}) \in \ell^\infty(\delta\Z)^2$ and their bound is independant of $\pi, \delta$.
  \item $V_{\pi,\delta}$ converges uniformly on compacts sets to $V$, the super-replication price of the contingent claim with payoff $g$.
  \end{enumerate}
\end{Proposition}

\begin{proof}
    We only show the first point, the second one is obtained by applying the arguments of \cite{BS91}, after proving monotonicity, stability and consistency following the steps of Subsection \ref{subsec proof}.\\
  Since $g$ is bounded, it is easy to show that $V_{\pi,\delta}$ is also bounded independently of $\pi,\delta$, and the proof is similar to the proof of Proposition \ref{stable scheme}.\\
  Since $g$ is Lipschitz-continuous, we get that $U_{\pi,\delta}(T,\cdot)$ is bounded. Using the Lipschitz-continuity of $f$, one deduces easily that $U_{\pi,\delta}$ is a solution of
  \begin{align*}
    u^j_k &- u^{j+1}_{k} - h \left( - \mu \nabla_{+,\delta}u^j_k - (\frac{\sigma^2}{2}+\theta\frac{\delta^2}{h})\Delta_\delta u^j_k - L - L|u^j_k| - L|\nabla_\delta u^j_k| \right) \ge 0, k \in \Z, 0 \le j < \kappa, \\
    u^j_k &- u^{j+1}_{k} - h \left( - \mu \nabla_{+,\delta}u^j_k - (\frac{\sigma^2}{2}+\theta\frac{\delta^2}{h})\Delta_\delta u^j_k + L + L|u^j_k| + L|\nabla_\delta u^j_k| \right) \le 0, k \in \Z, 0 \le j < \kappa, \\
    u^\kappa_k &= \frac{g(x_{k+1}) - g(x_{k-1})}{2\delta} \in [-L,L], k \in \Z.
  \end{align*}
  Again, comparison theorems can be proved, and it is now enough to show that there exists $(\underline u, \overline u) \in \ell^\infty(\pi \times \delta\Z)^2$ which are bounded uniformly in $\pi,\delta$ such that
  \begin{align*}
    \underline u^j_k &- \underline u^{j+1}_{k} - h \left( - \mu \nabla_{+,\delta} \underline u^j_k - (\frac{\sigma^2}{2}+\theta\frac{\delta^2}{h})\Delta_\delta \underline u^j_k - L - L| \underline u^j_k| - L|\nabla_\delta \underline u^j_k| \right) \le 0, k \in \Z, 0 \le j < \kappa, \\
    \overline u^j_k &- \overline u^{j+1}_{k} - h \left( - \mu \nabla_{+,\delta}\overline u^j_k - (\frac{\sigma^2}{2}+\theta\frac{\delta^2}{h})\Delta_\delta\overline u^j_k + L + L|\overline u^j_k| + L|\nabla_\delta\overline u^j_k| \right) \ge 0, k \in \Z, 0 \le j < \kappa, \\
    \underline u^\kappa_k &\le -L, \overline u^\kappa_k \ge L, k \in \Z.
  \end{align*}
  We deal with $\underline u$ only, we obtain similar results for $\overline u$.\\
  One can easily show that $\underline u^j := 1 - (L+1)\prod_{k=j+1}^\kappa \frac{1}{1-h_kL}$, where $h_k := t_k - t_{k-1}$, satisfies the requirements. Furthermore, one gets $\underline u^j \ge \underline u^0 \ge 1 - (L+1)2^{\frac{T}{2L}}$, thus one gets that $u$ is lower bounded by $1 - (L+1)2^{\frac{T}{2L}}$.
  \eproof
\end{proof}


\section*{Acknowledgements}
This work was partially funded in the scope of the research project ``Advanced techniques for non-
linear pricing and risk management of derivatives'' under the aegis of the Europlace Institute of Finance, with
the support of AXA Research Fund.

\newpage

\bibliographystyle{plain}
\bibliography{bib}

\begin{thebibliography}{10}

\bibitem{achdou2013hamilton}
Yves Achdou, Guy Barles, Hitoshi Ishii, and Grigorii~Lazarevich Litvinov.
\newblock {\em Hamilton-Jacobi equations: approximations, numerical analysis
  and applications}, volume~10.
\newblock Springer, 2013.

\bibitem{BDR95}
Guy Barles, Ch~Daher, and Marc Romano.
\newblock Convergence of numerical schemes for parabolic equations arising in
  finance theory.
\newblock {\em Mathematical Models and Methods in Applied Sciences},
  5(1):125--143, 1995.

\bibitem{BJ07}
Guy Barles and Espen Jakobsen.
\newblock Error bounds for monotone approximation schemes for parabolic
  {H}amilton-{J}acobi-{B}ellman equations.
\newblock {\em Mathematics of Computation}, 76(260):1861--1893, 2007.

\bibitem{BS91}
Guy Barles and Panagiotis~E Souganidis.
\newblock Convergence of approximation schemes for fully nonlinear second order
  equations.
\newblock {\em Asymptotic Analysis}, 4(3):271--283, 1991.

\bibitem{BBMZ09}
Olivier Bokanowski, Benjamin Bruder, Stefania Maroso, and Hasnaa Zidani.
\newblock Numerical approximation for a superreplication problem under gamma
  constraints.
\newblock {\em SIAM Journal on Numerical Analysis}, 47(3):2289--2320, 2009.

\bibitem{BBC16}
Bruno Bouchard, G{\'e}raldine Bouveret, and Jean-Fran{\c{c}}ois Chassagneux.
\newblock A backward dual representation for the quantile hedging of {B}ermudan
  options.
\newblock {\em SIAM Journal on Financial Mathematics}, 7(1):215--235, 2016.

\bibitem{BER15}
Bruno Bouchard, Romuald Elie, and Antony R{\'e}veillac.
\newblock {BSDEs} with weak terminal condition.
\newblock {\em The Annals of Probability}, 43(2):572--604, 2015.

\bibitem{bouchard2009stochastic}
Bruno Bouchard, Romuald Elie, and Nizar Touzi.
\newblock Stochastic target problems with controlled loss.
\newblock {\em SIAM Journal on Control and Optimization}, 48(5):3123--3150,
  2009.

\bibitem{bouchard2012stochastic}
Bruno Bouchard and Thanh Nam~Vu.
\newblock A stochastic target approach for {P\&L} matching problems.
\newblock {\em Mathematics of Operations Research}, 37(3):526--558, 2012.

\bibitem{BC17}
G{\'e}raldine Bouveret and Jean-Fran{\c{c}}ois Chassagneux.
\newblock A comparison principle for {PDEs} arising in approximate hedging
  problems: application to {B}ermudan options.
\newblock {\em Applied Mathematics \& Optimization}, pages 1--23, 2017.

\bibitem{chassagneux2018obliquely}
Jean-Fran{\c{c}}ois Chassagneux and Adrien Richou.
\newblock Obliquely reflected backward stochastic differential equations.
\newblock 2018.
\newblock <hal-01761991>.

\bibitem{crandall1992user}
Michael~G. Crandall, Hitoshi Ishii, and Pierre-Louis Lions.
\newblock User’s guide to viscosity solutions of second order partial
  differential equations.
\newblock {\em Bulletin of the American Mathematical Society}, 27(1):1--67,
  1992.

\bibitem{CL84}
Michael~G. Crandall and Pierre-Louis Lions.
\newblock Two approximations of solutions of {Hamilton-Jacobi} equations.
\newblock {\em Mathematics of Computation}, 43:1, 1984.

\bibitem{DP97}
Richard~W.R. Darling and Etienne Pardoux.
\newblock Backwards {SDE} with random terminal time and applications to
  semilinear elliptic {PDE}.
\newblock {\em The Annals of Probability}, 25(3):1135--1159, 1997.

\bibitem{dumitrescu2016bsdes}
Roxana Dumitrescu.
\newblock {BSDEs} with nonlinear weak terminal condition.
\newblock {\em arXiv preprint arXiv:1602.00321}, 2016.

\bibitem{dumitrescu2017bsdes}
Roxana Dumitrescu, Romuald Elie, Wissal Sabbagh, and Chao Zhou.
\newblock {BSDEs} with weak reflections and partial hedging of {American}
  options.
\newblock {\em arXiv preprint arXiv:1708.05957}, 2017.

\bibitem{DRZ18}
Roxana Dumitrescu, Christoph Reisinger, and Yufei Zhang.
\newblock Approximation schemes for mixed optimal stopping and control problems
  with nonlinear expectations and jumps.
\newblock {\em arXiv preprint arXiv:1803.03794}, 2018.

\bibitem{el1997backward}
Nicole El~Karoui, Shige Peng, and Marie-Claire Quenez.
\newblock Backward stochastic differential equations in finance.
\newblock {\em Mathematical finance}, 7(1):1--71, 1997.

\bibitem{FL99}
Hans F{\"o}llmer and Peter Leukert.
\newblock Quantile hedging.
\newblock {\em Finance and Stochastics}, 3(3):251--273, 1999.

\bibitem{GP15}
Emmanuel Gobet and Stefano Pagliarani.
\newblock Analytical approximations of {BSDEs} with nonsmooth driver.
\newblock {\em SIAM Journal on Financial Mathematics}, 6(1):919--958, 2015.

\bibitem{jakobsen2005continuous}
Espen~R. Jakobsen and Kenneth~H. Karlsen.
\newblock Continuous dependence estimates for viscosity solutions of
  integro-{PDEs}.
\newblock {\em Journal of Differential Equations}, 212(2):278--318, 2005.

\bibitem{jakobsen2019improved}
Espen~R. Jakobsen, Athena Picarelli, and Christoph Reisinger.
\newblock Improved order 1/4 convergence for piecewise constant policy
  approximation of stochastic control problems.
\newblock {\em arXiv preprint arXiv:1901.01193}, 2019.

\bibitem{Kry99}
Nikolaj~V. Krylov.
\newblock {A}pproximating value functions for controlled degenerate diffusion
  processes by using piece-wise constant policies.
\newblock {\em Electronic Journal of Probability}, 4(2):1--19, 1999.

\bibitem{K00}
Nikolaj~V. Krylov.
\newblock On the rate of convergence of finite-difference approximations for
  {Bellmans} equations with variable coefficients.
\newblock {\em Probability Theory and Related Fields}, 117(1):1--16, 2000.

\bibitem{moreau2011stochastic}
Ludovic Moreau.
\newblock Stochastic target problems with controlled loss in jump diffusion
  models.
\newblock {\em SIAM Journal on Control and Optimization}, 49(6):2577--2607,
  2011.

\bibitem{reisinger2018non}
Christoph Reisinger.
\newblock The non-locality of {M}arkov chain approximations to two-dimensional
  diffusions.
\newblock {\em Mathematics and Computers in Simulation}, 143:176--185, 2018.

\bibitem{RF16}
Christoph Reisinger and Peter~A. Forsyth.
\newblock Piecewise constant policy approximations to
  {H}amilton-{J}acobi-{B}ellman equations.
\newblock {\em Applied Numerical Mathematics}, 103:27--47, 2016.

\bibitem{warin2016some}
Xavier Warin.
\newblock Some non-monotone schemes for time dependent
  {Hamilton--Jacobi--Bellman} equations in stochastic control.
\newblock {\em Journal of Scientific Computing}, 66(3):1122--1147, 2016.

\end{thebibliography}

\end{document}